\documentclass[pdflatex,sn-basic,Numbered]{sn-jnl}

\usepackage{setspace}
\usepackage{amsthm}%
\usepackage{graphicx}%
\usepackage{multirow}%
\usepackage{amsmath,amssymb,amsfonts}%
\usepackage{mathrsfs}%
\usepackage[title]{appendix}%
\usepackage{xcolor}%
\usepackage{textcomp}%
\usepackage{manyfoot}%
\usepackage{booktabs}%
\usepackage{algorithm}%
\usepackage{algorithmicx}%
\usepackage{algpseudocode}%
\usepackage{listings}%
\usepackage{url}

\usepackage{stmaryrd}
\usepackage{enumerate}
\usepackage{mathtools}
\usepackage{ebproof}
\usepackage{proof}
\usepackage{mathbbol}
\usepackage[inline]{enumitem}
\usepackage{hyperref}
\hypersetup{
    colorlinks,
    linkcolor={black},
    urlcolor={black}
}
\usepackage{tikz-cd}

\newtheorem{theorem}{Theorem}[section]
\newtheorem{proposition}[theorem]{Proposition}
\newtheorem{lemma}[theorem]{Lemma} 

\newtheorem{definition}[theorem]{Definition}
\newtheorem{example}[theorem]{Example}

\AtEndEnvironment{example}{\hfill {\tiny $\blacksquare$}}

\newcommand{\IMLL}{\textrm{IMLL}}
\newcommand{\MILL}{\IMLL}

\renewcommand{\varGamma}{\mathrm{\Gamma}}
\renewcommand{\varDelta}{\mathrm{\Delta}}
\renewcommand{\varSigma}{\mathrm{\Sigma}}
\renewcommand{\varLambda}{\mathrm{\Lambda}}
\renewcommand{\varXi}{\mathrm{\Xi}}
\renewcommand{\varTheta}{\mathrm{\Theta}}

\newcommand{\IPL}{\textrm{IPL}}
\newcommand{\BI}{\textrm{BI}}

\renewcommand{\emptyset}{\varnothing}
\renewcommand{\phi}{\varphi}

\makeatletter
\def\labelandtag#1#2{\begingroup
   \def\@currentlabel{#2}%
   \phantomsection\label{#1}\endgroup
}
\makeatother

\newcommand{\myhypertarget}[2]{%
  \phantomsection
  \hypertarget{#1}{#2}%
  \expandafter\gdef\csname targettext@#1\endcsname{#2}%
}
\newcommand{\myhyperlink}[1]{%
  \hyperlink{#1}{\csname targettext@#1\endcsname}%
}

\newcommand{\base}[1]{\mathscr{#1}}
    \newcommand{\baseB}{\base{B}}
    \newcommand{\baseC}{\base{C}}
    \newcommand{\baseD}{\base{D}}
    \newcommand{\baseE}{\base{E}}
    \newcommand{\baseM}{\base{M}}
    \newcommand{\baseN}{\base{N}}
    \newcommand{\baseO}{\base{O}}
    \newcommand{\baseX}{\base{X}}
    \newcommand{\baseY}{\base{Y}}
\newcommand{\emptybase}{\varnothing}
\newcommand{\at}[1]{\ensuremath{\mathrm{#1}}}
\newcommand{\baseGeq}{\supseteq}

\newcommand{\supp}[2]{\Vdash_{\!\!#2}^{#1}}
    
\newcommand{\entails}{\Vdash}
\newcommand{\proves}[1][]{\vdash_{\!\!#1}}

\newcommand{\system}[1]{\mathsf{#1}}
\newcommand{\Atoms}{\set{A}}
\newcommand{\Formulas}{\set{F}}
\newcommand{\Bunches}{\set{B}}
    \newcommand{\BunchesWithHole}{\dot{\Bunches}}
\newcommand{\set}[1]{\mathbb{#1}}

\newcommand{\biflat}[1]{{#1}^\flat}
\newcommand{\binat}[1]{{#1}^\natural}

\newcommand{\deriveBaseM}[1]{\vdash_{\!\!#1}}

\newcommand{\suppTor}[1]{\Vdash_{#1}}
\newcommand{\suppM}[2]{\Vdash_{#1 }^{ #2 }}

\newcommand{\At}{\mathbb{A}}
\newcommand{\mand}{\mathrel{\ast}}
\newcommand{\mto}{\wand}
\newcommand{\mtop}{\top^{\ast}}

\renewcommand{\atop}{\top}

\newcommand{\addcontext}{\!\mathrel{\fatsemi}\!}
\newcommand{\multcontext}{\!\mathrel{\fatcomma}\!}
    \newcommand{\addcomma}{\addcontext}
    \newcommand{\multcomma}{\multcontext}
    \newcommand{\acomma}{\addcontext}
    \newcommand{\mcomma}{\multcontext}
    \newcommand{\aacomma}{\mathrel{\fatsemi}}
    \newcommand{\mmcomma}{\mathrel{\fatcomma}}
\newcommand{\commavar}{\star}

\newcommand{\holeDepth}[1]{\mathrm{holeDepth}{(#1)}}
\newcommand{\nestDepth}[2]{\mathrm{nestDepth}{(#1)}(#2)}

\newcommand{\seq}{\triangleright}

\DeclareSymbolFont{bbsymbol}{U}{bbold}{m}{n}
\DeclareMathSymbol{\fatsemi}{\mathbin}{bbsymbol}{"3B}
\DeclareMathSymbol{\fatcomma}{\mathbin}{bbsymbol}{"2C}

\newcommand{\buncheq}{\equiv}
\newcommand{\bunchWeakerThan}{\preceq}
    \newcommand{\bunchStrongerThan}{\succeq}

\newcommand{\wand}{\mathrel{-\mkern-12mu-\mkern-12mu\ast}}
\newcommand{\munit}{\varnothing_{\!\times}}
\newcommand{\aunit}{\varnothing_{\!+}}
\newcommand{\I}{\top^{\ast}}
\newcommand{\empset}{\varnothing}

\newcommand{\weak}{\rn{w}}
\newcommand{\cont}{\rn{c}}
\newcommand{\exch}{\rn{e}}
\newcommand{\rn}[1]{\mathsf{#1}}

\newcommand{\irn}[1]{\ensuremath{\rn{#1}_\mathsf{I}}}
\newcommand{\ern}[1]{\ensuremath{\rn{#1}_\mathsf{E}}}

\newcommand{\ax}{\rn{id}}

  \newcommand{\flatBI}[1]{{#1}^{\flat}}
\newcommand{\deflatBI}[1]{{#1}^{\natural
}}

\newcommand{\baseBI}{\base{N}}
\newcommand{\calculusBI}{\system{NBI}}

\newcommand{\natN}{\mathbb{N}}
\newcommand{\lrangle}[1]{\langle #1 \rangle}

\makeatletter
\let\orgdescriptionlabel\descriptionlabel
\renewcommand*{\descriptionlabel}[1]{%
  \let\orglabel\label
  \let\label\@gobble
  \phantomsection
  \edef\@currentlabel{#1}%
  \let\label\orglabel
  \orgdescriptionlabel{#1}%
}
\makeatother

\author*[1]{\fnm{Tao} \sur{Gu}}\email{tao.gu.18@ucl.ac.uk}

\author*[1]{\fnm{Alexander V.} \sur{Gheorghiu}}\email{alexander.gheorghiu.19@ucl.ac.uk}

\author*[1,2,3]{\fnm{David J.} \sur{Pym}}\email{d.pym@ucl.ac.uk, david.pym@sas.ac.uk}

\affil*[1]{\orgdiv{Department of Computer Science}, \orgname{University College London}, \orgaddress{\street{Gower Street}, \city{London}, \postcode{WC1E 6BT}, \country{United Kingdom}}}

\affil[2]{\orgdiv{Department of Philosophy}, \orgname{University College London}, \orgaddress{\street{Gower Street}, \city{London}, \postcode{WC1E 6BT}, \country{United Kingdom}}}

\affil[3]{\orgdiv{Institute of Philosophy}, \orgname{University of London}, \orgaddress{\street{Malet St}, \city{London}, \postcode{WC1E 7HU}, \country{United Kingdom}}}

\renewcommand{\abstract}[1]{\textbf{Abstract.} #1}

\title[Proof-theoretic Semantics for BI]{Proof-theoretic Semantics for the Logic of Bunched Implications}

\begin{document}

\maketitle              

\abstract{ 
The \emph{logic of bunched implications} (BI) can be seen as the free combination of \emph{intuitionistic propositional 
logic} (IPL) and \emph{intuitionistic multiplicative linear logic} (IMLL). We present here 
a base-extension semantics (B-eS) for BI in the spirit of Sandqvist's B-eS for IPL, 
deferring an analysis of proof-theoretic validity, in the sense of Dummett and Prawitz, 
to another occasion. Essential to BI's formulation in proof-theoretic terms is the concept 
of a `bunch' of hypotheses that is familiar from relevance logic. Bunches amount to trees 
whose internal vertices are labelled with either the IMLL context-former or the IPL 
context-former and whose leaves are labelled with propositions or units for the context-formers. This structure presents 
significant technical challenges in setting up a base-extension semantics for BI. Our 
approach starts from the B-eS for IPL and the B-eS for IMLL and provides a systematic 
combination. Such a combination requires that base rules carry bunched structure, and so 
requires a more complex notion of derivability in a base and a correspondingly richer 
notion of support in a base. One reason why BI is a substructural logic of interest is 
that the `resource interpretation' of its semantics, given in terms of sharing and 
separation and which gives rise to Separation Logic in the field of program verification, is quite distinct from the 
`number-of-uses' reading of the propositions of linear logic as resources. This resource 
reading of BI provides useful intuitions in the formulation of its proof-theoretic semantics. We discuss a simple example of the use of the given B-eS in security modelling. } 


\section{Introduction} \label{sec:introduction}


In model-theoretic semantics (M-tS), logical consequence is defined in terms of truth in models, which are abstract mathematical structures in which propositions are interpreted and their truth is judged. In the standard reading given by Tarski~\cite{Tarski1936,Tarski2002}, a propositional formula $\phi$ follows (model-theoretically) from a context
$\varGamma$ iff every model of $\varGamma$ is a model of $\phi$:  
\[
  \begin{array}{r@{\quad}c@{\quad}l}
    \mbox{$\varGamma \models \phi$} & \mbox{iff} & \mbox{for all models $\mathcal{M}$, if  $\mathcal{M} \models \psi$ for all $\psi \in \varGamma$, then $\mathcal{M} \models \phi$}
  \end{array}   
\]
Therefore, consequence is understood as the \emph{transmission of truth}. Importantly, from this perspective, \emph{meaning} and \emph{validity} are characterized \emph{in terms of truth}. 

Proof-theoretic semantics (P-tS) is an alternative approach to meaning and validity in which these two concepts are characterized in terms of \emph{proofs}, where `proofs' are understood as objects denoting collections of acceptable inferences from accepted premisses. This is subtle. Crucially, P-tS is not about providing a proof system for the logic, but about explicating meaning \emph{in terms of `proof'}. Indeed, as Schroeder-Heister \cite{Schroeder2007modelvsproof} observes, since no formal system is fixed (only notions of inference) the relationship between semantics and provability for the logic in question remains the same as it has always been. In particular, soundness and completeness with respect to a P-tS are desirable features of formal systems. Essentially, what differs in P-tS is that \emph{proof}  plays the part of \emph{truth} in M-tS. 

We defer to Schroeder-Heister~\cite{Schroeder2006validity,Schroeder2007modelvsproof,SEP-PtS} for discussion of the motivation and philosophy for P-tS. The paradigm of meaning supporting P-tS is \emph{inferentialism}; that is, the view that meaning (or validity) arises from rules of inference (see Brandom~\cite{Brandom2000}). This stands in contrast to, for example, denotationalism as the paradigm of meaning underpinning M-tS. 

Current work in P-tS largely follows two different approaches. We call the first \emph{proof-theoretic validity} (P-tV): it aims to define what makes an object purporting to be a `proof' \emph{valid}. Major approaches in P-tV follow ideas by Dummett \cite{Dummett1991logical} and Prawitz \cite{Prawitz1971ideas,Prawitz2006natural} --- see, for example, Schroeder-Heister \cite{Schroeder2006validity}. It is closely related to, though not coincident with, BHK semantics --- see Schroeder-Heister~\cite{Schroeder2007modelvsproof}.

We call the second approach to P-tS \emph{base-extension semantics} (B-eS): it aims to give a semantics to the logical constants through defining the validity of formulae in terms of consequence/inference in arbitrary pre-logical systems called \emph{bases}. Major ideas in B-eS relevant to this paper follow Makinson, Piecha, Sandqvist, and Schroeder-Heister, and Stafford~\cite{makinson2014inferential,Piecha2017definitional,Sandqvist2005inferentialist,Sandqvist2009CL,Sandqvist2015base,Schroeder2016atomic,Stafford2021}. In Section~\ref{sec:intro-B-eS}, we explain the B-eS for intuitionistic propositional logic (IPL) by Sandqvist~\cite{Sandqvist2015base} and explain how it motivates the work in this paper. 

Importantly, the nomenclature outlined here is taken from the historical development of the respective areas, but is not intended to divide the techniques of P-tS. In particular, both P-tV and B-eS concern \emph{validity} --- the former \emph{of} proofs, and the latter \emph{in terms of} proofs --- and both may use `bases' and `base-extensions' in so doing.  

This paper delivers a P-tS for \emph{the logic of bunched implications} (BI) \cite{o1999logic,Brotherston2012,GP2023-SABI}. The logic is, perhaps, best-known for its role as the logical basis for the program logic known as \emph{Separation Logic} \cite{Ishtiaq2001,reynolds2002separation,pym2019separation} and its many derivatives. Formally, as explained in \cite{Ishtiaq2001}, 
Separation Logic amounts to a specific theory of BI that is concerned with memory locations. This application of BI is essentially based on its reading as a logic of resource --- see, for example, Pym~\cite{pym2019resource}. In this paper, we concentrate on the technical development of P-tS for BI. 

Why study B-eS for BI? One motivation is to be able to understand and account for \emph{resources} in formal systems, semantically. Consider the celebrated `resource interpretation' of \emph{linear logic} (LL)~\cite{girard1995} via the number-of-uses reading in which a proof of a formula $\phi \multimap \psi$ determines a function that \emph{uses} its arguments exactly once. This reading is entirely proof-theoretic and is \emph{not} at all expressed in the truth-functional semantics of the logic --- see Girard \cite{girard1995}, Allwein and Dunn \cite{Allwein}, and Coumans et al. \cite{COUMANS201450}. Meanwhile, this reading of resource is quite clearly seen in the B-eS for IMLL by Gheorghiu et al.~\cite{gheorghiu2023imll}. This illustrates the value of studying B-eS for substructural logics in informatics as far as their use as logics of resources according to the typical number-of-uses reading. Of course, there are other applications of substructural logics to formal systems, but they all tend to be close to the proof-theoretic accounts of the logic for which P-tS is ideally suited.

Thus, on the one hand we have P-tS as a semantic paradigm in which a certain common notion of resource is handled explicitly and naturally, on the other we have BI as a logic whose resource interpretation has formidable applications to informatics and systems science. This paper brings these two worlds together. This connection is further explored by the present authors in \cite{MFPS-IRS-2024}. 

However, resource-sensitive semantics is not the only reason to study BI. Another use of the logic is as the basis for van~Benthem-Hennessy-Milner-type logics associated with process algebras that provide a foundation for simulation modelling tools --- see, for example, Anderson and Pym~\cite{ANDERSON201663}. Kuorikoski and Reijula \cite{jaakko} have recently used P-tS in the context of simulation modelling --- briefly, execution of a model (or a program, akin to a proof in a formal system) is understood as one of its possible interpretations --- giving another reason to study the B-eS for BI in particular. More 
fundamentally, however, BI is a logic that integrates `additive' (here, intuitionistic) and `multiplicative' structure and is closely related to the class of relevance logics --- note, for relevance logics the additive and multiplicative structures are often given opposite notation to what is used in BI. Hence, studying the P-tS for BI expands the field into the wider space of logical systems across philosophy, mathematics, and informatics. In this paper, we illustrate the value of the B-eS of BI by application to yet another domain: rule-based access control --- see Section~\ref{sec:support}.

In Section~\ref{sec:intro-B-eS}, we illustrate how B-eS for IPL and for IMLL, and discuss the challenges in handling BI. In Section~\ref{sec:bi}, we introduce BI in detail, providing its definition at the level of detail that is required to set up its B-eS. and including its sequential natural deduction system, NBI. In Section~\ref{sec:bes}, we set up BI's Be-S and give, in Section~\ref{sec:eg-access-control}, a toy example of how the semantics can be used to model a resource-sensitive systems example (about access control), where we also briefly discuss the developments of this idea that have been presented in \cite{MFPS-IRS-2024}. In Sections \ref{sec:soundness} and 
\ref{sec:completeness}, we establish the soundness and completeness, respectively, of BI's sequential natural deduction system, NBI. Finally, in Section~\ref{sec:conclusion}, we summarize our contribution and briefly discuss some directions for further work. 

\section{Background: Base-extension Semantics}
\label{sec:intro-B-eS}

A \emph{base} is a set of \emph{atomic rules}. An \emph{atomic rule} is a natural deduction rule (in the sense of Gentzen~\cite{Gentzen}) over atomic propositions (not closed under substitution); for example, following Sandqvist~\cite{sandqvistwld}, from the propositions  `Tammy is a fox' and `Tammy is female' one may infer `Tammy is a vixen',
\[
    \infer{\text{Tammy is a vixen}}{\text{Tammy is a fox} & \text{Tammy is female}}
\]
This is the standard approach to atomic rules, which will be generalized in this paper to accommodate certain technical issues explained below; in particular, we move from a natural  deduction presentation to a sequent calculus presentation. See Piecha and Schroeder-Heister~\cite{Schroeder2016atomic,Piecha2017definitional} for various interpretations of atomic systems. 

In B-eS, the validity of atoms and the meaning of the logical connectives is inductively defined relative to bases --- see, for example, Figure~\ref{fig:sandqvist} for \emph{intuitionistic propositional logic} (IPL), Figure~\ref{fig:IMLL} for \emph{intuitionistic multiplicative linear logic} (IMLL), and Figure~\ref{fig:suppBI} for \emph{the logic of bunched implications} (BI), which are all explained below. The choice of the form of atomic rules has profound repercussions; for example, limiting the atomic rules to simple production rules in the B-eS for IPL yields incompleteness, or (conversely) generalizing atomic rules include discharged hypotheses and making a subtle choice about the treatment of disjunction (cf. the incompleteness results by Piecha and Schroeder-Heister~\cite{Piecha2015failure,Piecha2016completeness,Piecha2019incompleteness}) in the semantics of CPL delivers a semantics for IPL --- see Sandqvist~\cite{Sandqvist2005inferentialist,Sandqvist2009CL,Sandqvist2015hypothesis}. This understanding of B-eS is the form of P-tS used in this paper; see Goldfarb~\cite{goldfarb2016dummett} and Stafford and Nascimento~\cite{Stafford2023} for alternative approaches. 


\begin{figure}[t]
\hrule \vspace{2mm}
 \[
        \begin{array}{l@{\qquad}c@{\quad}l@{\quad}r}
            \suppTor{\base{B}} \at{p}  & \mbox{iff} &   \proves[\base{B}] \at{p} & \mbox{(At)}  \\[1mm]
             \suppTor{\base{B}} \phi \to \psi & \mbox{iff} & \phi \suppTor{\base{B}} \psi & (\to) \\[1mm]
             \suppTor{\base{B}} \phi \land \psi   & \mbox{iff} & \mbox{$\suppTor{\base{B}} \phi$ and $\suppTor{\base{B}} \psi$} & (\land) \\[1mm]
             \suppTor{\base{B}} \phi \lor \psi & \mbox{iff} &  \mbox{$\forall \base{C} \supseteq \baseB$ and $\forall \at{p} \in \At$, if $\phi \suppTor{\base{C}} \at{p}$ and  $\psi \suppTor{\base{C}} \at{p}$, then $\suppTor{\base{C}} \at{p}$} & (\lor)  \\[1mm]
             \suppTor{\base{B}} \bot & \mbox{iff} &    \suppTor{\base{B}} \at{p} \text{ for any } \at{p} \in \At & (\bot) \\[1mm]
           \hspace{-1em} \varGamma \suppTor{\base{B}} \phi & \mbox{iff} & \mbox{$\forall \base{C} \supseteq \baseB$, if $\suppTor{\base{C}} \psi$ for every $\psi \in \varGamma$, then $\suppTor{\base{C}} \phi$ } &  \mbox{(Inf)} 
        \end{array}
        \]
        \vspace{2mm}
     \dotfill 
         \[
        \raisebox{-1em}{
        \infer[\rn{r}]{\at{c}}{\deduce{\at{p}_1}{[\at{P}_1]} & ... & \deduce{\at{p}_n}{[\at{P}_n]} }
        }
        \quad 
            \begin{array}{r@{\quad}l}
                \mbox{\hypertarget{ipl-ref}{\textsc{ref}}:} & \mbox{$\at{P} , \at{p} \vdash_\base{B} \at{p}$} \\ 
                \mbox{\hypertarget{ipl-app}{\textsc{app}}:} & \mbox{If $r \in \baseB$ and $\at{P},\at{P}_i \proves[\base{B}] \at{p}_i$ for $i = 1\ldots n$, then $\at{P} \proves[\base{B}] \at{c}$} 
            \end{array}
    \]
    \hrule
    \caption{Base-extension Semantics for IPL} 
    \label{fig:sandqvist}
\end{figure}

BI can be understood as the free combination of IPL --- with connectives $\land, \lor, \to, \top, \bot$ --- and IMLL --- with connectives $\mand, \wand, \I$ (the unit of $\ast$). Both of these logics already have  B-eS --- see Sandqvist~\cite{Sandqvist2015base} and Gheorghiu et al.~\cite{gheorghiu2023imll} --- which suggests that the semantics for BI should be a `free combination' of those semantics. The same is true of BI's model-theoretic semantics, but this topic is quite delicate --- see, for example, work by Docherty, Galmiche, Gheorghiu, M\'ery, Pym~\cite{Docherty2019,Galmiche2005,Alex:BI_Semantics}.  The issues arise out of the complex structures involved in BI. As a consequence of having two implications, contexts in BI are not the typical \emph{flat} data structures used in logic (e.g., lists or multisets), but instead are layered structures called \emph{bunches}, a term that derives from the relevance logic literature --- see, for example, Read~\cite{Read1988}. Therefore the technical challenges for providing a B-eS for BI amount to understanding how to handle the bunch structure, and how that manifests in the resource reading of the logic and semantics. The technical details are contained within the body of the paper. The remainder of this section provides intuitions as to how these ideas are developed.

Our starting point is the B-eS for IPL by Sandqvist~\cite{Sandqvist2015base}. Fix a (denumerable) set of atoms $\At$. A base $\base{B}$ is a set of atomic rules as described above and provability in a base ($\proves[\baseB])$ is defined as in Figure~\ref{fig:sandqvist} --- that is, inductively, by \hyperlink{ipl-ref}{\textsc{ref}} and \hyperlink{ipl-app}{\textsc{app}}. This provides the base case of the \emph{support} relation defined in Figure~\ref{fig:sandqvist} that provides the semantics --- here, $\at{p},\at{p}_1,\ldots,\at{p}_n,\at{c} \in \At$, $\at{P}_1,\dots \at{P}_n,\ldots \subseteq \At$ (finite), and $\phi,\psi, \ldots$ denote IPL-formulae, and $\varGamma$ denotes a set of IPL-formulae. It differs from the standard Kripke semantics for IPL \cite{kripke1965semantical} in the treatment of disjunction ($\lor$) and absurdity ($\bot$); indeed, this choice is critical as Piecha and Schroeder-Heister~\cite{Piecha2015failure,Piecha2016completeness,Piecha2019incompleteness} have shown that a meta-level `or' for disjunction results in incompleteness.

Recently, the present authors have given 
a B-eS for IMLL~\cite{gheorghiu2023imll} in the style of the B-eS for IPL given above. This is summarized in Figure~\ref{fig:IMLL} in which $\at{P}_1,\ldots,\at{P}_n,\at{U},\at{V}$ are \emph{multisets} of atoms, $\varGamma, \varDelta$ are \emph{multisets} of IMLL formulae, and $\mmcomma$ denotes multiset union. There are a few features to note about this semantics. First, observe that provability in a base ($\proves[\baseB]$) has been made substructural as the base case (\hyperlink{imll-ref}{\textsc{ref}}) is restricted to have only the conclusion as a context and the inductive step (\hyperlink{app}{\textsc{app}}) combines the contexts delivering each premiss of an atomic rule for the context delivering the conclusion of that rule. Second, observe that the support judgement 
$\suppM{}{}$ is parametrized on a multiset of atoms, which can usefully be thought 
of as `resources'. Doing this manifests the number-of-uses resource reading of IMLL restricted to atoms. The role of resource parametrization 
can be seen particularly clearly in the \hyperlink{imll-inf}{(Inf)} clause: the resources required for the sequent $\varGamma \seq \phi$ are combined with those required for $\varGamma$ in order to deliver those required for $\phi$. Third, the treatment of the support of a 
combination of contexts $\varGamma \mcomma \varDelta$ follows the na\"{\i}ve Kripke-style interpretation of multiplicative conjunction, corresponding to an introduction rule in natural 
deduction, but the support of the tensor product $\otimes$ follows the form of a natural deduction elimination rule. Fourth, observe that everywhere there are restrictions to the atomic case (e.g., the $\otimes$-clause takes the form of the $\otimes$-elimination rule \emph{restricted} to atomic conclusions), which are required to ensure the connection to provability in a base.  Overall, we observe that aside from expected modifications on the set-up for IPL, there are two key technical developments required in this work: first, the use of resources to capture substructurality; second, a delicate treatment 
of the relationship between the multiplicative context-building operation and the corresponding multiplicative conjunction.

\begin{figure}[t]
\hrule
\vspace{2mm}
\[
 \begin{array}{l@{\quad}c@{\quad}l@{\qquad}r}
            \mbox{$\suppM{\baseB}{\at{P}} \at{p}$} &  \mbox{iff} & \mbox{$\at{P} \deriveBaseM{\baseB} \at{p}$} & \mbox{(At)}  \\ [1mm]
           \mbox{$\suppM{\baseB}{\at{P}} \phi \multimap \psi$} &  \mbox{iff} & 
            \mbox{$\phi \suppM{\baseB}{\at{P}} \psi$} &   \mbox{($\multimap$)} \\[1mm]
             \mbox{$\suppM{\baseB}{\at{P}} \phi \otimes \psi$} & \mbox{iff} &\mbox{$\forall \baseX \baseGeq \baseB$, $\forall \at{U}$, $\forall \at{p} \in\At$, if $\phi \mcomma \psi \suppM{\baseX}{\at{U}} \at{p}$, then $\suppM{\baseX}{\at{P} \mmcomma \at{U}} \at{p}$} & \mbox{($\otimes$)} \\ [1mm]
           \mbox{$\suppM{\baseB}{\at{P}} \mtop$} & \mbox{iff} & \mbox{$\forall \baseX \baseGeq \baseB$, $\forall \at{U}$, $\forall \at{p} \in \At$, if $\suppM{\baseX}{\;\at{U}} \at{p}$, then $\suppM{\baseX}{\;\at{P} \;\mcomma\;  \at{U}} \at{p}$} &  \mbox{($\mtop$)}  \\ [1mm]
           \mbox{$ \suppM{\baseB}{\at{P}} \varGamma \mcomma \varDelta$} & \mbox{iff} &  \mbox{$\exists \at{U},\at{V}$ s.t. $\at{P} = (\at{U} \mcomma \at{V})$, $\suppM{\baseB}{\at{U}} \varGamma$, and  $\suppM{\baseB}{\at{V}} \varDelta$} &  \mbox{(\,\,$\mcomma$\,\,)} 
            \\[1mm]
             \hspace{-1.4em} \mbox{$\varGamma \suppM{\baseB}{\at{P}} \phi$} & \mbox{iff} & 
        \mbox{$\forall \baseX \baseGeq \baseB$ and $\forall \at{U}$, if $\suppM{\baseX}{\at{U}} \varGamma$, then $\suppM{\baseX}{\;\at{P} \;\mcomma\; \at{U}} \phi$} & \mbox{\hypertarget{imll-inf}{(Inf)}}
        \end{array}
\]
    \dotfill \vspace{-1mm}
         \[
        \raisebox{-1em}{
        \infer[\rn{r}]{\at{c}}{\deduce{\at{p}_1}{[\at{P}_1]} & ... & \deduce{\at{p}_n}{[\at{P}_n]} }
        }
        \quad 
            \begin{array}{r@{\quad}l}
                \mbox{\hypertarget{imll-ref}{\textsc{ref}}:} & \at{p} \proves[\baseB] \at{p} \\ 
                \mbox{\hypertarget{imll-app}{\textsc{app}}:} & \mbox{If $\rn{r} \in \baseB$ and $\at{S}_i \mcomma \at{P}_i \deriveBaseM{\baseB} \at{p}_i$ for $i = 1 \ldots n$, then $\at{S}_1 \mcomma \dots \mcomma \at{S}_n \deriveBaseM{\baseB} \at{c}$} 
            \end{array}
    \]
\hrule
\caption{Base-extension Semantics for IMLL} \label{fig:IMLL}
\vspace{-1em}
\end{figure}

 In this paper, we deliver the B-eS for BI by generalizing the set-up for IMLL from multisets to bunches, mixing multiplicative structure and additive structure. There are a number of technical challenges that arise when doing this. First, we diverge from the traditional literature on atomic systems in P-tS --- see, for example, Piecha and Schroeder-Heister~\cite{Schroeder2016atomic,Piecha2017definitional} --- as natural deduction in the sense of Gentzen~\cite{Gentzen} appears insufficiently expressive to make all the requisite distinctions with regards to multiplicative and additive structures in bunches. 
 Second, in order to recover the appropriate amount of structurality in the clauses of the logical constants we introduce a \emph{bunch-extension} relation. Third, we introduce formally a notion of \emph{contextual bunch} (i.e., 
a bunch with a hole) in order to handle substitution in bunches parameterizing resources in the support relation. The importance of this structure can be seen from the (Inf) clause 
that we shall need for BI (see Figure~\ref{fig:suppBI}): 
\begin{equation*}
\begin{array}{l@{\quad}c@{\quad}l@{\qquad}r}
   \varGamma \supp{\at{R}(\cdot)}{\base{B}} \phi & \mbox{iff} & 
   \mbox{$\forall \base{X} \baseGeq \base{B}$, $\forall \at{U} \in \Bunches(\At)$, 
   if $\supp{\at{U}}{\baseX} \varGamma$, then $\supp{\at{R(U)}}{\baseX} \phi$}  & \mbox{(Inf)} 
\end{array}
\end{equation*}
That is, the contextual bunch $\at{R}(\cdot)$ manages the construction of the bunch of atomic resources in a way that matches the structure of the required judgement.


\section{The Logic of Bunched Implications} \label{sec:bi}

In this section, we define \BI{} and give the technical background required for the development of this paper. First, we define its syntax and consequence relation in Section~\ref{sec:bi:syntax}. Second, we give the meta-theory required to understand the use of \emph{bunches} in this paper in Section~\ref{sec:bunches_modulo_ce}. We defer motivating BI as a logic to Pym and O'Hearn~\cite{o1999logic} to concentrate on the technical question of its P-tS in this paper. \vspace{2mm}

\noindent \textbf{Notations \& Conventions.} Henceforth, we fix a denumerable set of atomic propositions $\At$, and the following conventions: 
$\at{p}, \at{q}, \dots$ denote atoms; $\at{P}, \at{Q}, \dots$ denote \emph{bunches} of atoms; $\phi, \psi, \theta, \dots$ denote formulae; $\varGamma, \varDelta, \dots$ denote \emph{bunches} of formulae. We also use $\base{A},\baseB,\baseC, \ldots$ to denote bases (defined below).

\subsection{Syntax and Consequence}
\label{sec:bi:syntax}
In this section, we present the syntax of \BI{} and define the consequence relation in terms of a natural deduction system in sequent calculus form. We follow standard presentations (cf. O'Hearn and Pym~\cite{o1999logic}) except that it also introduces \emph{contextual bunches}, which are used throughout this paper but are not part of the standard background of \BI{}. 

\begin{definition}[Formulae]
		 The set of formulae $\Formulas$ (over $\Atoms$) is defined by the following grammar:
		\[
		\phi ::= \at{p} \in \Atoms \mid \top \mid \bot \mid \mtop \mid \phi \land \phi \mid \phi \lor \phi \mid \phi \to \phi \mid \phi \mand \phi \mid \phi \wand \phi
		\]
	\end{definition}
 
\begin{definition}[Bunches] \label{def:bunch}
	Let $\set{X}$ be a set of syntactic structures. The set of bunches over $\set{X}$ is denoted $\Bunches(\set{X})$ and is defined by the following grammar:
		\[
		\varGamma ::= x \in \set{X} \mid \aunit \mid \munit \mid \varGamma \fatsemi \varGamma \mid \varGamma \fatcomma \varGamma
		\]
    The $\fatsemi$ is the additive context-former, and the $\aunit$ is the additive unit; the $\fatcomma$ is the multiplicative context-former, and the $\munit$ is the multiplicative unit.
\end{definition}

In the bunch $\Gamma_1 \acomma \Gamma_2$ (resp. $\Gamma_1 \mcomma \Gamma_2$), we call $\aacomma$ (resp. $\mmcomma$) the \emph{principal} context-former.
If a bunch $\varDelta$ is a sub-tree of a bunch $\varGamma$, then $\varDelta$ is a \emph{sub-bunch} of $\varGamma$. Importantly, sub-bunches are sensitive to occurrence --- for example, in the bunch $\varDelta \mcomma \varDelta$, each occurence of $\varDelta$ is a different sub-bunch. We may write $\varGamma(\varDelta)$ to express that $\varDelta$ is a sub-bunch of $\varGamma$. This notation is traditional for \BI{}~\cite{o1999logic}, but requires sensitive handling. To avoid confusion, we never write $\varGamma$ and later $\varGamma(\varDelta)$ to denote the same bunch, but rather maintain one presentation.  The substitution of $\varDelta'$ for $\varDelta$ in $\varGamma$ is denoted $\varGamma(\varDelta)[\varDelta \mapsto \varDelta']$. To emphasize: this is not a universal substitution, but a substitution of the particular occurrence of $\varDelta$ meant in writing $\varGamma(\varDelta)$. The bunch $\varGamma(\varDelta)[\varDelta \mapsto \varDelta']$ may be denoted $\varGamma(\varDelta')$ when no confusion arises.

\begin{example} \label{ex:bunch-sub}
    Consider the bunches $\Gamma(q):= \munit \mcomma (p \acomma q)$. We have by direct substitution $\Gamma(\aunit) = \munit \mcomma (p \acomma \aunit)$. 
\end{example}

Substitution in bunches is an essential part of expressing \BI{}. It plays the part of concatenation of lists and union of sets and multisets in the presentations of other logics --- see, for example, van Dalen~\cite{vanDalen}. Indeed, so integral are substitutions that it is instructive to have a notion of \emph{contextual bunch} that allows us to handle substitutions smoothly.  

\begin{definition}[Contextual Bunch]
    A \emph{contextual bunch} (over $\set{X}$) is a function $b:\Bunches(\set{X}) \to \Bunches(\set{X})$ for which there is $\varGamma(\varDelta) \in \Bunches(\set{X})$ such that $b(\varSigma) = \varGamma(\varSigma)$ for any $\varSigma \in \Bunches(\set{X})$. The set of all contextual bunches (over $\set{X}$) is $\dot{\Bunches}(\set{X})$.
\end{definition}

Importantly, the identity on $\Bunches(\set{X})$ is a contextual bunch, it is denoted $(\cdot)$.
We think of a contextual bunch as  a \emph{bunch with a hole} --- cf. \emph{nests with a hole} in work by Br\"unner~\cite{Brunnler2009}. That is, $\dot{\Bunches}(\set{X})$ can be identified as the subset of $\Bunches(\set{X}\cup\{\circ\})$, where $\circ \not \in \set{X}$, in which bunches contain a single occurrence of $\circ$. Specifically,  if $b(x) =\varGamma(x)$, then identify $b$ with $\varGamma(\circ) \in \Bunches(\set{X}\cup\{\circ\})$. We write $\varGamma(\cdot)$ for the contextual bunch identified with $\varGamma(\circ)$. 

 \begin{example}\label{ex:bunch-with-hole}
 The contextual bunch $b: x \mapsto \phi_1 \acomma ( (\phi_2 \acomma \munit \acomma (\phi_3 \mcomma \phi_4) ) \mcomma (\phi_5 \acomma (x )) )$ is identified with $\varGamma(\cdot)$, where $\varGamma(\circ):=\phi_1 \acomma ( (\phi_2 \acomma \munit \acomma (\phi_3 \mcomma \phi_4) ) \mcomma (\phi_5 \acomma \circ ) )$.
 \end{example}

In \BI{}, bunches serve as the data structures over which collections of formulae are organized. Each of the two context-formers (i.e., $\aacomma$ and $\mmcomma$) is intended behave as a multiset constructor, possessing the usual properties of commutativity and associativity. For example, a bunch $\Delta \mcomma \Delta'$ (resp. $\Delta \acomma \Delta'$) is meant to represent the same data as the bunch $\Delta' \mcomma \Delta$ (resp. $\Delta \mcomma \Delta'$). For clarity, we explicitly define the equivalence relation that captures this intended meaning, called \emph{coherent equivalence}. This is similar to how lists for the contexts of classical logic sequents may be understood `up to permutation' to render multisets. Importantly, in BI, the two context-formers do not distribute over each other.

\begin{definition}[Coherent Equivalence] \label{def:coherent-equivalence}
	Two bunches $\varGamma,\varGamma' \in \Bunches(\set{X})$ are \emph{coherently equivalent} when $\varGamma \equiv \varGamma'$, where $\equiv$ is the least relation satisfying: 
	
	\begin{itemize}[label=--]
		\item commutative monoid equations for $\fatsemi$ with unit $\aunit$--- that is, for arbitrary $\Gamma_1, \Gamma_2, \Gamma_3 \in \Bunches(\set{X})$,
            \[
            (\Gamma_1 \acomma \Gamma_2) \acomma \Gamma_3 \equiv \Gamma_1 \acomma (\Gamma_2 \acomma \Gamma_3) \qquad \Gamma_1 \acomma \Gamma_2 \equiv \Gamma_2 \acomma \Gamma_1 \qquad \Gamma_1 \acomma \aunit \equiv \Gamma_1
            \]
		\item commutative monoid equations for $\fatcomma$ with unit $\munit$--- that is, for arbitrary $\Gamma_1, \Gamma_2, \Gamma_3 \in \Bunches(\set{X})$,
            \[
            (\Gamma_1 \mcomma \Gamma_2) \mcomma \Gamma_3 \equiv \Gamma_1 \mcomma (\Gamma_2 \mcomma \Gamma_3) \qquad \Gamma_1 \mcomma \Gamma_2 \equiv \Gamma_2 \mcomma \Gamma_1 \qquad \Gamma_1 \mcomma \munit \equiv \Gamma_1
            \]
		\item coherence; that is, if $\varDelta\equiv\varDelta' $ then $\varGamma(\varDelta) \equiv \varGamma(\varDelta')$. 
	\end{itemize} 
	
\end{definition}

\begin{example}[Example~\ref{ex:bunch-sub} cont'd]
    Recall the bunch $\Gamma(q):= \munit \mcomma (p \acomma q)$ from Example~\ref{ex:bunch-sub}. Observe, $\Gamma(\aunit) \equiv p$ by the the unitality of $\aunit$ and $\munit$ with respect to $\aacomma$ and $\mmcomma$, respectively.
\end{example}

 Further discussion about bunches as data structures is given in Section~\ref{sec:bunches_modulo_ce} and Section~\ref{sec:bes}. Presently, we continue defining \BI{}.

	\begin{definition}[Sequent] \label{def:sequent}
		A \emph{sequent} is a pair $\varGamma \seq \phi $ in which $\varGamma\in \Bunches(\Formulas)$ is a bunch of formulae and $\phi$ is a formula.
	\end{definition}
 
That a sequent $\varGamma \seq \phi$ is a consequence of \BI{} is denoted $\varGamma \proves \phi$. In this paper, we characterize \BI{} by its natural deduction system $\system{NBI}$. We assume general familiarity with sequent calculus presentation of proof systems --- see, for example, Troelstra and Schwichtenberg~\cite{troelstra2000basic}.

\begin{definition}[Natural Deduction System $\system{NBI}$]
    System $\system{NBI}$ comprises the rules of Figure~\ref{fig:NBI}, in which $\phi, \psi, \chi \in \Formulas$, $\varGamma, \varDelta, \varDelta' \in \Bunches(\set{F})$, and $\hyperlink{NBI-exchange}{\exch}$ has the side-condition $\varGamma \equiv \varGamma'$ and  $i\in \{1,2\}$ in $\hyperlink{NBI-or-i}{\irn\lor}$.
\end{definition}

\begin{figure}[t]
\hrule
\vspace{2mm}
        \[
        \begin{array}{c}
        \infer[\myhypertarget{NBI-id}{\ax}]{\phi \seq \phi}{ } 
        \quad 
        \infer[\myhypertarget{NBI-exchange}{\exch}]{\varGamma \seq \phi}{\varGamma' \seq \phi} 
        \quad 
        \infer[\weak]{\varGamma( \varDelta \addcomma \varDelta') \seq \phi}{\varGamma(\varDelta) \seq \phi} 
        \quad 
        \infer[\cont]{\varGamma(\varDelta) \seq \phi}{\varGamma(\varDelta \addcomma \varDelta) \seq \phi} 
        \\[1.5ex]
        \infer[\irn \mtop]{\munit \seq \mtop}{} 
        \quad
        \infer[\ern \mtop]{\varGamma(\varDelta) \seq \chi}{\varGamma(\munit) \seq \chi & \varDelta \seq \mtop} 
        \quad 
        \infer[\irn \wand]{\varGamma \seq \phi \wand \psi}{\varGamma \multcomma \phi \seq \psi}
        \\[1.5ex]
        \infer[\ern \wand]{\varGamma \multcomma \varDelta \seq \psi}{\varGamma \seq \phi \wand \psi & \varDelta \seq \phi} 
        \quad 
        \infer[\irn \mand]{\varGamma \multcomma \varDelta \seq \phi \mand \psi}{\varGamma \seq \phi & \varDelta \seq \psi} 
        \quad 
        \infer[\ern \mand]{\varGamma(\varDelta) \seq \chi}{\varGamma(\phi \multcomma \psi) \seq \chi & \varDelta \seq \phi \mand \psi}
        \\[1.5ex]
        \infer[\irn\top]{\aunit \seq \top}{} 
        \quad
        \infer[\ern\top]{\varGamma(\varDelta) \seq \chi}{\varGamma(\aunit) \seq \chi & \varDelta \seq \top} 
        \quad 
        \infer[\irn\to]{\varGamma \seq \phi \to \psi}{\varGamma \acomma \phi \seq \psi} 
        \\[1.5ex]
        \infer[\ern\to]{\varGamma \acomma \varDelta \seq \psi}{\varGamma \seq \phi \to \psi & \varDelta \seq \phi} 
        \quad 
        \infer[\irn\land]{\varGamma \acomma \varDelta \seq \phi \land \psi}{\varGamma \seq \phi & \varDelta \seq \psi} 
        \quad 
        \infer[\ern\land]{\varGamma(\varDelta) \seq \chi}{\varGamma(\phi \acomma \psi) \seq \chi & \varDelta \seq \phi \land \psi}
        \\[1.5ex]
        \infer[\myhypertarget{NBI-or-i}{\irn\lor}]{\varGamma \seq \phi_1 \lor \phi_2}{\varGamma \seq \phi_i}
        \quad 
        \infer[\ern\lor]{\varGamma(\varDelta) \seq \chi}{\varGamma(\phi) \seq \chi & \varGamma(\psi) \seq \chi & \varDelta \seq \phi \lor \psi} 
        \quad 
        \infer[\ern\bot]{\varGamma \seq \phi}{\varGamma \seq \bot} 
        \end{array}
        \vspace{2mm}
        \]
\hrule
	\caption{Natural Deduction System $\system{NBI}$}
	\label{fig:NBI}
 \end{figure}

Let $\proves[\system{NBI}]$ denote provability in $\system{NBI}$. When we say that $\system{NBI}$ characterizes \BI{}, we mean the following: $ \varGamma \proves \phi \mbox{ iff } \varGamma \proves[\system{NBI}] \phi$.  Of course, \BI{} admits $\rn{cut}$: 
\begin{lemma}[Brotherston~\cite{Brotherston2012}] \label{lem:cut-admissibility}
 If $\varGamma(\chi) \proves \phi$ and $\varDelta \proves \chi$, then $\varGamma(\varDelta) \proves \phi$.
\end{lemma}


\subsection{Bunches Modulo Coherent Equivalence} \label{sec:bunches_modulo_ce}
As stated in Section~\ref{sec:bi}, bunches as data structures are really the syntax trees in Definition~\ref{def:bunch} modulo coherent equivalence (Definition~\ref{def:coherent-equivalence}). Gheorghiu and Marin~\cite{Alex:Focusing} gave a formal description of $\Bunches(\set{X})/\scriptsize{\equiv}$ as \emph{nested multisets of two kinds} that delivers this view. While the formalism is intuitive, it ultimately requires more metatheory and background work than is useful for the purpose of the current paper. Therefore, here we forego a formal treatment, and instead illustrate the main idea. We will move smoothly between bunches and bunches modulo coherent equivalence when no confusion arises. Regarding bunches in this way allows us to perform structural analysis (e.g., inductions) and not repeat trivial case-distinctions arising from associativity and commutative of the context-formers.

Essentially, bunches modulo coherent equivalence are trees satisfying the following:
    \begin{itemize}[label = --]
        \item the internal nodes are either $\aacomma$ or $\mmcomma$
        \item the leaves are either $x \in \set{X}$, $\aunit$, or $\munit$ 
        \item for a node $\aacomma$\,, all the non-leaf children are $\mmcomma$
        \item for a node $\mmcomma$\,, all its non-leaf children are $\aacomma$ 
        \item no child of a $\aacomma$-node is labelled with $\aunit$ 
        \item no child of a $\mmcomma$-node is labelled with $\munit$. 
    \end{itemize}
In other words, $\aacomma$- and $\mmcomma$-nodes appear alternately, and units are removed if they can be. This produces a layered effect of alternating additive and multiplicative combinations of data. 

Intuitively, we may identify bunches as being additive or multiplicative according to whether their principal context-former is an additive or multiplicative context-former or unit, respectively. This is not quite a partition as bunches consisting of just a formula are considered both to be additive and multiplicative. An additive (resp. multiplicative) bunch that is not a formula may then be expressed $\varSigma_1 \aacomma \cdots \aacomma \varSigma_k$ (resp. $\varSigma_1 \mmcomma \cdots \mmcomma \varSigma_k$), where $\varSigma_i$ for $i=1,\dots, k$ is a multiplicative (resp. additive) bunch. 


 Henceforth, we do not distinguish bunches and bunches modulo coherent equivalence unless necessary. In the remainder of this section, we provide some simple background that illustrate this quotients reading of bunches and which will be useful later.

\begin{definition}[Nest Depth] \label{def:nestDepth}
    Let $\varGamma \in \Bunches(\set{X})$ be a bunch, $\varDelta$ be a sub-bunch of $\varGamma$. The nest depth of $\varDelta$ in $\varGamma $ --- denoted $\nestDepth{\varGamma}{\varDelta}$ --- is defined as follows:
    \begin{itemize}[label=--]
        \item if $\varGamma = \varDelta$, then $\nestDepth{\varGamma}{\varDelta}:=0$
        \item if $\varGamma$ is of the form $\varSigma_1 \commavar \cdots \commavar \varSigma_k$, where $\commavar \in \{\aacomma,\mmcomma\}$, and $\varDelta$ is a sub-bunch of $\varSigma_i$, then $\nestDepth{\varGamma}{\varDelta}:= \nestDepth{\varSigma_i}{\varDelta}+1$.
    \end{itemize}
\end{definition}

In other words, $\nestDepth{\Gamma}{\Delta}=0$ if $\varDelta$ is $\varGamma$ itself, and otherwise one more than the number of context-former alternations between the principal context-former of $\varDelta$ and the principal context-former of $\varGamma$.

\begin{definition}[Hole Depth] \label{def:holeDepth}
    Let $b \in \dot{\Bunches}(\set{X})$ be a contextual bunch. Identify $b$ with the element $\varGamma(\circ)$ in $\Bunches(\set{X}\cup\{\circ\})$. The \emph{hole depth} of $b$ --- denoted $\holeDepth{b}{}$ --- is $\nestDepth{\varGamma(\circ)}{\circ}$.
\end{definition}
\begin{example}
Let $\varGamma(\cdot)$ be as in Example~\ref{ex:bunch-with-hole}; that is, $\varGamma(\circ) = \phi_1 \acomma ( (\phi_2 \acomma \munit \acomma (\phi_3 \mcomma \phi_4) ) \mcomma (\phi_5 \acomma \circ ) )$.
Then $\holeDepth{\varGamma(\cdot)}{}$ can be calculated as follows:
\begin{align*}
    \holeDepth{\varGamma(\cdot)}{} & = \nestDepth{\varGamma(\circ)}{\circ} = \nestDepth{(\phi_2 \acomma \munit \acomma (\phi_3 \mcomma \phi_4) ) \mcomma (\phi_5 \acomma \circ ) }{\circ} + 1 \\
    & = \nestDepth{\varphi \mcomma \circ}{\circ} + 2 = \nestDepth{\circ}{\circ} + 3 = 3
\end{align*}
This is exactly the number of the alternations of context-formers starting from $\circ$ to the principal context-former in $\Gamma(\circ)$.
\end{example}
%
\section{Base-extension Semantics for BI} \label{sec:bes}

This section defines the B-eS for \BI{}. We defer elaboration and motivation for the sequence of steps and the final form of the semantics to earlier work on IMLL~\cite{gheorghiu2023imll} and other discussions~\cite{GGP2024stateeffect}. This enables us to concentrate on the technical questions of proving soundness and completeness. 

First, in Section~\ref{sec:der-in-a-base}, we define derivability in a base, departing somewhat from existing treatments of atomic systems in P-tS. Second, in Section~\ref{sec:support}, we define the \emph{support} judgement that makes the semantics. Third, some preliminary results that show that support has the expected behaviour with respect to structurality and monotonicity in Section~\ref{sec:preliminary}. Fourth, a toy example of the application to rule-based access control is presented in Section~\ref{sec:eg-access-control}.

\subsection{Derivability in a Base} \label{sec:der-in-a-base}
The definition of a B-eS begins with defining systems of rules called \emph{bases} and the corresponding notion of \emph{derivability in a base}. This forms the base case of the semantics. Traditionally, bases are presented in the style of  natural deduction rules in the sense of Gentzen~\cite{Gentzen} --- see, for example, Peicha and Schroeder-Heister~\cite{Schroeder2016atomic,Piecha2017definitional}. This is  difficult in the case of BI as more careful context-management is required in order to accommodate bunching. 
For this reason, we depart from the standard presentation of bases and atomic rules in P-tS, and present them in a more general sequent calculus form. 

\begin{definition}[Sequent of Atoms] \label{def:atomic-sequent}
    A \emph{sequent of atoms} is a pair $\at{P} \seq \at{p}$ in which $\at{P} \in \Bunches(\Atoms)$ and $\at{p} \in \Atoms$. 
\end{definition}

\begin{definition}[Atomic Rule] \label{def:atomic-rule}
    An \emph{atomic rule} is a rule-figure with a (finite) set of sequents of atoms above, and a sequent of atoms below,
    \begin{equation*}
    \label{eq:define-atomic-rule}
        \infer{\at{P} \seq \at{p}}{\at{P}_1 \seq \at{p}_1 & \hdots & \at{P}_n \seq \at{p}_n}
    \end{equation*}
    In an atomic rule, the sequents above are called the \emph{premisses} and the sequent below is called the \emph{conclusion}. An atomic rule with an empty set of premisses is called an atomic axiom.
\end{definition}
Importantly, atomic rules are taken \emph{per se} and not closed under substitution; that is, for a given rule, we do not consider `instances' of it, rather it applies only to those atoms explicit in it. We may write $(\at{P}_1 \seq \at{p}_1, \dots, \at{P}_n \seq \at{p}_n) \Rightarrow (\at{P} \seq \at{p})$ to denote the atomic rule in Definition~\ref{def:atomic-rule} --- note, axioms are the cases when the left-hand side is empty.

A collection of atomic rules determines a \emph{base}. Intuitively, a base is a system of inference that is logic \emph{free} and thus suitable to ground the semantics of the logical constants.

\begin{definition}[Base] \label{def:base}
A base is a set of atomic rules.
\end{definition}

We write $\baseB, \baseC, \dots$ to denote bases, and $\emptyset$ to denote the empty base (i.e., the base with no rules). We say $\base{C}$ is an \emph{extension} of $\base{B}$ if $\baseC$ is a superset of $\baseB$; hence it is denoted as $\base{C} \supseteq \base{B}$. 

\begin{example}\label{ex:atomic-rule}
The following is an example of an atomic rule, in which $\at{p},\at{q},\at{r}, \at{s} \in \Atoms$:
\[
    \infer{\at{p} \fatsemi \at{q} \seq \at{r}}{\at{p} \seq \at{s} & \at{s} \fatsemi \at{q} \seq \at{r}}
\]
It may be expressed inline as follows: $((\at{p} \seq \at{s}), (\at{s} \fatsemi \at{q} \seq \at{r})) \Rightarrow (\at{p} \fatsemi \at{q} \seq \at{r})$. The set consisting of just this atomic rule is a base $\baseB$. 

Let $\baseC$ be the set of all atomic rules taking the same form as the above while  $\at{p}$,$\at{q}$,$\at{r}$, and $\at{s}$ range over $\At$ --- that is,
\[
 \baseC := \left\{ 
    \raisebox{-.5em}{
        \infer{\at{x} \fatsemi \at{y} \seq \at{z}}{
            \at{x} \seq \at{w} & 
            \at{w} \fatsemi \at{y} \seq \at{z}
        }
}  
\mid \at{x},\at{y},\at{z}, \at{w} \in \Atoms\right\}
\]
Observe, $\baseC \supseteq \baseB$ and $\baseB \neq \baseC$.
\end{example}

Intuitively, $\Rightarrow (\at{P} \seq \at{p})$ means that the  sequent of atoms  $(\at{P} \seq \at{p})$ may be concluded whenever, while $(\at{P}_1 \seq \at{p}_1, \dots, \at{P}_n \seq \at{p}_n) \Rightarrow (\at{P} \seq \at{p})$ means that one may derive $\at{P} \seq \at{p}$ if one has derived $(\at{P}_i \seq \at{p}_i)$ for $i=1,...,n$. However, we regard the context-formers as \emph{meta-logical} so that it is essential that derivability in a base respects \emph{coherent equivalence} (see Definition~\ref{def:coherent-equivalence}). Similarly, the additive context-former must admit weakening and contractions because that is what it means as a data constructor in BI. Hence, we arrive at the following definition of derivability in a base, which is in the sequent calculus format (cf. the natural deduction format used for \IPL{}~\cite{Sandqvist2015base} and \MILL{}~\cite{gheorghiu2023imll}):


\begin{definition}[Derivability in a Base] \label{def:derivability-in-a-base}
    Let $\base{B}$ be a base. Derivability in $\base{B}$ --- denoted $\proves[\baseB]$ --- is the smallest relation satisfying the following:
    \begin{itemize}[label=--]
    \item \myhypertarget{def:derivability-in-a-base:taut}{\textsc{taut}}. If $\at{p} \in \Atoms$, then $\at{p} \proves[\baseB] \at{p}$
    \item  \myhypertarget{def:derivability-in-a-base:initial}{\textsc{initial}}. If $\Rightarrow (\at{P} \seq \at{p})\in \base{B}$, then $\at{P} \proves[\baseB] \at{p}$
    \item \myhypertarget{def:derivability-in-a-base:rule}{\textsc{rule}}. If $(\at{P}_1 \seq \at{p}_1 , \hdots,  \at{P}_n \seq \at{p}_n) \Rightarrow (\at{P} \seq \at{p}) \in \base{B}$ and
 $\at{P}_1 \proves[\baseB] \at{p}_1, \dots, \at{P}_n \proves[\baseB] \at{p}_n$, then $\at{P} \proves[\baseB] \at{p}$
    \item  \myhypertarget{def:derivability-in-a-base:weak}{\textsc{weak}}. If $\at{P}(\at{Q}) \proves[\baseB] \at{p}$, then $\at{P}(\at{Q} \fatsemi \at{Q}') \proves[\baseB] \at{p}$ for any $\at{Q}' \in \Bunches(\Atoms)$
        \item \myhypertarget{def:derivability-in-a-base:cont}{\textsc{cont}}. If $\at{P}(\at{Q}\fatsemi \at{Q}) \proves[\baseB] \at{p}$, then $\at{P}(\at{Q}) \proves[\baseB] \at{p}$
        \item \myhypertarget{def:derivability-in-a-base:exch}{\textsc{exch}}. If $\at{P}(\at{Q}) \proves[\baseB] \at{p}$ and $\at{Q} \equiv \at{Q}'$, then $\at{P}(\at{Q}') \proves[\baseB] \at{p}$.
        \item \myhypertarget{def:derivability-in-a-base:cut}{\textsc{cut}}. If $\at{T} \proves[\baseB] \at{q}$ and $\at{S}(\at{q}) \proves[\baseB] \at{p}$, then $\at{S}(\at{T}) \proves[\baseB] \at{p}$.
    \end{itemize}
\end{definition}

\begin{example}[Example~\ref{ex:atomic-rule} cont'd] \label{ex:derivability-in-a-base}
    The expressive power of $\baseB$ and $\baseC$ is the same as $\emptyset$ because of \myhyperlink{def:derivability-in-a-base:cut} --- that is, $\at{P} \proves[\baseB] \at{p}$ iff $\at{P} \proves[\baseC] \at{p}$ iff $\at{P} \proves[\emptyset] \at{p}$. 
\end{example}

As Example~\ref{ex:derivability-in-a-base} illustrates, it is possible for a base to satisfy some of \myhyperlink{def:derivability-in-a-base:weak}, \myhyperlink{def:derivability-in-a-base:cont}, \myhyperlink{def:derivability-in-a-base:exch}, \myhyperlink{def:derivability-in-a-base:cut} using just \myhyperlink{def:derivability-in-a-base:initial} and \myhyperlink{def:derivability-in-a-base:rule}. Indeed, in Section~\ref{sec:completeness}, we construct a base $\base{N}$ that satisfies all of them. However, we require these properties to hold of arbitrary bases, and that may not be the case. For this reason, we include properties as closure conditions on derivability in a base. Another possibility is to restrict the notion of base to be exactly those sets of atomic rules that satisfies these conditions, but this option requires wanton repetition of infinitely many structural rules in each base. 


To conclude this section, we state an obvious but important fact about derivability in a base: 
\begin{proposition}\label{prop:der-monotone-on-bases}
     If $\at{T} \proves[\baseB] \at{p}$ and $\baseC \supseteq \baseB$, then $\at{T} \proves[\baseC] \at{p}$.
\end{proposition} 
\begin{proof}
    Follows by induction on the derivation of $\at{T} \proves[\baseB] \at{p}$ since all the rules of $\baseB$ are contained in $\baseC$.
\end{proof}

\subsection{Support in a Base} \label{sec:support}

Derivability in a base is the starting point for the B-eS of \BI{}. That is, it gives the semantics of the non-logical parts of BI --- namely, the atoms. The semantics of the logical constants is expressed in terms of a support judgement that unfolds according the clauses, analogous to the satisfaction relation used in model-theoretic semantics. This section is about defining the support judgement.

Intuitively, support (in a base) for \BI{} combines the notion of support (in a base) for \IPL{}~\cite{Sandqvist2015base} and for \IMLL{}~\cite{gheorghiu2023imll}. However, the use of bunches as data-structures renders it more complex than in either of the preceding works. As in the B-eS for \IMLL{}, to handle substructurality, support carries a collection of atoms thought of as \emph{resources}. In \BI{}, that collection is a \emph{bunch}. Hence, it is necessary to accommodate the particular idiosyncrasies of this data structure. To this end, we use \emph{bunch-extensions}, which appear elsewhere in the semantics and proof theory of \BI{}, but without uniform treatment --- see, for example, Gheorghiu et al.~\cite{Alex:Samsonschrift}, where it is in-advisably called \emph{weak coherence}.

\begin{definition}[Bunch-extension]\label{def:bunch-stronger-than}
    The \emph{bunch-extension} relation $\bunchStrongerThan$ is the least relation satisfying: 
    \begin{itemize}[label=--]
        \item if $\varGamma \equiv \varGamma'[\varDelta \mapsto (\varDelta \fatsemi \varDelta')]$, then $\varGamma \bunchStrongerThan \varGamma'$
        \item if $\varGamma \bunchStrongerThan \varGamma'$ and $\varGamma' \bunchStrongerThan\varGamma''$, then $\varGamma \bunchStrongerThan\varGamma''$.
    \end{itemize}
\end{definition}

\begin{example} \label{ex:bunchstrongerthan}
       Observe $((\at{p} \mmcomma \at{q}) \aacomma \at{u} ) \mcomma \at{r} \bunchStrongerThan (\at{p} \mmcomma \at{q}) \mcomma \at{r}$, since $((\at{p} \mmcomma \at{q}) \aacomma \at{u} ) \mcomma \at{r} \equiv ((\at{p} \mmcomma \at{q}) \mcomma \at{r})[(\at{p} \mmcomma \at{q}) \mapsto ((\at{p} \mmcomma \at{q}) \aacomma \at{u})]$. Moreover,  $(\at{p} \mmcomma \at{q})  \mcomma \at{r} \bunchStrongerThan \at{p} \mmcomma (\at{q}  \mcomma \at{r})$, since $(\at{p} \mmcomma \at{q})  \mcomma \at{r} \equiv \at{p} \mmcomma (\at{q}  \mcomma \at{r})$. Hence, $((\at{p} \mmcomma \at{q}) \aacomma \at{u} ) \mcomma \at{r} \bunchStrongerThan \at{p} \mmcomma (\at{q}  \mcomma \at{r})$.
\end{example} 

\begin{proposition}[Compositionality of Bunch-extension] 
\label{prop:bunch-weaker-compositional}
    If $\at{T} \bunchWeakerThan \at{S}$, then $\at{P(T)} \bunchWeakerThan \at{P(S)}$.
\end{proposition}
\begin{proof}
    It is easy to see that $\varGamma \bunchStrongerThan \varGamma'$ iff $\varGamma \equiv \varGamma'\sigma_1...\sigma_n$, where $\sigma_i$ is a substitution of the form $[\varDelta \mapsto \varDelta \acomma \varDelta']$. We proceed by induction on $\holeDepth{\at{P}(\cdot)}$.
    \begin{itemize}[label=--]
        \item \textsc{base case.} $\holeDepth{\at{P}(\cdot)} =0$. By Definition~\ref{def:holeDepth}, $\at{P}(\cdot)=(\cdot)$. Hence, the antecedent of the claim is identical to the consequent.
        \item  \textsc{inductive step.} $\holeDepth{\at{P}(\cdot)} =d+1$ for some $d>0$. By Definition~\ref{def:holeDepth}, $\at{P}(\cdot)= \at{P}'(\cdot)\commavar \at{Q}$, where $\holeDepth{\at{P}'(\cdot)} = d$, $\at{Q} \in \Bunches(\Atoms)$, and $\commavar \in \{\aacomma, \mmcomma\}$. By the induction hypothesis (IH), $\at{P}'(\at{T}) \bunchWeakerThan \at{P}'(\at{S})$.  Let $\tau_1,...,\tau_k$ witness the extension. Since $\at{P}(\at{S}) \equiv  \at{P}'(\at{S})\commavar \at{Q} \equiv \at{P}'(\at{T})\tau_1...\tau_k\commavar \at{Q} \equiv (\at{P}'(\at{T})\commavar \at{Q})\tau_1...\tau_k\equiv \at{P}(\at{T})\tau_1...\tau_k$, we have $\at{P(\at{S})} \bunchStrongerThan \at{P(\at{T})}$. 
    \end{itemize}
    This completes the induction.
\end{proof}

Having defined derivability in a base, and base- and bunch-extension, it is possible to give the semantics.

\begin{definition}[Support in a Base]
\label{def:support} 
    Supports in the base $\base{B}$ relative to either a contextual bunch or bunch of atoms $\at{P}$ is the relation $\supp{\at{P}}{\base{B}}$ defined by the clauses of Figure~\ref{fig:suppBI} in which $\at{p} \in \Atoms$, $\phi, \psi \in \Formulas$, $\at{S} \in \Bunches(\Atoms)$, $\at{R}(\cdot) \in \BunchesWithHole(\At)$ and $\varGamma, \varDelta \in \Bunches(\Formulas)$.
\end{definition}

It is easy to see that Figure~\ref{fig:suppBI} is an inductive definition on a structure of formulas that prioritizes the conjunctions and disjunction (i.e., $\mand$, $\land$, and $\lor$) over the implications ($\mto$ and $\to$); this follows similar treatments in the B-eS of IPL~\cite{Sandqvist2015hypothesis} and IMLL~\cite{gheorghiu2023imll}. The bunch-extension relation is used in \eqref{cl:BI-BeS-supp:acomma} and \eqref{cl:BI-BeS-supp:mcomma}. As explained in Section~\ref{sec:introduction}, the purpose of the bunches of atoms indexing the support relation is to express the susbtructurality of the logical constants and delivers a resource interpretation of the logic. This is made clear in application of BI via the B-eS to access control in Section~\ref{sec:eg-access-control}.

\begin{figure}[t]
    \hrule
    \begin{alignat}{3}
        & \supp{\at{S}}{\base{B}} \at{p} \quad 
        && \text{iff} \quad 
        && \at{S} \proves[\base{B}]\at{p} 
        \tag{At} \label{cl:BI-BeS-supp:At} \\
        %
        & \supp{\at{S}}{\base{B}} \phi \land \psi \quad
        && \text{iff} \quad 
        && \forall \base{X} \baseGeq \base{B}, \forall \at{U}(\cdot) \in \BunchesWithHole(\At), \forall \at{p} \in \At,  \text{ if } \phi \addcontext \psi\supp{\at{U}(\cdot)}{\base{X}} \at{p}, \text{ then} \supp{\at{U(S)}}{\base{X}} \at{p} 
        \tag{$\land$} \label{cl:BI-BeS-supp:conjunction}   \\
        %
        & \supp{\at{S}}{\base{B}} \phi \mand \psi \quad
        && \text{iff} \quad
        && \forall \base{X} \baseGeq \base{B}, \forall \at{U}(\cdot) \in \BunchesWithHole(\At), \forall \at{p} \in \At,  \text{ if } \phi \multcontext \psi\supp{\at{U}(\cdot)}{\base{X}} \at{p}, \text{ then} \supp{\at{U(S)}}{\base{X}} \at{p} 
        \tag{$\mand$} \label{cl:BI-BeS-supp:mand} \\
        %
        & \supp{\at{S}}{\base{B}} \phi \lor \psi \quad 
        && \text{iff} \quad 
        && \forall \base{X} \baseGeq \base{B}, \forall \at{U}(\cdot) \in \BunchesWithHole(\At), \forall \at{p} \in \At, \notag
        \\ 
        & && 
        && \text{if } \phi \supp{\at{U}(\cdot)}{\base{X}} \at{p} \text{ and } \psi \supp{\at{U(\cdot)}}{\base{X}} \at{p}, \text{ then }  \supp{\at{U(S)}}{\base{X}} \at{p} 
        \tag{$\lor$} \label{cl:BI-BeS-supp:disjunction} \\
        %
        & \supp{\at{S}}{\base{B}} \phi \to \psi \quad
        && \text{iff} 
        && \phi \supp{\at{S} \aacomma (\cdot)}{\base{B}} \psi \tag{$\to$} \label{cl:BI-BeS-supp:implication} \\
        %
        & \supp{\at{S}}{\base{B}} \phi \wand \psi \quad
        && \text{iff} \quad 
        && \phi \supp{\at{S} \mmcomma(\cdot)}{\base{B}} \psi 
        \tag{$\wand$} \label{cl:BI-BeS-supp:wand} \\
        %
        & \supp{\at{S}}{\base{B}} \bot \quad
        && \mbox{iff} \quad 
        && \forall \at{U}(\cdot) \in \BunchesWithHole(\At), \forall \at{p} \in \At, \supp{\at{U(S)}}{\baseB} \at{p} 
        \tag{$\bot$} \label{cl:BI-BeS-supp:bot} \\
        %
        & \supp{\at{S}}{\base{B}} \top \quad
        && \text{iff} \quad
        && \forall \base{X} \baseGeq \base{B}, \forall \at{U(\cdot)} \in \BunchesWithHole(\At), \forall \at{p} \in \At, \text{ if } \supp{\at{U(\aunit)}}{\baseX} \at{p}, \text{ then} \supp{\at{U(S)}}{\base{X}} \at{p} 
        \tag{$\top$} \label{cl:BI-BeS-supp:top} \\
        %
        & \supp{\at{S}}{\base{B}} \mtop \quad
        && \text{iff} \quad 
        && \forall \base{X} \baseGeq \base{B}, \forall \at{U(\cdot)} \in \BunchesWithHole(\At), \forall \at{p} \in \At, \text{ if } \supp{\at{U(\munit)}}{\baseX} \at{p}, \text{ then} \supp{\at{U(S)}}{\base{X}} \at{p} 
        \tag{$\mtop$} \label{cl:BI-BeS-supp:mtop} \\
        %
        \varGamma & \supp{\at{R(\cdot)}}{\base{B}} \psi \quad
        && \text{iff} \quad 
        && \forall \base{X} \baseGeq \base{B}, \forall \at{U} \in \Bunches(\At), \text{ if } \supp{\at{U}}{\baseX} \varGamma, \text{ then } \supp{\at{R(U)}}{\baseX} \psi 
        \tag{Inf} \label{cl:BI-BeS-supp:inf}  \\ 
    & \supp{\at{S}}{\baseB} \munit \hspace{2.5em}
    && \text{iff} \quad 
    && \at{S} \bunchStrongerThan \munit 
    \tag{$\munit$} \label{cl:BI-BeS-supp:munit} \\
    & \supp{\at{S}}{\baseB} \aunit \quad
    && 
    && \text{always} 
    \tag{$\aunit$} \label{cl:BI-BeS-supp:aunit} \\ 
    %
    & \supp{\at{S}}{\base{B}} \varGamma \fatcomma \varDelta \quad
    && \text{iff} \quad 
    && \exists \at{Q_1}, \at{Q_2} \in \Bunches(\At) \text{ such that } \at{S} \bunchStrongerThan \at{Q_1} \mcomma \at{Q_2},  \supp{\at{Q_1}}{\baseB} \varGamma, \supp{\at{Q_2}}{\baseB} \varDelta  \tag{$\mmcomma$} \label{cl:BI-BeS-supp:mcomma} \\ 
    %
    & \supp{\at{S}}{\baseB} \varGamma \acomma \varDelta \quad 
    && \text{iff} \quad
    && \exists \at{Q_1}, \at{Q_2} \in \Bunches(\At) \text{ such that } \at{S} \bunchStrongerThan \at{Q_1}, \at{S} \bunchStrongerThan \at{Q_2}, \supp{\at{Q_1}}{\baseB} \varGamma, \supp{\at{Q_2}}{\baseB} \varDelta 
    \tag{$\aacomma$} \label{cl:BI-BeS-supp:acomma}
    \end{alignat}
    \hrule
    \caption{Support for \BI{}}
    \label{fig:suppBI} 
\end{figure}

\begin{definition}[Validity] \label{def:validity}
    A sequent $\varGamma \seq \phi$ is \emph{valid} iff  $\varGamma \entails \phi$, where
    \[
    \varGamma \entails \phi \qquad \mbox{iff} \qquad \mbox{for any base $\base{B}$,} \quad \varGamma \supp{(\cdot)}{\baseB} \phi  
    \]
\end{definition}

\begin{example}
    It is instructive to see how $\supp{((\at{p} \mmcomma \at{q}) \aacomma \at{u}) \mmcomma \at{r} }{\emptybase} \at{p} \mcomma (\at{q} \mand \at{r}) $ obtains. From Example~\ref{ex:bunchstrongerthan}, $((\at{p} \mcomma \at{q}) \acomma \at{u} ) \mcomma \at{r} \bunchStrongerThan \at{p} \mcomma (\at{q}  \mcomma \at{r})$. Hence, by \eqref{cl:BI-BeS-supp:mcomma}, it suffices to show the following: (1) $\supp{\at{p}}{\emptybase} \at{p}$, and (2) $\supp{ \at{q} \mmcomma \at{r}}{\emptyset}  \at{q} \mand \at{r}$. That (1) obtains is immediate by \eqref{cl:BI-BeS-supp:At} and~\myhyperlink{def:derivability-in-a-base:taut} in Definition~\ref{def:derivability-in-a-base}. By \eqref{cl:BI-BeS-supp:mand}, claim $(2)$ is equivalent to the following:
    \begin{equation}
    \label{eq:ex-bunch-extension}
    \forall \base{X} \baseGeq \emptybase, \forall \at{U}(\cdot) \in \BunchesWithHole(\At), \forall \at{s} \in \At,  \text{ if } \at{q} \mcomma \at{r} \supp{\at{U}(\cdot)}{\base{X}} \at{s}, \text{ then} \supp{\at{U}(\at{q} \mmcomma \at{r})}{\base{X}} \at{s}  \tag{$\dagger$}
    \end{equation}
    It is clear that we have $\supp{\at{q}}{\baseX} \at{q}$ and $\supp{\at{r}}{\baseX} \at{r}$ for any $\baseX$. Therefore, by \eqref{cl:BI-BeS-supp:mcomma}, $\supp{\at{q} \mmcomma \at{r}}{\baseX} \at{q} \mcomma \at{r}$. Hence, the result follows by \eqref{cl:BI-BeS-supp:inf} on the antecedent of \eqref{eq:ex-bunch-extension} by choosing $U(\cdot) = (\cdot)$.  We see in this example why $\bunchWeakerThan$ must apply to sub-bunches and not merely at the top-level. \end{example}

The main result of this paper is to show that consequence for BI is the same as validity in the B-eS --- that is, $\varGamma \proves \phi \mbox{ iff } \varGamma \supp{}{} \phi$. This is captured in Theorem~\ref{thm:BI-BeS-soundness} (Soundness) and Theorem~\ref{thm:BI-BeS-completeness} (Completeness), which are stated and proved in Section~\ref{sec:soundness} and Section~\ref{sec:completeness}, respectively.

\subsection{Preliminary Results}\label{sec:preliminary}
Before proceeding to soundness and completeness, it is instructive to observe that some properties about support.

\begin{proposition}[Monotonicity] \label{prop:monotonicity}
If $\supp{\at{T}}{\baseB} \phi$ and $\baseC \supseteq \baseB$ and $\at{S} \bunchStrongerThan \at{T}$, then $\supp{\at{S}}{\baseC} \phi$. 
\end{proposition}

This follows as the combination of Proposition~\ref{prop:der-monotone-on-bases} (above), Lemma~\ref{lem:supp-formula-monotone-on-bases}, and Lemma~\ref{lem:supp-formula-monotone-on-resource}: 

\begin{lemma} \label{lem:supp-formula-monotone-on-bases}
     If $\supp{\at{T}}{\baseB} \phi$ and $\baseC \supseteq \baseB$, then $\supp{\at{T}}{\baseC} \phi$. 
\end{lemma}
\begin{proof}
    We proceed by induction on the structure of $\phi$.
    \begin{itemize}[label=--]
        \item \textsc{base case}. This follows from Proposition~\ref{prop:der-monotone-on-bases}.
    \item \textsc{inductive step.} If there is no base-extension in the relevant clause for $\phi$, then the result follows from applying the  induction hypothesis (IH) to the judgements in the definiens. For example, consider the case where $\phi$ is $\bot$: by \eqref{cl:BI-BeS-supp:bot}, from $\supp{\at{T}}{\baseB} \phi$ infer that $\supp{\at{T}}{\baseB} \at{p}$ for every $\at{p \in \Atoms}$; hence, by the IH, $\supp{\at{T}}{\baseC} \at{p}$ for every $\at{p \in \Atoms}$; whence,  by \eqref{cl:BI-BeS-supp:bot}, $\supp{\at{T}}{\baseC} \bot$, as required. 

    If base-extension is used  in the relevant clause for $\phi$, the result holds immediately since $\baseX \baseGeq \baseC$ implies $\baseX \baseGeq \baseB$. For example, consider the case where $\phi$ is $\psi_1 \mand \psi_2$. By \eqref{cl:BI-BeS-supp:mand}, 
    \[
        \forall \base{X} \baseGeq \base{B}, \forall \at{U}(\cdot) \in \BunchesWithHole(\Formulas), \forall \at{p} \in \Atoms,  \text{ if } \psi_1 \multcontext \psi_2 \supp{\at{U}(\cdot)}{\base{X}} \at{p}, \text{ then } \supp{\at{U(T)}}{\base{X}} \at{p} 
    \]
Hence, since $\baseC \baseGeq \baseB$,
    \[
        \forall \base{X} \baseGeq \base{C}, \forall \at{U}(\cdot) \in \BunchesWithHole(\Formulas), \forall \at{p} \in \Atoms,  \text{ if } \psi_1 \multcontext \psi_2 \supp{\at{U}(\cdot)}{\base{X}} \at{p}, \text{ then } \supp{\at{U(T)}}{\base{X}} \at{p} 
    \]
    By \eqref{cl:BI-BeS-supp:mand}, $\supp{\at{T}}{\baseC} \psi_1 \mand \psi_2$, as required.

   For clarity, we also show the case where $\varphi$ is $\psi_1 \mto \psi_2$. Assume $\suppM{\baseB}{\at{T}}\varphi \mto \psi$ and some $\baseC \baseGeq \baseB$. We want to show $\suppM{\baseC}{\at{T}} \varphi \mto \psi$. Therefore, fix some arbitrary $\baseD \baseGeq \baseC$ and $\at{U}$ such that $\suppM{\baseD}{\at{U}} \varphi$. The goal is to show $\suppM{\baseD}{\at{T} \mmcomma \at{U}} \psi$. This follows immediately from $\baseD \baseGeq \baseC \baseGeq \baseB$ and $\varphi \suppM{\baseB}{\at{T}} \psi$.
    \end{itemize}
This completes the induction.
\end{proof}

\begin{lemma}\label{lem:supp-formula-monotone-on-resource}
    If $\supp{\at{T}}{\baseB} \phi$ and $\at{S} \bunchStrongerThan \at{T}$, then $\supp{\at{S}}{\baseB} \phi$. 
\end{lemma}
\begin{proof}
    We proceed by induction on the structure of $\phi$. 
    \begin{itemize}[label=--]
        \item $\phi = \at{p} 
 \in \Atoms$. By \eqref{cl:BI-BeS-supp:At}, $\supp{\at{T}}{\baseB} \at{p}$ iff $\at{T} \proves[\baseB] \at{p}$. By,~\myhyperlink{def:derivability-in-a-base:weak} in Definition~\ref{def:derivability-in-a-base}, since $\at{S} \bunchStrongerThan \at{T}$, it follows that $\at{S} \proves[\baseB] p$. 
        \item $\phi = \phi_1 \land \phi_2$. By \eqref{cl:BI-BeS-supp:conjunction} it suffices to show the following:
        \[
        \mbox{$\forall \baseX \baseGeq \baseB$, $\forall \at{U}(\cdot) \in \BunchesWithHole(\Formulas)$, $\forall \at{p} \in \Atoms$, if $\phi_1 \acomma \phi_2 \supp{\at{U}(\cdot)}{\baseX} \at{p}$, then $\supp{\at{U(S)}}{\baseX} \at{p}$}
        \]
         Therefore, let $\baseC \baseGeq \baseB$, $\at{P}(\cdot)$, and $\at{q}$ be arbitrary  such that $\phi_1 \acomma \phi_2 \supp{\at{P}(\cdot)}{\baseC} \at{q}$.  We require to show that $\supp{\at{P(S)}}{\baseC} \at{q}$. 
        
        By \eqref{cl:BI-BeS-supp:conjunction}, from the assumption $\supp{\at{T}}{\baseB} \phi_1 \land \phi_2$, infer $\supp{\at{P(T)}}{\baseC} \at{q}$. By Proposition~\ref{prop:bunch-weaker-compositional}, from $\at{S} \bunchStrongerThan \at{T}$, infer $\at{P(S) }\bunchStrongerThan \at{P(T)}$. By the induction hypothesis (IH), from $\supp{\at{P(T)}}{\baseC} \at{q}$ and $\at{P(S)} \bunchStrongerThan \at{P(T)}$, infer $\supp{\at{P(S)}}{\baseC} \at{q}$. 
        
        \item $\phi = \top$. By \eqref{cl:BI-BeS-supp:top} it suffices to show the following:
        \[
       \mbox{$\forall \baseX \baseGeq\baseB$, $\forall \at{U}(\cdot) \in \BunchesWithHole(\Formulas)$, $\forall \at{p} \in \Atoms$, if $\supp{\at{U}(\aunit)}{\baseX} \at{p}$, then $\supp{\at{U(T)}}{\baseX} \at{p}$}
        \]
        Therefore, let $\baseC \baseGeq \baseB$, $\at{P}(\cdot)$, and $\at{q}$ be arbitrary such that $\supp{\at{P}(\aunit)}{\baseC} \at{q}$. We require to show $\supp{\at{P(S)}}{\baseC} \at{q}$ holds. 

        By \eqref{cl:BI-BeS-supp:top}, from the assumption, $\supp{\at{T}}{\baseB} \top$, infer 
         $\supp{\at{P(T)}}{\baseC} \at{q}$. By Proposition~\ref{prop:bunch-weaker-compositional}, from $\at{S} \bunchStrongerThan \at{T}$, infer $\at{P(S) }\bunchStrongerThan \at{P(T)}$. By the IH, from $\supp{\at{P(T)}}{\baseC} \at{q}$ and $\at{P(S)} \bunchStrongerThan \at{P(T)}$, infer $\supp{\at{P(S)}}{\baseC} \at{q}$.

        \item $\phi = \phi_1 \to \phi_2$. By \eqref{cl:BI-BeS-supp:implication} it suffices to show $\phi_1 \supp{\at{S} \aacomma (\cdot)}{\baseB} \phi_2$. By \eqref{cl:BI-BeS-supp:inf}, this is equivalent to the following:
        \[
        \mbox{$\forall \baseX \baseGeq \baseB$, $\forall \at{U} \in \Bunches(\Formulas)$, if $\supp{\at{U}}{\baseX} \phi_1$, then $\supp{\at{S}\aacomma \at{U}}{\baseX} \phi_2$.}
        \]
        Therefore, let $\baseC \baseGeq \baseB$, $\at{P}$ be arbitrary such that $\supp{\at{P}}{\baseC} \phi_1$. We require to show $\supp{\at{S}\aacomma \at{P}}{\baseC} \phi_2$. 

        By \eqref{cl:BI-BeS-supp:implication} and $\eqref{cl:BI-BeS-supp:inf}$, from the assumption $\supp{\at{T}}{\baseB} \phi_1 \to \phi_2$, infer 
         $\supp{\at{T}\aacomma \at{P}}{\baseC} \phi_2$. By Proposition~\ref{prop:bunch-weaker-compositional}, from $\at{S} \bunchStrongerThan \at{T}$, infer $\at{S}\aacomma \at{P} \bunchStrongerThan \at{T} \aacomma \at{P}$. By the IH, from $\supp{\at{T} \aacomma \at{P}}{\baseC} \phi_2$ and $\at{S}\aacomma \at{P} \bunchStrongerThan \at{T} \aacomma \at{P}$, infer $\supp{\at{S}\aacomma \at{P}}{\baseC} \phi_2$.
         
        \item $\phi = \phi_1 \mand \phi_2$. \emph{Mutatis mutandis} 
 on the $\land$-case. 
        \item $\phi = \mtop$. \emph{Mutatis mutandis} 
 on the $\top$-case. 
        \item $\phi = \phi_1 \mto \phi_2$.\emph{Mutatis mutandis} 
 on the $\to$-case. 
        \item $\phi = \phi_1 \lor \phi_2$. \emph{Mutatis mutandis} 
 on the $\land$-case. 
        \item $\phi = \bot$. By \eqref{cl:BI-BeS-supp:bot} on the assumption $\supp{\at{T}}{\baseB} \bot$, infer  $\supp{\at{T}}{\baseB} \at{p}$ holds for all $\at{p} \in \At$. By either the IH or the atomic case, $\supp{\at{S}}{\baseB} \at{p}$ for all $\at{p} \in \At$. Hence, by \eqref{cl:BI-BeS-supp:bot}, $\supp{\at{S}}{\baseB} \bot$. 
    \end{itemize}
    This completes the induction. 
\end{proof}

It is easy to see that monotonicity (Proposition~\ref{prop:monotonicity}) extends to bunches:

\begin{lemma} \label{lem:supp-bunch-monotone-on-bases}
     If $\supp{\at{T}}{\baseB} \varGamma$ and $\baseC \supseteq \baseB$, then $\supp{\at{T}}{\baseC} \varGamma$. 
\end{lemma}
\begin{proof}
    We proceed by induction on the structure of $\varGamma$.
    \begin{itemize}[label=--]
        \item $\varGamma = \phi \in \Formulas$. This is Lemma~\ref{lem:supp-formula-monotone-on-bases}.
        \item $\varGamma = \aunit$. This follows immediately by \eqref{cl:BI-BeS-supp:aunit}.
        \item $\varGamma = \munit$. By \eqref{cl:BI-BeS-supp:munit}, $\at{T}  \bunchStrongerThan \munit$. Hence, by \eqref{cl:BI-BeS-supp:munit}, $\supp{\at{T}}{\baseC} \munit$.
        \item $\varGamma = \varGamma_1 \acomma \varGamma_2$. By \eqref{cl:BI-BeS-supp:acomma}, there are $\at{Q}_1$ and $\at{Q}_2$ such that $\at{T} \bunchStrongerThan \at{Q}_1$, $\at{T} \bunchStrongerThan \at{Q}_1$, $\supp{\at{Q}_1}{\baseB} \varGamma_1$, and $\supp{\at{Q}_2}{\baseB} \varGamma_2$. By the induction hypothesis (IH), $\supp{\at{Q}_1}{\baseC} \varGamma_1$ and $\supp{\at{Q}_2}{\baseC} \varGamma_2$. Whence, by \eqref{cl:BI-BeS-supp:acomma}, $\supp{\at{T}}{\baseC} \varGamma_1 \acomma \varGamma_2$.
        \item $\varGamma = \varGamma_1 \mcomma \varGamma_2$. By \eqref{cl:BI-BeS-supp:mcomma}, there are $\at{Q}_1$ and $\at{Q}_2$ such that $\at{T} \bunchStrongerThan (\at{Q}_1 \mcomma \at{Q}_2)$, $\supp{\at{Q}_1}{\baseB} \varGamma_1$, and $\supp{\at{Q}_2}{\baseB} \varGamma_2$. By the IH, $\supp{\at{Q}_1}{\baseC} \varGamma_1$ and $\supp{\at{Q}_2}{\baseC} \varGamma_2$. Whence, by \eqref{cl:BI-BeS-supp:mcomma}, $\supp{\at{T}}{\baseC} \varGamma_1 \mcomma \varGamma_2$.
    \end{itemize}
    This completes the induction.
\end{proof}

\begin{lemma}
\label{lem:supp-bunch-monotone-on-resource}
    If $\supp{\at{T}}{\baseB} \varGamma$ and $\at{S} \bunchStrongerThan \at{T}$, then $\supp{\at{S}}{\baseB} \varGamma$. 
\end{lemma}
\begin{proof}
    We proceed by case-analysis on the structure of $\varGamma$.
    \begin{itemize}[label=--]
        \item $\varGamma = \phi \in \Formulas$. This is Lemma~\ref{lem:supp-formula-monotone-on-resource}.
        \item $\varGamma = \aunit$. This follows immediately by \eqref{cl:BI-BeS-supp:aunit}.
        \item $\varGamma = \munit$. By \eqref{cl:BI-BeS-supp:munit}, $\at{T}  \bunchStrongerThan \munit$. Hence,  $\at{S} \bunchStrongerThan \munit$. Whence, \eqref{cl:BI-BeS-supp:munit}, $\supp{\at{S}}{\baseB}\munit$.
        \item $\varGamma = \varGamma_1 \acomma \varGamma_2$. By \eqref{cl:BI-BeS-supp:acomma}, there are $\at{Q}_1$ and $\at{Q}_2$ such that $\at{T} \bunchStrongerThan \at{Q}_1$, $\at{T} \bunchStrongerThan \at{Q}_1$, $\supp{\at{Q}_1}{\baseB} \varGamma_1$, and $\supp{\at{Q}_2}{\baseB} \varGamma_2$. From $\at{S} \bunchStrongerThan \at{T}$, infer $\at{S} \bunchStrongerThan \at{Q}_1$ and $\at{S} \bunchStrongerThan \at{Q}_2$.  Hence, by \eqref{cl:BI-BeS-supp:acomma}, $\supp{\at{S}}{\baseB} \varGamma_1 \acomma \varGamma_2$.
        \item $\varGamma = \varGamma_1 \mcomma \varGamma_2$. By \eqref{cl:BI-BeS-supp:mcomma}, there are $\at{Q}_1$ and $\at{Q}_2$ such that $\at{T} \bunchStrongerThan (\at{Q}_1 \mcomma \at{Q}_2)$, $\supp{\at{Q}_1}{\baseB} \varGamma_1$, and $\supp{\at{Q}_2}{\baseB} \varGamma_2$. From $\at{S} \bunchStrongerThan \at{T}$, infer $\at{S} \bunchStrongerThan (\at{Q}_1 \mcomma \at{Q}_2)$. Hence, by \eqref{cl:BI-BeS-supp:mcomma}, $\supp{\at{S}}{\baseB} \varGamma_1 \mcomma \varGamma_2$.
    \end{itemize}
    This completes the induction.
\end{proof}

As a corollary of monotonicity, validity corresponds to support in the empty base with respect to no resources.

\begin{lemma}
    $\varGamma \entails \phi$ iff $\varGamma \supp{(\cdot)}{\emptybase} \phi$
\end{lemma}
\begin{proof}
The `if' direction is follows from Proposition~\ref{prop:monotonicity}. The `only if' direction is immediate by Definition~\ref{def:validity}.
\end{proof}

We now proceed to showing that support respects the intended structurality of bunches.

\begin{proposition}\label{prop:support-identity}
    For any base $\baseB$ and $\at{P} \in \Bunches(\At)$, $\supp{\at{P}}{\baseB} \at{P}$ obtains.
\end{proposition} 
\begin{proof}
     We proceed by induction on the structure of $\at{P}$.
\begin{itemize}[label=--]
    \item $\at{P}$ is in $\Atoms$. That is, $\at{P}$ is some atom $\at{p}$. By \eqref{cl:BI-BeS-supp:At}, this is equivalent to $\at{p} \proves[\baseB] \at{p}$, which obtains by~\myhyperlink{def:derivability-in-a-base:taut} in Definition~\ref{def:derivability-in-a-base}.
    \item $\at{P}$ is $\aunit$. This is immediate by \eqref{cl:BI-BeS-supp:aunit}. 
    \item $\at{P}$ is $\munit$. This is immediate by \eqref{cl:BI-BeS-supp:munit} as $\munit \bunchStrongerThan \munit$. 
    \item $\at{P}$ is $\at{P}_1 \acomma \at{P}_2$.  By the induction hypothesis (IH), both $\supp{\at{P}_1}{\baseB}{\at{P}_1}$ and $\supp{\at{P}_2}{\baseB}{\at{P}_2}$ obtain.  Since $\at{P}_1 \acomma \at{P}_2 \bunchStrongerThan \at{P}_1$ and $\at{P}_1 \acomma \at{P}_2 \bunchStrongerThan \at{P}_2$, the result follows by \eqref{cl:BI-BeS-supp:acomma}.
    \item $\at{P}$ is $\at{P}_1 \mcomma \at{P}_2$. By the IH, both $\supp{\at{P}_1}{\baseB}{\at{P}_1}$ and $\supp{\at{P}_2}{\baseB}{\at{P}_2}$ obtain. Since $\at{P}_1 \mcomma \at{P}_2 \bunchStrongerThan \at{P}_1 \mcomma \at{P}_2$, the result follows by \eqref{cl:BI-BeS-supp:mcomma}.
\end{itemize}
This completes the induction.
\end{proof}
The following lemmas show that the structurality of bunches in BI is adequately captured in the support judgement:

\begin{proposition}
\label{prop:supp-resource-weaken-contract-deep}
       For any base $\baseB$, $\at{S}, \at{S}' \in \Bunches(\At)$, $\at{P}(\cdot) \in \BunchesWithHole(\At)$,  and $\phi \in \Formulas$,
       \[
        \mbox{\myhypertarget{item:supp-resource-weaken-deep}{(i)} if $\supp{\at{P}(\at{S})}{\baseB} \phi$, then $\supp{\at{P}(\at{S}\aacomma \at{S}')}{\baseB} \phi$} \qquad \mbox{\myhypertarget{item:supp-resource-contract-deep}{(ii)} 
        if $\supp{\at{P}(\at{S}\aacomma \at{S})}{\baseB} \phi$, then $\supp{\at{P}(\at{S})}{\baseB} \phi$.} 
    \]
\end{proposition} 
%
\begin{proof}
   Since $\at{P}(\at{S}) \bunchWeakerThan \at{P}(\at{S} \acomma \at{S}')$, claim~\myhyperlink{item:supp-resource-weaken-deep} follows from Proposition~\ref{prop:monotonicity}. It remains to show claim~\myhyperlink{item:supp-resource-contract-deep}. We proceed by induction on $\phi$.
   \begin{itemize}[label=--]
       \item \textsc{base case.} This follows immediately by~\myhyperlink{def:derivability-in-a-base:cont} in Definition~\ref{def:derivability-in-a-base}.
       \item \textsc{inductive step.} We demonstrate the case for $\phi = \phi_1 \lor \phi_2$ and $\phi = \phi_1 \to \phi_2$, the others being similar. 
    \begin{itemize}[label=--]
    
        \item Assume $\supp{\at{P}(\at{S}\aacomma \at{S})}{\baseB} \phi_1 \lor \phi_2$. By \eqref{cl:BI-BeS-supp:disjunction}, the consequent of~\myhyperlink{item:supp-resource-contract-deep} is equivalent to the following:
        \[
        \forall \base{X} \baseGeq \base{B}, \forall \at{U}(\cdot) \in \BunchesWithHole(\At), \forall \at{p} \in \At, \text{if } \phi_1 \supp{\at{U}(\cdot)}{\base{X}} \at{p} \text{ and } \phi_2 \supp{\at{U}(\cdot)}{\base{X}} \at{p}, \text{ then} \supp{\at{U(P(S))}}{\base{X}} \at{p} 
        \]
    Let $\at{V}(\cdot) \in \BunchesWithHole(\At)$ and $\baseC \supseteq \baseB$ and $\at{q} \in \Atoms$ be arbitrary such that $\phi_1 \supp{\at{V}(\cdot)}{\base{C}} \at{q}$ and $\phi_2 \supp{\at{V}(\cdot)}{\base{C}} \at{q}$. We require to show $\supp{\at{V(P(S))}}{\base{C}} \at{q}$.  By \eqref{cl:BI-BeS-supp:disjunction} on the assumption, we have $\supp{\at{V(P(S \aacomma S))}}{\base{C}} \at{q}$. The result follows immediately by~\myhyperlink{def:derivability-in-a-base:cont} in Definition~\ref{def:derivability-in-a-base}.

    \item Assume $\supp{\at{P}(\at{S}\aacomma \at{S})}{\baseB} \phi_1 \to \phi_2$. By \eqref{cl:BI-BeS-supp:implication} and \eqref{cl:BI-BeS-supp:inf}, the consequent of~\myhyperlink{item:supp-resource-contract-deep} is equivalent to the following:
    \[
    \mbox{$\forall \baseX \baseGeq \baseB, \forall \at{U}(\cdot) \in \BunchesWithHole(\At), \forall \at{p} \in \At,$ if $\supp{\at{U}}{\baseX} \phi$, then $\supp{\at{P(S)}\aacomma\at{U}}{\baseX} \psi$}
    \]
    Let $\at{V} \in \Bunches(\At)$ and $\baseC\supseteq \baseB$ be arbitrary such that $\supp{\at{V}}{\baseC} \phi_1$. We require to show $\supp{\at{P(S)} \aacomma \at{V}}{\base{C}} \phi_2$. 
    By \eqref{cl:BI-BeS-supp:implication} on the assumption, we have  $\supp{\at{P(S \aacomma S)}\aacomma \at{V}}{\base{C}} \phi_2$. The result follows from the induction hypothesis (IH).
    \end{itemize}
This completes the case analysis. Observe that in the $\lor$-case we did not require the IH.
    \end{itemize}
This completes the induction. 
\end{proof}

\begin{proposition}
\label{lem:supp-bunch-weaken-contract-deep}
    For any base $\baseB$, $\at{S} \in \Bunches(\At)$, $\varDelta, \varDelta' \in \Bunches(\Formulas)$, $\varGamma(\cdot) \in \BunchesWithHole(\Formulas)$, 
    \[
    \mbox{\myhypertarget{item:supp-bunch-weaken-deep}{(i)} if $\supp{\at{S}}{\baseB} \varGamma(\varDelta \acomma \varDelta')$, then $\supp{\at{S}}{\baseB} \varGamma(\varDelta)$} \qquad
        \mbox{\myhypertarget{item:supp-bunch-contract-deep}{(ii)} if $\supp{\at{S}}{\baseB} \varGamma(\varDelta)$, then $\supp{\at{S}}{\baseB} \varGamma(\varDelta \acomma \varDelta)$.} 
    \]
\end{proposition} 
%
\begin{proof}
    We prove claim~\myhyperlink{item:supp-bunch-weaken-deep}, the other being similar. We proceed by induction on $\holeDepth{\varGamma(\cdot)}$ --- if $\supp{\at{S}}{\baseB} \varGamma(\varDelta \acomma \varDelta')$, then $\supp{\at{S}}{\baseB} \varGamma(\varDelta)$. 
    
    \begin{itemize}[label=--]
    \item \textsc{base case.}  $\holeDepth{\varGamma(\cdot)}= 0$. By Definition~\ref{def:holeDepth}, $\varGamma(\cdot) = (\cdot)$, so the statement becomes: if $\supp{\at{S}}{\baseB} \varDelta \acomma \varDelta'$, then $\supp{\at{S}}{\baseB} \varDelta$. By \eqref{cl:BI-BeS-supp:acomma}, from $\supp{\at{S}}{\baseB} \varDelta \acomma \varDelta'$, infer $\exists \at{Q}_1, \at{Q}_2 \in \Bunches(\At)$ such that $\at{S} \bunchStrongerThan \at{Q}_1$, $\at{S} \bunchStrongerThan \at{Q}_2$, $\supp{\at{Q}_1}{\baseB} \varDelta$, and $\supp{\at{Q}_2}{\baseB} \varDelta'$. Hence, by monotonicity (Proposition~\ref{prop:monotonicity}), infer $\supp{\at{S}}{\baseB} \varDelta$ --- that is, $\supp{\at{S}}{\baseB} \varGamma(\varDelta)$.

    \item \textsc{inductive step.} Suppose $\holeDepth{\varGamma(\cdot)} = d + 1$ for some $d \in \natN$. Therefore (by Definition~\ref{def:holeDepth}), $\varGamma(\varDelta) \equiv \varSigma(\varDelta) \commavar \varTheta$, where $\holeDepth{\varSigma(\cdot)} = d$, and $\commavar$ is either $\aacomma$ or $\mmcomma$. We consider the two cases of $\commavar$ separately:
    \begin{itemize}[label=--]
        \item $\commavar= \,\,\aacomma$. By \eqref{cl:BI-BeS-supp:acomma}, from $\supp{\at{S}}{\baseB} \varSigma(\varDelta \acomma \varDelta') \acomma \varTheta$, infer  $\exists \at{Q}_1, \at{Q}_1  \in \Bunches(\At)$ such that $\at{S} \bunchStrongerThan \at{Q}_1$, $\at{S} \bunchStrongerThan \at{Q}_2$, $\supp{\at{Q}_1}{\baseB} \varSigma(\varDelta \acomma \varDelta')$, and $\supp{\at{Q}_2}{\baseB} \varTheta$. By the induction hypothesis (IH), $\supp{\at{Q}_1}{\baseB} \varSigma(\varDelta)$. Hence, by \eqref{cl:BI-BeS-supp:acomma}, infer $\supp{\at{S}}{\baseB} \varSigma(\varDelta) \acomma \varTheta$ --- that is, $\supp{\at{S}}{\baseB} \varGamma(\varDelta)$. 
        \item $\commavar = \,\,\mmcomma$.  By \eqref{cl:BI-BeS-supp:mcomma}, from $\supp{\at{S}}{\baseB} \varSigma(\varDelta \acomma \varDelta') \mcomma \varTheta$, infer $\exists \at{Q}_1, \at{Q}_2 \in \Bunches(\At)$ such that $\at{S} \bunchStrongerThan (\at{Q}_1 \mcomma \at{Q}_2)$, $\supp{\at{Q}_1}{\baseB} \varSigma(\varDelta \acomma \varDelta')$, and $\supp{\at{Q}_2}{\baseB} \varTheta$. By the IH, from $\supp{\at{Q}_1}{\baseB} \varSigma(\varDelta \acomma \varDelta')$, infer  $\supp{\at{Q}_2}{\baseB} \varSigma(\varDelta)$. Hence, by \eqref{cl:BI-BeS-supp:mcomma}, infer $\supp{\at{S}}{\baseB} \varSigma(\varDelta) \mcomma \varTheta$ --- that is, $\supp{\at{S}}{\baseB} \varGamma(\varDelta)$. 
    \end{itemize}
    This completes the case analysis.
    \end{itemize}
    This completes the induction.
\end{proof}

\begin{proposition}[Compositionality] \label{lem:support-compositionality}
    If $\phi \supp{\at{S}(\cdot)}{\baseB} \psi$ and $\psi \supp{\at{T}(\cdot)}{\baseB} \chi$, then $\phi \supp{\at{T}(\at{S}(\cdot))}{\baseB} \chi$.
\end{proposition}%
\begin{proof}   
    It suffices to fix some $\baseC \baseGeq \baseB$ and $\at{P} \in \Bunches(\At)$ such that $\supp{\at{P}}{\baseC} \varphi$, and prove $\supp{\at{T}(\at{S}(\at{P}))}{\baseC} \chi$. $\supp{\at{P}}{\baseC} \varphi$ and $\phi \supp{\at{S}(\cdot)}{\baseB} \psi$ imply $\supp{\at{S}(\at{P})}{\baseC} \psi$; together with $\psi \supp{\at{T}(\cdot)}{\baseB} \chi$, it follows that $\supp{\at{T}(\at{S} (\at{P}) )}{\baseC} \chi$. 
\end{proof}
This final result implies that this semantics is \emph{categorical} in the sense that the B-eS of IPL by Sandqvist~\cite{Sandqvist2015base} is categorical --- see Pym et al.~\cite{Pym2022catpts}. 
Briefly, objects are bunches of formulae (i.e., $\Bunches(\Formulas)$) and morphisms  $\varGamma \to \varDelta$ are pairs $\langle \base{B}, \at{S}(\cdot)\rangle$ such that $\varGamma \supp{\at{S}(\cdot)}{\baseB} \varDelta$.
In this set-up, Proposition~\ref{prop:monotonicity} and Lemma~\ref{lem:support-compositionality} enable composition to be defined as follows:
    \[
    \begin{tikzcd}[row sep = scriptsize]
        & \Delta \ar[dr, "\lrangle{\baseB', \at{T}(\cdot)}"] & \\
        \Gamma \ar[ur, "\lrangle{\baseB, \at{S}(\cdot)}"] \ar[rr, "\lrangle{\baseB \cup \baseB', \at{T}(\at{S}(\cdot))}"'] & & \Sigma
    \end{tikzcd}
    \]

This can be viewed as an enrichment of the traditional category-theoretic treatment of IPL in which Heyting algebras are regarded in terms of posets --- see, for example, Lambek and Scott~\cite{lambek1988introduction}. 

\subsection{Example: Rule-based Access Control Policies}
\label{sec:eg-access-control}
In this section, we introduce a toy example to demonstrate the application  of the B-eS given in this paper to \emph{access control}, which is central to security in computer systems --- see, for example,  Abadi~\cite{abadi2003logic} --- and in other systems contexts, such as the passage of people from ground-side to air-side in airports. 

Access control policies specify declaratively the conditions under which agents are intended to be able to access resources in a system. Access control policies are implemented by operational authentication and authorization technologies. Our example lies at the declarative level. 

Suppose you subscribe to the movie streaming service \emph{BunchedFlix}. There are three films on the platform: $\at{a}$, $\at{b}$, and $\at{c}$. Users can become members or purchase tokens that are consumed to access the films. More specifically:
\begin{itemize}[label = --]
    \item to watch $\at{a}$, a user must either be a member or pay three tokens
    \item to watch $\at{b}$, a user must be a member
    \item to watch $\at{c}$, a user must be a member and still pay two tokens.
\end{itemize}
These access-control conditions for \emph{BunchedFlix} can readily be modelled using the B-eS for BI. 

Let $\at{p}^k := \at{p} \mcomma \ldots \mcomma \at{p}$ with $k$-copies. The following base (consisting solely of atomic axioms; that is, atomic rules with empty sets of premisses) captures the access control conditions above:
\[
\base{A}:= \left\{\infer{\at{m} \seq \at{a}}{}\right\}\cup\left\{\infer{\at{p}^k \seq \at{a}}{} \mid k\geq 3\right\} \cup \left\{\infer{\at{m} \seq \at{b}}{}\right\} \cup \left\{\infer{\at{m} \fatsemi \at{p}^k \seq \at{c}}{} \mid k \geq 2\right\}
\]
Here $\at{m}$ denotes membership, and $\at{p}$ denotes a token. Consider an arbitrary user: let $\at{u}$ denote her membership status --- that is, $\at{u}:=\at{m}$ if she is a member of \emph{BunchedFlix}, and $\at{u}:= \aunit$ otherwise. Suppose that she has $k$ tokens on her account. To check that she has the right to watch the film $\at{x} \in \{\at{a},\at{b},\at{c}\}$ on \emph{BunchedFlix} is modelled by the support judgement $\at{u}\supp{\at{p}^k\aacomma(\cdot)}{\base{A}
} \at{x}$. 

For example, let Alice be one such user, and suppose she has $4$ tokens. It is easy to see that Alice has the right to watch all three films $\at{a}$, $\at{b}$, and $\at{c}$ if she is a member, and otherwise only film $\at{a}$.

Suppose Alice is indeed member of \emph{BunchedFlix}. Can she watch both film $a$ and film $\at{c}$? Meaning, in a single transaction, could she get both. In one sense of `both', yes; in another sense, no. As stated above, she may indeed get \emph{either} film, but she may \emph{not} get both at the same time. The first reading corresponds to the judgement $\at{u} \supp{\at{p}^4 \aacomma (\cdot)}{\base{A}} \at{a} \land \at{c}$, which obtains, and the second to $\at{u} \supp{\at{p}^4 \aacomma (\cdot)}{\base{A}} \at{a} \mand \at{c}$, which does not obtain because she lacks the tokens. Of course, being a member, Alice could watch $a$ first and \emph{then} watch $c$, but these are two separate transactions.  This illustrates the resource-sensitivity of the semantics and how it is handled by the additive and multiplicative parts of the logic.

Observe that not much of the possible structures in the B-eS of BI are required to make this example work, despite its making some delicate distinctions about access (i.e., what it means to access both films $\at{a}$ and $\at{c}$). For example, all the rules of $\base{A}$ are axioms, while the general set-up of the B-eS indeed allows rules with premisses. Indeed, it is evident that simpler logics can easily represent this example. Our purpose here is to give a simple example of how BI's B-eS works. 

It should be clear that this basic example suggests many generalizations and variations. In particular, in \cite{MFPS-IRS-2024}, the present authors have shown how BI's B-eS gives rise to an 
`inferentialist resource semantics' in which both BI's sharing/separation interpretation and linear logic's `number-of-uses' reading of propositions as resources can be understood uniformly. It is also shown in \cite{MFPS-IRS-2024} how such an inferentialist resource semantics might provide a basis for an inferentialist account of reasoning about distributed systems (cf. \cite{BarwiseSeligman1997}).

\section{Soundness} \label{sec:soundness}
%
%
\begin{theorem}[Soundness]
\label{thm:BI-BeS-soundness}
    If $\varGamma \proves \phi$, then $\varGamma \supp{}{} \phi$. 
\end{theorem}
The proof proceeds by the typical method of showing that validity respects the inductive definition of logical consequence as determined by provability; for example, the $\ern \mand$-rule in $\system{NBI}$ means we expect validity to satisfy the following:
\[
\mbox{If $\Gamma(\phi \mcomma \psi) \supp{}{} \chi$ and $\Delta \supp{}{} \phi \mand \psi$, then $\Gamma(\Delta) \supp{}{} \chi$}
\]
In some cases, this requires knowing that support \emph{locally} (i.e., with arbitrary base and resources) respect the corresponding rule; for example, for any base $\baseB$ and $R(\cdot) \in \BunchesWithHole{\At}$,
\[
\mbox{If $\Gamma(\phi \mcomma \psi) \supp{R(\cdot)}{\baseB} \chi$ and $\Delta \supp{}{\baseB} \phi \mand \psi$, then $\Gamma(\Delta) \supp{R(\cdot)}{\baseB} \chi$}
\]
These propositions are collected in Lemma~\ref{lem:soundness-key-general}, whose base case has been extracted as Lemma~\ref{lem:soundness-key-base} for readability.

\begin{proof}
Given the inductive definition of $\proves[\system{NBI}]$ (see Figure~\ref{fig:NBI}), it suffices to prove the following statements:


\begin{itemize}[align=right]
    \item[\myhypertarget{item:BI-BeS-sound-axiom}{\text{Ax}}] \quad $\phi \supp{}{} \phi$.
    \item[\myhypertarget{item:BI-BeS-sound-buncheq}{E}] \quad If $\varGamma' \supp{}{} \phi$ and $\varGamma \buncheq \varGamma'$, then $\varGamma \supp{}{} \phi$.
    \item[\myhypertarget{item:BI-BeS-sound-weakening}{W}] \quad If $\varGamma(\varDelta) \supp{}{} \phi$, then $\varGamma(\varDelta \acomma \varDelta') \supp{}{} \phi$. 
    \item[\myhypertarget{item:BI-BeS-sound-contraction}{C}] \quad If $\varGamma(\varDelta \acomma \varDelta) \supp{}{} \phi$, then $\varGamma(\varDelta) \supp{}{} \phi$
    \item[\myhypertarget{item:BI-BeS-sound-mtop-I}{$\mtop$I}] \quad $\munit \supp{}{} \mtop$. 
    \item[\myhypertarget{item:BI-BeS-sound-mtop-E}{\ensuremath{\mtop}E}] \quad If $\varGamma(\munit) \supp{}{} \chi$ and $\varDelta \supp{}{} \mtop$, then $\varGamma(\varDelta) \supp{}{} \chi$. 
    \item[\myhypertarget{item:BI-BeS-sound-wand-I}{$\wand$I}] \quad If $\varGamma \mcomma \phi \supp{}{} \psi$, then $\varGamma \supp{}{} \phi \wand \psi$. 
    \item[\myhypertarget{item:BI-BeS-sound-wand-E}{$\wand$E}] \quad If $\varGamma \supp{}{} \phi \wand \psi$ and $\varDelta \supp{}{} \phi$, then $\varGamma \mcomma \varDelta \supp{}{} \psi$. 
    \item[\myhypertarget{item:BI-BeS-sound-mand-I}{$\mand$I}] \quad If $\varGamma \supp{}{} \phi$ and $\varDelta \supp{}{} \psi$, then $\varGamma \mcomma \varDelta \supp{}{} \phi \mand \psi$.
    \item[\myhypertarget{item:BI-BeS-sound-mand-E}{$\mand$E}] \quad If $\varGamma(\phi \mcomma \psi) \supp{}{} \chi$ and $\varDelta \supp{}{} \phi \mand \psi$, then $\varGamma(\varDelta) \supp{}{} \chi$. 
    \item[\myhypertarget{item:BI-BeS-sound-top-I}{$\top$I}] \quad $\aunit \supp{}{} \top$. 
    \item[\myhypertarget{item:BI-BeS-sound-top-E}{$\top$E}] \quad If $\varGamma(\aunit) \supp{}{} \chi$ and $\varDelta \supp{}{} \top$, then $\varGamma(\varDelta) \supp{}{} \chi$. 
    \item[\myhypertarget{item:BI-BeS-sound-implication-I}{$\to$I}] \quad If $\varGamma \acomma \phi \supp{}{} \psi$ then $\varGamma \supp{}{} \phi \to \psi$. 
    \item[\myhypertarget{item:BI-BeS-sound-implication-E}{$\to$E}] \quad If $\varGamma \supp{}{} \phi \to \psi$ and $\varDelta \supp{}{} \phi$, then $\varGamma \acomma \varDelta \supp{}{} \psi$. 
    \item[\myhypertarget{item:BI-BeS-sound-conjunction-I}{$\land$I}] \quad If $\varGamma \supp{}{} \phi$ and $\varDelta \supp{}{} \psi$, then $\varGamma \acomma \varDelta \supp{}{} \phi \land \psi$.
    \item[\myhypertarget{item:BI-BeS-sound-conjunction-E}{$\land$E}] \quad If $\varGamma(\phi \acomma \psi) \supp{}{} \chi$ and $\varDelta \supp{}{} \phi \land \psi$, then $\varGamma(\varDelta) \supp{}{} \chi$. 
    \item[\myhypertarget{item:BI-BeS-sound-disjunction-I}{$\lor$I}] \quad If $\varGamma \supp{}{} \phi_i$, then $\varGamma \supp{}{} \phi_1 \lor \phi_2$, for $i = 1, 2$. 
    \item[\myhypertarget{item:BI-BeS-sound-disjunction-E}{$\lor$E}] \quad If $\varGamma(\phi) \supp{}{} \chi$, $\varGamma(\psi) \supp{}{} \chi$, and $\varDelta \supp{}{} \phi \lor \psi$, then $\varGamma(\varDelta) \supp{}{} \chi$. 
    \item[\myhypertarget{item:BI-BeS-sound-bot-E}{$\bot$E}] \quad If $\varGamma \supp{}{} \bot$, then $\varGamma \supp{}{} \phi$.
\end{itemize}

We consider each proposition separately. Note that some cases require subinduction, which we move to separate lemmas listed after the current proof. 
\begin{enumerate}
    \item[\myhyperlink{item:BI-BeS-sound-axiom}.] This holds \emph{a fortiori}: according to \eqref{cl:BI-BeS-supp:inf} and Definition~\ref{def:support}, $\phi \supp{}{} \phi$ says that for arbitrary base $\baseX$ and $\at{S} \in \Bunches(\At)$, if $\supp{\at{S}}{\baseB} \phi$, then $\supp{\at{S}}{\baseB} \phi$.
    \item[\myhyperlink{item:BI-BeS-sound-buncheq}.] It suffices to show that whenever $\varGamma \buncheq \varGamma'$, $\supp{\at{S}}{\baseB} \varGamma$ iff $\supp{\at{S}}{\baseB} \varGamma'$. Note that $\buncheq$ is the reflexive transitive closure obtained by applying associativity and unitality laws, so it suffices to show that the supporting relation is invariant under application of associativity and unitality. Associativity is trivial, so restrict attention to unitality. In particular, we verify the multiplicative case, namely $\supp{\at{S}}{\baseB} \varGamma \mcomma \munit$ iff $\supp{\at{S}}{\baseB} \varGamma$. 
    By definition, $\supp{\at{S}}{\baseB} \varGamma \mcomma \munit$ iff there exist $\at{Q}$ and $\at{R}$ satisfying $\at{S} \bunchStrongerThan \at{Q} \mcomma \at{R}$ such that $\supp{\at{Q}}{\baseB} \varGamma$ and $\supp{\at{R}}{\baseB} \munit$. This means $\at{R} \bunchStrongerThan \munit$, thus $\at{S} \bunchStrongerThan \at{Q} \mcomma \at{R} \bunchStrongerThan \at{Q} \mcomma \munit \buncheq \at{Q}$. Therefore $\supp{\at{S}}{\baseB} \varGamma \mcomma \munit$ iff there exists $\at{Q} \bunchWeakerThan \at{S}$ such that $\supp{\at{S}}{\baseB} \varGamma$, namely $\supp{\at{S}}{\baseB} \varGamma$. 
    \item[\myhyperlink{item:BI-BeS-sound-weakening}.] We assume that $\varGamma(\varDelta) \supp{}{} \phi$, and show that $\varGamma(\varDelta \acomma \varDelta') \supp{}{} \phi$ holds for an arbitrary $\varDelta' \in \Bunches(\Formulas)$. So let us fix some $\baseB$ and $\at{S}$ such that $\supp{\at{S}}{\baseB} \varGamma(\varDelta \acomma \varDelta')$, and show that $\supp{\at{S}}{\baseB} \phi$. 
    According to Lemma~\ref{lem:supp-bunch-weaken-contract-deep}, $\supp{\at{S}}{\baseB} \varGamma(\varDelta \acomma \varDelta')$ implies $\supp{\at{S}}{\baseB} \varGamma(\varDelta)$. This together with $\varGamma(\varDelta) \supp{}{} \phi$ conclude that $\supp{\at{S}}{\baseB} \phi$. 
    \item[\myhyperlink{item:BI-BeS-sound-contraction}.] We assume $\varGamma(\varDelta \acomma \varDelta) \supp{}{} \phi$, and show that $\varGamma(\varDelta) \supp{}{} \phi$. So we fix some $\baseB$ and $\at{S}$ such that $\supp{\at{S}}{\baseB} \varGamma(\varDelta)$, and show that $\supp{\at{S}}{\baseB} \phi$. By Lemma~\ref{lem:supp-bunch-weaken-contract-deep}, $\supp{\at{S}}{\baseB} \varGamma(\varDelta)$ implies $\supp{\at{S}}{\baseB} \varGamma(\varDelta \acomma \varDelta)$. This together with $\varGamma(\varDelta \acomma \varDelta) \supp{}{} \phi$ implies $\supp{\at{S}}{\baseB} \phi$.
    \item[\myhyperlink{item:BI-BeS-sound-mtop-I}.]  It suffices to fix some $\baseB$ and $\at{S}$ such that $\supp{\at{S}}{\baseB} \munit$, and then show $\supp{\at{S}}{\baseB} \mtop$. By Definition~\ref{def:support}, $\supp{\at{S}}{\baseB} \munit$ if and only if $\at{S} \bunchStrongerThan \munit$.
    To show $\supp{\at{S}}{\baseB} \mtop$, we fix some $\baseC \baseGeq \baseB$, $\at{P}(\cdot) \in \BunchesWithHole(\At)$, and $\at{q} \in \At$, such that $\supp{\at{P}(\munit)}{\baseC} \at{q}$, and show that $\supp{\at{P}(\at{S})}{\baseC} \at{q}$. Since $\at{S} \bunchStrongerThan \munit$, we have $\at{P}(\at{S}) \bunchStrongerThan \at{P}(\munit)$ (Proposition~\ref{prop:bunch-weaker-compositional}). By monotonicity (Proposition~\ref{prop:monotonicity}), $\supp{\at{P}(\munit)}{\baseC} \at{q}$ then implies $\supp{\at{P}(\at{S})}{\baseC} \at{q}$. 
    \item[\myhyperlink{item:BI-BeS-sound-mtop-E}.] This follows as a special case of Lemma~\ref{lem:soundness-key-general}.\ref{lem:soundness-key-mtop-general}. In particular, one takes there $\at{R}(\cdot)$ to be $(\cdot)$, and $\baseB$ to be the empty base $\emptybase$. 
    \item[\myhyperlink{item:BI-BeS-sound-wand-I}.] We assume $\varGamma \mcomma \phi \supp{}{} \psi$, and show $\varGamma \supp{}{} \phi \wand \psi$. Towards this, we further assume some $\baseB$ and $\at{S}$ satisfying $\supp{\at{S}}{\baseB} \varGamma$, and show $\supp{\at{S}}{\baseB} \phi \wand \psi$. By \eqref{cl:BI-BeS-supp:wand}, this amounts to showing $\phi \supp{\at{S} \mmcomma (\cdot)}{\baseB} \psi$. So we in addition fix some arbitrary base $\baseC \baseGeq \baseB$ and atomic bunch $\at{T}$ such that $\supp{\at{T}}{\baseC} \phi$, and then verify that $ \supp{\at{S} \mmcomma \at{T}}{\baseC} \psi$. Since $\supp{\at{S}}{\baseB} \varGamma$ and $\baseC \baseGeq \baseB$, we have $\supp{\at{S}}{\baseC} \varGamma$; so $\supp{\at{S} \mmcomma \at{T}}{\baseC} \varGamma \mcomma \phi$. Together with $\varGamma \mcomma \phi \supp{}{} \psi$, it follows that $\supp{\at{S} \mmcomma \at{T}}{\baseC} \psi$. 
    \item[\myhyperlink{item:BI-BeS-sound-wand-E}.] We assume $\varGamma \supp{}{} \phi \wand \psi$ and $\varDelta \supp{}{} \phi$, and prove $\varGamma \mcomma \varDelta \supp{}{} \psi$. For this, we fix some arbitrary $\baseB$ and atomic bunch $\at{S}$ such that $\supp{\at{S}}{\baseB} \varGamma \mcomma \varDelta$, and show $\supp{\at{S}}{\baseB} \psi$. By definition, $\supp{\at{S}}{\baseB} \varGamma \mcomma \varDelta$ means that there exist $\at{T}_1, \at{T}_2 \in \Bunches(\At)$ such that $\at{T}_1 \mcomma \at{T}_2 \bunchWeakerThan \at{S}$, $\supp{\at{T}_1}{\baseB} \varGamma$, and $\supp{\at{T}_2}{\baseB} \varDelta$. 
    Now $\varGamma \supp{}{} \phi \mto \psi$ and $\supp{\at{T}_1}{\baseB} \varGamma$ imply that $\supp{\at{T}_1}{\baseB} \phi \mto \psi$; that is, $\phi \supp{\at{T}_1 \mmcomma (\cdot)}{\baseB} \psi$. 
    $\varDelta \supp{}{} \phi$ and $\supp{\at{T}_2}{\baseB} \varDelta$ imply $\supp{\at{T}_2}{\baseB} \phi$. 
    Then $\phi \supp{\at{T}_1 \mmcomma (\cdot)}{\baseB} \psi$ together with $\supp{\at{T}_2}{\baseB} \phi$ imply $\supp{\at{T}_1 \mmcomma \at{T}_2}{\baseB} \psi$, thus $\supp{\at{S}}{\baseB} \psi$ by Proposition~\ref{prop:monotonicity}. 
    \item[\myhyperlink{item:BI-BeS-sound-mand-I}.] Suppose $\varGamma \supp{}{} \phi$ and $\varDelta \supp{}{} \psi$, we prove that $\varGamma \mcomma \varDelta \supp{}{} \phi \mand \psi$. So we fix arbitrary base $\baseB$ and $\at{S} \in \Bunches(\At)$ such that $\supp{\at{S}}{\baseB} \varGamma \mcomma \varDelta$, and show that $\supp{\at{S}}{\baseB} \phi \mand \psi$. 
    By \eqref{cl:BI-BeS-supp:mcomma}, there exist $\at{T}_1, \at{T}_2 \in \Bunches(\At)$ satisfying $(\at{T}_1 \mcomma \at{T}_2) \bunchWeakerThan \at{S}$, such that $\supp{\at{T}_1}{\baseB} \varGamma$ and $\supp{\at{T}_2}{\baseB} \varDelta$. 
    The goal $\supp{\at{S}}{\baseB} \phi \mand \psi$ means that for arbitrary $\baseX \baseGeq \baseB$, $\at{U}(\cdot) \in \BunchesWithHole(\At)$, and $\at{p} \in \At$, if $\phi \mcomma \psi \supp{\at{U}(\cdot)}{\baseC} \at{p}$, then $\supp{\at{U}(\at{S})}{\baseX} \at{p}$. So we fix some $\baseC \baseGeq \baseB$, $\at{P}(\cdot) \in \BunchesWithHole(\At)$, and $\at{q} \in \At$ such that $\phi \mcomma \psi \supp{\at{P}(\cdot)}{\baseC} \at{q}$, and show that $\supp{\at{P}(\at{S})}{\baseC} \at{q}$. 
    Since $\varGamma \supp{}{} \phi$ and $\supp{\at{T}_1}{\baseB} \Gamma$, it follows that $\supp{\at{T}_1}{\baseB} \phi$; since $\varDelta \supp{}{} \psi$ and $\supp{\at{T}_2}{\baseB} \varDelta$, it follows that $\supp{\at{T}_2}{\baseB} \psi$. So $\supp{\at{T}_1 \mmcomma \at{T}_2}{\baseB} \phi \mcomma \psi$; together with $\at{T}_1 \mcomma \at{T}_2 \bunchWeakerThan \at{S}$, it follows that $\supp{\at{S}}{\baseB} \phi \mcomma \psi$ by Proposition~\ref{prop:monotonicity}. Together with $\phi \mcomma \psi \supp{\at{P}(\cdot)}{\baseC} \at{q}$ and $\baseC \baseGeq \baseB$, we conclude that $\supp{P(\at{S})}{\baseC} \at{q}$. 
    \item[\myhyperlink{item:BI-BeS-sound-mand-E}.] This follows as a special case of Lemma~\ref{lem:soundness-key-general}.\ref{lem:soundness-key-mand-general}, where one takes $\baseB$ and $R(\cdot)$ to be the empty base $\emptybase$ and the empty context $(\cdot)$, respectively. 
    %
    \item[\myhyperlink{item:BI-BeS-sound-top-I}.] It suffices to fix some base $\baseB$ and $\at{S} \in \Bunches(\At)$ satisfying $\supp{\at{S}}{\baseB} \aunit$, and the goal is to show $\supp{\at{S}}{\baseB} \top$. To this end, we further fix some arbitrary $\baseC \baseGeq \baseB$ and $\at{P}(\cdot) \in \BunchesWithHole(\At)$ satisfying $\supp{\at{P}(\aunit)}{\baseC} \at{q}$, and prove $\supp{\at{P}(\at{S})}{\baseC} \at{q}$. Since $\supp{\at{S}}{\baseB} \aunit$ says $\at{S} \bunchStrongerThan \aunit$, by Proposition~\ref{prop:monotonicity},  $\supp{\at{P}(\aunit)}{\baseC} \at{q}$ implies $\supp{\at{P}(\at{S})}{\baseC} \at{q}$. 
    \item[\myhyperlink{item:BI-BeS-sound-top-E}.] This follows immediately from Lemma~\ref{lem:soundness-key-general}.\ref{lem:soundness-key-top-general}. 
    \item[\myhyperlink{item:BI-BeS-sound-implication-I}.] Fix some base $\baseB$ and $\at{S} \in \Bunches(\At)$ such that $\supp{\at{S}}{\baseB} \varGamma$, and the goal is to prove $\supp{\at{S}}{\baseB} \varphi \to \psi$. By \eqref{cl:BI-BeS-supp:implication}, this amounts to showing $\phi \supp{\at{S} \aacomma (\cdot)}{\baseB} \psi$. So we further fix some $\baseC \baseGeq \baseB$ and $\at{T} \in \Bunches(\At)$ satisfying $\supp{\at{T}}{\baseC} \phi$, and prove $\supp{\at{S} \aacomma \at{T}}{\baseC} \psi$. By \eqref{cl:BI-BeS-supp:acomma}, $\supp{\at{S}}{\baseB} \varGamma$ and $\supp{\at{T}}{\baseC} \phi$ imply $\supp{\at{S} \aacomma \at{T}}{\baseC} \varGamma \acomma \phi$. Together with $\varGamma \acomma \phi \supp{}{} \psi$, it follows that $\supp{\at{S} \aacomma \at{T}}{\baseC} \psi$. 
    \item[\myhyperlink{item:BI-BeS-sound-implication-E}.] Given the assumptions, we assume some arbitrary base $\baseB$ and $\at{S} \in \Bunches(\At)$ such that $\supp{\at{S}}{\baseB} \varGamma \acomma \varDelta$, and prove that $\phi \supp{\at{S}}{\baseB} \psi$. By \eqref{cl:BI-BeS-supp:acomma}, there exist $\at{P}, \at{Q} \in \Bunches(\At)$ satisfying $\at{P} \bunchWeakerThan \at{S}$ and $\at{Q} \bunchWeakerThan \at{S}$, such that $\supp{\at{P}}{\baseB} \varGamma$ and $\supp{\at{Q}}{\baseB} \varDelta$. 
    $\varGamma \supp{}{} \varphi \to \psi$ and $\supp{\at{P}}{\baseB} \varGamma$ imply $\supp{\at{P}}{\baseB} \varphi \to \psi$, namely $\supp{\at{P} \aacomma (\cdot)}{\baseB} \psi$. $\varDelta \supp{}{} \psi$ and $\supp{\at{Q}}{\baseB} \varDelta$ imply $\supp{\at{Q}}{\baseB} \psi$. Then, by \eqref{cl:BI-BeS-supp:inf}, it follows that $\supp{\at{P} \aacomma \at{Q}}{\baseB} \psi$. Since $\at{S} \aacomma \at{S} \bunchStrongerThan \at{P} \aacomma \at{Q}$, it follows that $\supp{\at{S} \aacomma \at{S}}{\baseB} \psi$, by Proposition~\ref{prop:monotonicity}. Therefore, $\supp{\at{S} \aacomma \at{S}}{\baseB} \psi$, by Proposition~\ref{prop:supp-resource-weaken-contract-deep}.
    \item[\myhyperlink{item:BI-BeS-sound-conjunction-I}.] Suppose $\varGamma \supp{}{} \phi$ and $\varDelta \supp{}{} \psi$, and we show that $\varGamma \acomma \varDelta \supp{}{} \phi \land \psi$. So, we fix some arbitrary $\baseB$ and $\at{S}$ satisfying that $\supp{\at{S}}{\baseB} \varGamma \acomma \varDelta$, and show $\supp{\at{S}}{\baseB} \phi \land \psi$. 
    Applied to \eqref{cl:BI-BeS-supp:acomma}, $\supp{\at{S}}{\baseB} \varGamma \acomma \varDelta$, there exist $\at{P}, \at{Q} \in \Bunches(\At)$ satisfying $\at{P} \bunchWeakerThan \at{S}$ and $\at{Q} \bunchWeakerThan \at{S}$, such that $\supp{\at{P}}{\baseB} \varGamma$ and $\supp{\at{Q}}{\baseB} \varDelta$. $\varGamma \supp{}{} \varphi$ and $\supp{\at{P}}{\baseB} \varGamma$ imply $\supp{\at{P}}{\baseB} \varphi$. $\varDelta \supp{}{} \psi$ and $\supp{\at{Q}}{\baseB} \varDelta$ imply $\supp{\at{Q}}{\baseB} \psi$. 
    In order to show $\supp{\at{S}}{\baseB} \phi \land \psi$, we fix some $\baseC \baseGeq \baseB$, $\at{R}(\cdot) \in \BunchesWithHole(\At)$, and $\at{q} \in \At$ such that $\phi \acomma \psi \supp{\at{R}(\cdot)}{\baseC} \at{q}$, and prove $\supp{\at{R}(\at{S})}{\baseC} \at{q}$. Since $\supp{\at{P}}{\baseB} \varphi$, $\supp{\at{Q}}{\baseB} \psi$, $\at{P} \bunchWeakerThan \at{S}$, and $\at{Q} \bunchWeakerThan \at{S}$, it follows by \eqref{cl:BI-BeS-supp:acomma} that $\supp{\at{S}}{\baseB} \varphi \acomma \psi$. This together with $\phi \acomma \psi \supp{\at{R}(\cdot)}{\baseC} \at{q}$ imply $\supp{\at{R}(\at{S})}{\baseC} \at{q}$. 
    %
    %
    \item[\myhyperlink{item:BI-BeS-sound-conjunction-E}.] This follows as a special case of Lemma~\ref{lem:soundness-key-general}.\ref{lem:soundness-key-and-general}. 
    \item[\myhyperlink{item:BI-BeS-sound-disjunction-I}.]
    Without loss of generality, assume that $\varGamma \supp{}{} \phi_1$ and show that $\varGamma \supp{}{} \phi_1 \lor \phi_2$. So let us assume some arbitrary base $\baseB$ and $\at{S} \in \Bunches(\At)$ such that $\supp{\at{S}}{\baseB} \varGamma$, and show that $\supp{\at{S}}{\baseB} \phi_1 \lor \phi_2$. Under \eqref{cl:BI-BeS-supp:disjunction}, this amounts to showing that: 
    \[
        \forall \baseX \baseGeq \baseB, \forall \at{U}(\cdot) \in \BunchesWithHole(\At), \forall \at{p} \in \At, \text{ if } \phi_i \supp{\at{U}(\cdot)}{\baseX} \at{p} \text{ for } i = 1, 2, \text{ then} \supp{\at{U}(\at{S})}{\baseX} \at{p} 
    \]
    So, we fix some arbitrary $\baseC \baseGeq \baseB$, $\at{T}(\cdot) \in \BunchesWithHole(\At)$, and $\at{q} \in \At$ such that $\phi_i \supp{\at{T}(\cdot)}{\baseC} \at{q}$ for $i = 1, 2$, and then prove that $\supp{\at{T}(\at{S})}{\baseC} \at{q}$. Note that $\supp{\at{S}}{\baseB} \varGamma$ and $\varGamma \supp{}{} \phi_1$ imply that $\supp{\at{S}}{\baseB} \phi_1$. This together with $\phi_1 \supp{\at{T}(\cdot)}{\baseC} \at{q}$ entail that $\supp{\at{T}(\at{S})}{\baseC} \at{q}$. 
    \item[\myhyperlink{item:BI-BeS-sound-disjunction-E}.] follows immediately from Lemma~\ref{lem:soundness-key-general}.\ref{lem:soundness-key-or-general}. 
    \item[\myhyperlink{item:BI-BeS-sound-bot-E}.] We assume $\varGamma \supp{}{} \bot$, and show that $\varGamma \supp{}{} \phi$. So we fix some base $\baseB$ and atomic bunch $\at{S}$ such that $\supp{\at{S}}{\baseB} \varGamma$, and show that $\supp{\at{S}}{\baseB} \phi$. $\supp{\at{S}}{\baseB} \varGamma$ and $\varGamma \supp{}{} \bot$ imply $\supp{\at{S}}{\baseB} \bot$. Then by Lemma~\ref{lem:soundness-key-general}.\ref{lem:general-bot-elim-sound}, it follows immediately that $\supp{\at{S}}{\baseB} \varphi$. 
\end{enumerate}
    This completes the induction. 
\end{proof}

\begin{lemma}\label{lem:soundness-key-general}
    The following hold:
    \begin{enumerate}
        \item  \label{lem:soundness-key-mtop-general} If $\varGamma(\munit) \supp{\at{R}(\cdot)}{\baseB} \chi$ and $\varDelta \supp{}{} \mtop$, then $\varGamma(\varDelta) \supp{\at{R}(\cdot)}{\baseB} \chi$. 
        \item  \label{lem:soundness-key-mand-general} If $\varGamma(\phi \mcomma \psi) \supp{\at{R}(\cdot)}{\baseB} \chi$ and  $\varDelta \supp{}{} \phi \mand \psi$, then $\varGamma(\varDelta) \supp{\at{R}(\cdot)}{\baseB} \chi$.
    \item \label{lem:soundness-key-top-general}  If $\varGamma(\aunit) \supp{\at{R}(\cdot)}{\baseB} \chi$ and $\varDelta \supp{}{} \top$, then $\varGamma(\varDelta) \supp{\at{R}(\cdot)}{\baseB} \chi$. 
    \item \label{lem:soundness-key-and-general}  If $\varGamma(\phi \acomma \psi) \supp{R(\cdot)}{\baseB} \chi$ and $\varDelta \supp{}{} \phi \land \psi$, then $\varGamma(\varDelta) \supp{R(\cdot)}{\baseB} \chi$. 
    \item  \label{lem:soundness-key-or-general}  If $\varDelta \supp{}{} \phi \lor \psi$, $\varGamma(\phi) \supp{\at{R}(\cdot)}{\baseB} \chi$, and $\varGamma(\psi) \supp{\at{R}(\cdot)}{\baseB} \chi$, then $\varGamma(\varDelta) \supp{\at{R}(\cdot)}{\baseB} \chi$.
    \item \label{lem:general-bot-elim-sound} For any formula $\phi$, $\bot \supp{\at{R}(\cdot)}{\baseB} \phi$. 
        \end{enumerate}
\end{lemma}

\begin{proof}[Proof of Lemma~\ref{lem:soundness-key-general}.\ref{lem:soundness-key-mtop-general}]
    We proceed by induction on the size of $\holeDepth{\varGamma(\cdot)}{}$. 
    \begin{itemize}[label=--]
        \item \textsc{base case}. $\holeDepth{\varGamma(\cdot)}{} = 0$. This means $\varGamma(\cdot)$ is of the form $(\cdot)$. The proposition to be proven becomes: if $\munit \supp{\at{R}(\cdot)}{\baseB} \chi$ and $\varDelta \supp{}{} \mtop$, then $\varDelta \supp{\at{R}(\cdot)}{\baseB} \chi$. This follows immediately from Lemma~\ref{lem:soundness-key-base}.\ref{lem:soundness-key-mtop-base}. 
        \item \textsc{inductive step}.  $\holeDepth{\varGamma(\cdot)}{} = d + 1$ for some $d \geq 0$. This $\varGamma(\cdot)$ is coherently equivalent to $\varSigma(\cdot) \commavar \varTheta$, where $\commavar \in \{ \,\acomma\, , \,\mcomma\, \}$, $\varSigma(\cdot) \in \BunchesWithHole(\Formulas)$, $\varTheta \in \Bunches(\Formulas)$ and $\holeDepth{\varSigma(\cdot)} = d$. We consider the two cases for $\commavar$ seperately: 
        \begin{itemize}[label = --]
            \item $\commavar$ is $\mmcomma$. Given $\varSigma(\munit) \mcomma \varTheta \supp{\at{R}(\cdot)}{\baseB} \chi$ and $\varDelta \supp{}{} \mtop$, the goal is to show that $\varSigma(\varDelta) \mcomma \varTheta \supp{\at{R}(\cdot)}{\baseB} \chi$. To this, we fix some $\baseC \baseGeq \baseB$ and $\at{P} \in \Bunches(\At)$ such that $\supp{\at{P}}{\baseC} \varSigma(\varDelta) \mcomma \varTheta$, and show that $\supp{\at{R}(\at{P})}{\baseC} \chi$ holds as well. $\supp{\at{P}}{\baseC} \varSigma(\varDelta) \mcomma \varTheta$ means that there exist $\at{P}_1, \at{P}_2 \in \Bunches(\At)$ such that $\at{P}_1 \mcomma \at{P}_2 \bunchWeakerThan \at{P}$, $\supp{\at{P}_1}{\baseC} \varSigma(\varDelta)$, and $\supp{\at{P}_2}{\baseC} \varTheta$. This together with $\varSigma(\munit) \mcomma \varTheta \supp{\at{R}(\cdot)}{\baseB} \chi$ imply that, for arbitrary $\baseX \baseGeq \baseC$ and $\at{U} \in \Bunches(\At)$, if $\supp{\at{U}}{\baseX} \varSigma(\munit)$, then $\supp{\at{R}(\at{U} \mmcomma P_2)}{\baseX} \chi$; that is, $\varSigma(\munit) \supp{\at{R}((\cdot) \mmcomma \at{P}_2)}{\baseC} \chi$. Since $\holeDepth{\varSigma(\cdot)}{} = d$, we can apply IH to $\varSigma(\munit) \supp{\at{R}((\cdot) \mmcomma \at{P}_2)}{\baseC} \chi$ and $\varDelta \supp{}{} \mtop$ to conclude that $\varSigma(\varDelta) \supp{\at{R}( (\cdot) \mmcomma \at{P}_2 )}{\baseC} \chi$. Since $\supp{\at{P}_1}{\baseC} \varSigma(\varDelta)$, we have $\supp{\at{R}( \at{P}_1 \mmcomma \at{P}_2 )}{\baseC} \chi$, thus $\supp{\at{R}(\at{P})}{\baseC} \chi$. 
            \item $\commavar$ is $\aacomma$. Given $\varSigma(\munit) \acomma \varTheta \supp{\at{R}(\cdot)}{\baseB} \chi$ and $\varDelta \supp{}{} \mtop$, the goal is to show that $\varSigma(\varDelta) \acomma \varTheta \supp{\at{R}(\cdot)}{\baseB} \chi$. So we fix some $\baseC \baseGeq \baseB$ and $\at{P} \in \Bunches(\At)$ such that $\supp{\at{P}}{\baseC} \varSigma(\varDelta) \acomma \varTheta$, and prove $\supp{\at{R}(\at{P})}{\baseC} \chi$. By definition, $\supp{\at{P}}{\baseC} \varSigma(\varDelta) \acomma \varTheta$ says that there exists $\at{P}_1, \at{P}_2 \in \Bunches(\At)$ such that $\at{P}_i \bunchWeakerThan \at{P}$ for $i = 1, 2$, $\supp{\at{P}_1}{\baseC} \varSigma(\varDelta)$, and $\supp{\at{P}_2}{\baseC} \varTheta$. Spelling out $\varGamma(\munit) \supp{\at{R}(\cdot)}{\baseB} \chi$, we have: 
            \begin{equation}
            \label{eq:general-mtop-elim-sound-1}
                \forall \baseX \baseGeq \baseB,  \forall \at{U} \in \Bunches(\At), \text{if } \exists \at{U}_1, \at{U}_2 \text{ s.t. } \at{U}_1 \bunchWeakerThan \at{U}, \at{U}_2 \bunchWeakerThan \at{U}, \supp{\at{U}_1}{\baseX} \varSigma(\munit), \supp{\at{U}_2}{\baseX} \varTheta, \text{then} \supp{\at{R}(\at{U})}{\baseX} \chi 
            \end{equation}
            Since $\supp{\at{P}_2}{\baseC} \varTheta$, \eqref{eq:general-mtop-elim-sound-1} entails: 
            \begin{equation*}
                \forall \baseX \baseGeq \baseC, \forall \at{U} \in \Bunches(\At), \text{if } \at{P}_2 \bunchWeakerThan \at{U} \text{ and } \exists \at{U}_1 \bunchWeakerThan \at{U} \text{ s.t.} \supp{\at{U}_1}{\baseX} \varSigma(\munit), \text{ then} \supp{\at{R}(\at{U})}{\baseX} \chi 
            \end{equation*}
            thus, 
            \begin{equation*}
                \forall \baseX \baseGeq \baseC, \forall \at{U}_1 \in \Bunches(\At), \text{if } \supp{\at{U}_1}{\baseX} \varSigma(\munit), \text{ then} \supp{\at{R}(\at{U}_1 \aacomma \at{P}_2)}{\baseX} \chi 
            \end{equation*}
            By definition, this precisely says $\varSigma(\munit) \supp{\at{R}( (\cdot) \aacomma \at{P}_2 )}{\baseX} \chi$. Since $\holeDepth{\varSigma(\cdot)} = d$, we apply IH to this and $\varDelta \supp{}{} \mtop$, and conclude $\varSigma(\varDelta) \supp{\at{R}( (\cdot) \aacomma \at{P}_2 )}{\baseC} \chi$. Together with $\supp{\at{P}_1}{\baseC} \varSigma(\varDelta)$, we have $\supp{\at{R}(\at{P}_1 \aacomma \at{P}_2)}{\baseC} \chi$. Since $\at{P}_i \bunchWeakerThan \at{P}$ for $i = 1, 2$, by Proposition~\ref{prop:monotonicity}, $\supp{\at{R}(\at{P}_1 \aacomma \at{P}_2)}{\baseC} \chi$ implies $\supp{\at{R}(\at{P} \aacomma \at{P})}{\baseC} \chi$, thus $\supp{\at{R}(\at{P})}{\baseC} \chi$. 
        \end{itemize}
    \end{itemize}
    This completes the inductive proof. 
\end{proof}

\begin{proof}[Proof of Lemma~\ref{lem:soundness-key-general}.\ref{lem:soundness-key-mand-general}]
    We proceed by induction on $\holeDepth{\varGamma(\cdot)}$.
    \begin{itemize}[label=--]
    
    \item \textsc{base case.} $\holeDepth{\varGamma(\cdot)}{} = 0$. That is, $\varGamma(\cdot) = (\cdot)$. Then the proposition becomes: if $\phi \mcomma \psi \supp{\at{R}(\cdot)}{\baseB} \chi$ and $\varDelta \supp{}{} \phi \mand \psi$, then $\varDelta \supp{\at{R}(\cdot)}{\baseB} \chi$. This follows immediately from Lemma~\ref{lem:soundness-key-base}.\ref{lem:soundness-key-mand-base}. 

    \item \textsc{inductive step.}  $\holeDepth{\varGamma(\cdot)}{} = d + 1$ for $d \geq 0$. There are two cases, depending on the principal context former of $\varGamma(\cdot)$ is $\aacomma$ or $\mmcomma$ (note that $\varGamma(\cdot)$ must have a principal context-former given its depth). 
    \begin{itemize}[label = --]
    \item $\varGamma(\cdot)$ is coherently equivalent to $\varSigma(\cdot) \acomma \varTheta$, where $\holeDepth{\varSigma(\cdot)} = d$, and $\varTheta \in \Bunches(\At)$. We assume the premises of the proposition as well as $\supp{\at{S}}{\baseC} \varGamma(\varDelta)$ for some $\baseC \baseGeq \baseB$ and $\at{S} \in \Bunches(\At)$. The goal is to prove $\supp{\at{R}(\at{S})}{\baseC} \chi$. By the definition of $\aacomma$, $\supp{\at{S}}{\baseC} \varSigma(\varDelta) \acomma \varTheta$ means that there exist $\at{P}, \at{Q} \in \Bunches(\At)$, such that $\at{P} \bunchWeakerThan \at{S}$, $\at{Q} \bunchWeakerThan \at{S}$, $\supp{\at{P}}{\baseC} \varSigma(\varDelta)$, and $\supp{\at{Q}}{\baseC} \varTheta$. 
    $\varSigma(\phi \mcomma \psi) \acomma \varTheta \supp{\at{R}(\cdot)}{\baseB} \chi$ says that, for arbitrary $\baseX \baseGeq \baseB$ and $\at{W} \in \Bunches(\At)$, $\supp{\at{W}}{\baseX} \varSigma(\phi \mcomma \psi) \acomma \varTheta$ implies $\supp{\at{R}(\at{W})}{\baseX} \chi$. It is further spelled out as: 
    \begin{equation}
    \label{eq:general-mand-elim-sound-1}
        \begin{split}
            & \text{for arbitrary }  \baseX \baseGeq \baseB \text{ and } \at{W} \in \Bunches(\At), \text{if there exist } \at{U} \bunchWeakerThan \at{W}, \at{V} \bunchWeakerThan \at{W} \\
            & \text{s.t.} \supp{\at{U}}{\baseX} \varSigma(\phi \mcomma \psi) \text{ and} \supp{\at{W}}{\baseX} \varTheta, \text{ then} \supp{\at{R}(\at{W})}{\baseX} \chi 
        \end{split}
    \end{equation}
    Since $\baseC \baseGeq \baseB$, and $\at{Q} \bunchWeakerThan \at{S}$ satisfies $\supp{\at{Q}}{\baseC} \varTheta$, \eqref{eq:general-mand-elim-sound-1} entails the following: 
    \begin{equation}
    \label{eq:general-mand-elim-sound-2}
        \begin{split}
            \forall \baseX \baseGeq \baseC, \forall \at{U} \in \Bunches(\At), \text{ if}   \supp{\at{U}}{\baseC} \varSigma(\phi \mcomma \psi), \text{ then} \supp{\at{R}(\at{Q} \aacomma \at{U})}{\baseX} \chi
        \end{split}
    \end{equation}
    This is exactly $\varSigma(\phi \mcomma \psi) \supp{\at{R}(\at{Q} \aacomma (\cdot))}{\baseC} \chi$, by definition. Since $\holeDepth{\varSigma(\cdot)}{} = d$, we can apply induction hypothesis to $\varSigma(\phi \mcomma \psi) \supp{\at{R}(\at{Q} \aacomma (\cdot))}{\baseC} \chi$ and $\varDelta \supp{}{} \phi \mand \psi$, it follows that $\varGamma(\varDelta) \supp{\at{R}(\at{Q} \aacomma (\cdot))}{\baseC} \chi$. This together with $\supp{\at{P}}{\baseC} \varGamma(\varDelta)$ implies $\supp{\at{R}(\at{P} \aacomma \at{Q})}{\baseC} \chi$, hence $\supp{\at{R}(\at{S} \aacomma \at{S})}{\baseC} \chi$ by $(\at{P} \aacomma \at{Q}) \bunchWeakerThan (\at{S} \aacomma \at{S})$ and Proposition~\ref{prop:monotonicity}; hence $\supp{\at{R}(\at{S})}{\baseC} \chi$ by Proposition~\ref{prop:supp-resource-weaken-contract-deep}
    \item $\varGamma(\cdot)$ is coherently equivalent to $\varSigma(\cdot) \mcomma \varTheta$, where $\varSigma(\cdot) \in \BunchesWithHole(\At)$, $\varTheta \in \Bunches(\At)$, and $\holeDepth{\varSigma(\cdot)} = d$. 
    We assume the premises of the proposition as well as $\supp{\at{S}}{\baseC} \varGamma(\varDelta)$ for some $\baseC \baseGeq \baseB$ and $\at{S} \in \Bunches(\At)$. The goal is to prove $\supp{\at{R}(\at{S})}{\baseC} \chi$. 
    According to \eqref{cl:BI-BeS-supp:mcomma}, $\supp{\at{S}}{\baseC} \varSigma(\varDelta) \mcomma \varTheta$ means that there exist $\at{P}, \at{Q}\in \Bunches(\At)$ satisfying $(\at{P} \mcomma \at{Q}) \bunchWeakerThan \at{S}$, such that $\supp{\at{P}}{\baseC} \varSigma(\varDelta)$ and $\supp{\at{Q}}{\baseC} \varTheta$. 
    Also, $\varGamma(\phi \mcomma \psi) \supp{\at{R}(\cdot)}{\baseB} \chi$ says, for arbitrary $\baseX \baseGeq \baseB$ and $\at{W} \in \Bunches(\At)$, $\supp{\at{W}}{\baseX} \varGamma(\phi \mcomma \psi)$ implies $\supp{\at{R}(\at{W})}{\baseX} \chi$. This, according to the form of $\varGamma(\cdot)$,  boils down to the following: 
    \begin{equation}
    \label{eq:general-mand-elim-sound-4}
        \begin{split}
            &\forall \baseX \baseGeq \baseB, \forall \at{W} \in \Bunches(\At), \text{ if } \exists \at{U}, \at{V} \in \Bunches(\At) \text{ s.t. } 
            \at{U} \mcomma \at{V} \bunchWeakerThan \at{W}, \\
            & \supp{\at{U}}{\baseX} \varSigma(\phi \mcomma \psi),
            \text{ and} \supp{\at{V}}{\baseX} \varTheta, \text{ then} \supp{\at{R}(\at{W})}{\baseX} \chi 
        \end{split}
    \end{equation}
    Take into consideration $\baseC \baseGeq \baseB$ and $\supp{\at{Q}}{\baseC} \varTheta$, \eqref{eq:general-mand-elim-sound-4} entails the follows: 
    \begin{equation}
        \forall \baseX \baseGeq \baseC, \forall \at{U} \in \Bunches(\At), \text{ if} \supp{\at{U}}{\baseX} \varSigma(\phi \mcomma \psi), \text{ then} \supp{\at{R}(\at{U} \mmcomma \at{Q})}{\baseX} \chi 
    \end{equation}
    This is exactly $\varSigma(\phi \mcomma \psi) \supp{\at{R}( (\cdot) \mmcomma \at{Q} )}{\baseC} \chi$. Now apply the induction hypothesis to $\varSigma(\phi \mcomma \psi) \supp{\at{R}( (\cdot) \mmcomma \at{Q} )}{\baseC} \chi$ and $\varDelta \supp{}{} \phi \mand \psi$, one get $\varSigma(\varDelta) \supp{\at{R}( (\cdot) \mmcomma \at{Q} )}{\baseC} \chi$. This together with $\supp{\at{P}}{\baseC} \varSigma(\varDelta)$ entails that $\supp{\at{R}(\at{P} \mmcomma \at{Q})}{\baseC} \chi$. Since $\at{S} \bunchStrongerThan \at{P} \mcomma \at{Q}$, it follows that $\supp{\at{R}(\at{S})}{\baseC} \chi$. 
    \end{itemize}
    \end{itemize}
    This completes the proof by induction. 
\end{proof}

\begin{proof}[Proof of Lemma~\ref{lem:soundness-key-general}.\ref{lem:soundness-key-top-general}]
    The proof follows a similar strategy as that for $\mtop$, namely Lemma~\ref{lem:soundness-key-mtop-general}. In particular, it follows from a proof by induction on $\holeDepth{\varGamma(\cdot)}$, whose base case is deduced from an analogy to Lemma~\ref{lem:soundness-key-base}.\ref{lem:soundness-key-mtop-base}. 
\end{proof}

\begin{proof}[Proof of Lemma~\ref{lem:soundness-key-general}.\ref{lem:soundness-key-and-general}]

    We proceed by induction on the size of $\holeDepth{\varGamma(\cdot)}$. 
    \begin{itemize}[label=--]
     \item \textsc{base case}. $\holeDepth{\varGamma(\cdot)}{} = 0$. That is, $\varGamma(\cdot) \buncheq (\cdot)$. Then the proposition amounts to: if $\phi \acomma \psi \supp{\at{R(\cdot)}}{\baseB} \chi$ and $\varDelta \supp{}{} \phi \land \psi$, then $\varDelta \supp{\at{R(\cdot)}}{\baseB} \chi$. This is an immediate consequence of Lemma~\ref{lem:soundness-key-base}.\ref{lem:soundness-key-and-base}. 

    \item \textsc{inductive step.} $\holeDepth{\varGamma(\cdot)}{} = d + 1$ for some $d \geq 0$. The analysis consist of two cases wrt the principal context-former of $\varGamma(\cdot)$ being $\aacomma$ or $\mmcomma$. In both cases, we assume the premises of the proposition as well as $\supp{\at{S}}{\baseC} \varGamma(\varDelta)$ for some arbitrary fixed $\baseC \baseGeq \baseB$ and $\at{S}$, and the goal is to show $\supp{\at{R}(\at{S})}{\baseC} \chi$. 
    \begin{itemize}[label = --]
        \item $\varGamma(\cdot)$ is coherently equivalent to $\varSigma(\cdot) \acomma \varTheta$, where $\varTheta \in \Bunches(\At)$, $\varSigma(\cdot) \in \BunchesWithHole(\At)$, and $\holeDepth{\varSigma(\cdot)} = d$. By definition, $\supp{\at{S}}{\baseC} \varSigma(\varDelta) \acomma \varTheta$ means that there exist $\at{P}, \at{Q} \in \Bunches(\At)$ satisfying $\at{P} \bunchWeakerThan \at{S}, \at{Q} \bunchWeakerThan \at{S}$, such that $\supp{\at{P}}{\baseC} \varSigma(\varDelta)$ and $\supp{\at{Q}}{\baseC} \varTheta$. 
        By the definition of $\varSigma(\phi \acomma \psi) \acomma \varTheta \supp{\at{R}(\cdot)}{\baseB} \chi$, we have that for arbitrary $\baseX \baseGeq \baseB$ and $\at{W}$, if $\supp{\at{W}}{\baseX} \varSigma(\phi \acomma \psi) \acomma \varTheta$, then $\supp{\at{R}(\at{W})}{\baseX} \chi$. This means: 
        \begin{equation}
        \label{eq:general-and-elim-sound-1}
            \begin{split}
                &\forall \baseX \baseGeq \baseC,  \forall \at{W} \in \Bunches(\At), \text{ if } \exists \at{U}, \at{V} \in \Bunches(\At) \text{ s.t. } \at{U} \bunchWeakerThan \at{W}, \at{V} \bunchWeakerThan \at{W}, \\
                & \supp{\at{U}}{\baseX} \varSigma(\phi \acomma \psi), \text{ and} \supp{\at{V}}{\baseX} \varTheta, \text{ then} \supp{\at{R}(\at{W})}{\baseC} \chi 
            \end{split}
        \end{equation}
        Since $\supp{\at{Q}}{\baseC} \varTheta$, it follows that: 
        \begin{equation}
            \label{eq:general-and-elim-sound-2}
        \forall \baseX \baseGeq \baseC, \forall \at{U} \in \Bunches(\At), \text{ if} \supp{\at{U}}{\baseX} \varSigma(\varphi \acomma \psi), \text{ then} \supp{\at{R}(\at{U} \aacomma \at{Q})}{\baseX} \chi 
        \end{equation}
        In other words, $\varSigma(\phi \acomma \psi) \supp{\at{R}( (\cdot) \aacomma \at{Q} )}{\baseC} \chi$. 
        Remember that we have $\holeDepth{\varSigma(\cdot)}{} = d$, so we can apply the induction hypothesis to $\varSigma(\phi \acomma \psi) \supp{\at{R}( (\cdot) \aacomma \at{Q} )}{\baseC} \chi$ and $\varDelta \supp{}{} \phi \land \psi$ to conclude that $\varSigma(\varDelta) \supp{\at{R}( (\cdot) \aacomma \at{Q} )}{\baseC} \chi$. This together with $\supp{\at{P}}{\baseC} \varSigma(\varDelta)$ imply that $\supp{\at{R}(\at{P} \aacomma \at{Q})}{\baseC} \chi$; hence $\supp{\at{R}(\at{S} \aacomma \at{S})}{\baseC} \chi$ by Proposition~\ref{prop:monotonicity} and $\at{P} \acomma \at{Q} \bunchWeakerThan \at{S} \acomma \at{S}$; hence $\supp{\at{R}(\at{S})}{\baseC} \chi$ by Lemma~\ref{lem:supp-bunch-weaken-contract-deep}. 
        \item $\varGamma(\cdot)$ is coherently equivalent to $\varSigma(\cdot) \mcomma \varTheta$, where $\varSigma(\cdot) \in \BunchesWithHole(\At)$, $\varTheta \in \Bunches(\At)$, and $\holeDepth{\varSigma(\cdot)} = d$. By \eqref{cl:BI-BeS-supp:mcomma}, $\supp{\at{S}}{\baseC} \varSigma(\varDelta) \mcomma \varTheta$ means that there exist $\at{P}, \at{Q} \in \Bunches(\At)$ satisfying $\at{P} \mcomma \at{Q} \bunchWeakerThan \at{S}$, such that $\supp{\at{P}}{\baseC} \varSigma(\varDelta)$ and $\supp{\at{Q}}{\baseC} \varTheta$. Then $\varSigma(\phi \acomma \psi) \mcomma \varTheta \supp{\at{R}(\cdot)}{\baseB} \chi$ means: 
        \begin{equation*}
        \begin{split}
        &\text{$\forall \baseX \baseGeq \baseB$, $\forall \at{W} \in \Bunches(\At)$, if $\exists \at{U}, \at{V} \in \Bunches(\At)$ s.t. $\at{U} \mcomma \at{V} \bunchWeakerThan \at{W}$} \\
        &\text{$\supp{\at{U}}{\baseX} \varSigma(\varphi \acomma \psi)$, and $\supp{\at{V}}{\baseX} \varTheta$, then $\supp{\at{R}(\at{W})}{\baseX} \chi$}
        \end{split}
        \end{equation*}
        Since $\supp{\at{Q}}{\baseC} \varTheta$, it follows that: 
        \begin{equation}
            \forall \baseX \baseGeq \baseC, \forall \at{U} \in \Bunches(\At), \text{ if} \supp{\at{U}}{\baseX}\varSigma(\varphi \acomma \psi), \text{ then} \supp{\at{R}(\at{U} \mmcomma \at{Q})}{\baseX} \chi
        \end{equation}
        That is, $\varSigma(\varphi \acomma \psi) \supp{\at{R}( (\cdot) \mmcomma \at{Q} )}{\baseC} \chi$. Then, apply the IH for $\varSigma(\cdot)$ to $\varSigma(\varphi \acomma \psi) \supp{\at{R}( (\cdot) \mmcomma \at{Q} )}{\baseC} \chi$ and $\varDelta \supp{}{} \varphi \land \psi$, it follows that $\varSigma(\varDelta) \supp{\at{R}( (\cdot) \mmcomma \at{Q} )}{\baseC} \chi$. Together with $\supp{\at{P}}{\baseC} \varSigma(\varDelta)$, it follows that $\supp{\at{R} (\at{P} \mmcomma \at{Q})}{\baseC} \chi$. So $\supp{\at{R}(\at{S})}{\baseC} \chi$, by $\at{P} \mcomma \at{Q} \bunchWeakerThan \at{S}$ and Proposition~\ref{prop:monotonicity}. 
    \end{itemize}
    \end{itemize}
    This completes the inductive proof. 
\end{proof}

\begin{proof}[Proof of Lemma~\ref{lem:soundness-key-general}.\ref{lem:soundness-key-or-general}]
    We proceed by induction on$\holeDepth{\varGamma(\cdot)}{}$. 
    \begin{itemize}[label = --]
        \item \textsc{base case.} $\holeDepth{\varGamma(\cdot)}{} = 0$. This means $\varGamma(\cdot)$ is of the form $(\cdot)$. The statement becomes: 
        \begin{equation*}\label{eq:general-and-elim-sound-basecase}
            \text{If } \phi \supp{\at{R}(\cdot)}{\baseB} \chi, \psi \supp{\at{R}(\cdot)}{\baseB} \chi, \text{ and } \varDelta \supp{}{} \phi \lor \psi, \text{ then } \varDelta \supp{\at{R}(\cdot)}{\baseB} \chi. 
        \end{equation*}
        This follows immediately from Lemma~\ref{lem:soundness-key-base}.\ref{lem:soundness-key-or-base}. 
        \item \textsc{inductive step.} $\holeDepth{\varGamma(\cdot)}{} = d + 1$ for some $d \geq 0$. Then $\varGamma(\cdot)$ is of the form $\varSigma(\cdot) \commavar\varTheta$, where $\commavar \in \{ \aacomma, \mmcomma \}$, $\varSigma(\cdot) \in \BunchesWithHole(\At)$, $\varTheta \in \Bunches(\At)$, and $\holeDepth{\varSigma(\cdot)} = d$. We make a case distinction on $\commavar$. 
        \begin{itemize}[label = --]
            \item $\commavar$ is $\aacomma$. Given $\varSigma(\phi) \acomma \varTheta \supp{\at{R}(\cdot)}{\baseB} \chi$, $\varSigma(\psi) \acomma \varTheta  \supp{\at{R}(\cdot)}{\baseB} \chi$, and $\varDelta \supp{}{} \phi \lor \psi$, we assume $\supp{\at{P}}{\baseC} \varGamma(\varDelta)$ for some $\baseC \baseGeq \baseB$ and $\at{P} \in \Bunches(\At)$. The goal is to show $\supp{\at{R}(\at{P})}{\baseC} \chi$. First, $\supp{\at{P}}{\baseC} \varSigma(\varDelta) \acomma \varTheta$ means that there exist $\at{Q}_1, \at{Q}_2 \in \Bunches(\At)$ satisfying $\at{Q}_1 \bunchWeakerThan \at{P}$, $\at{Q}_2 \bunchWeakerThan \at{P}$, such that $\supp{\at{Q}_1}{\baseC} \varSigma(\varDelta)$ and $\supp{\at{Q}_2}{\baseC} \varTheta$. 
            Next we show that $\varSigma(\phi) \supp{\at{R}( (\cdot) \aacomma \at{Q}_2 )}{\baseB} \chi$ and $\varSigma(\psi) \supp{\at{R}( (\cdot) \aacomma \at{Q}_2 )}{\baseB} \chi$ hold, which enables the application of IH. We only look at $\varSigma(\phi) \supp{\at{R}( (\cdot) \aacomma \at{Q}_2 )}{\baseB} \chi$, and that for $\psi$ is exactly the same by replacing $\phi$ with $\psi$. Spelling out the definition of $\varSigma(\phi) \acomma \varTheta \supp{\at{R}(\cdot)}{\baseB} \chi$, we have: 
            \begin{equation}\label{eq:general-and-elim-sound-step-1}
            \begin{split}
                & \forall \baseX \baseGeq \baseB,  \forall \at{W} \in \Bunches(\At), \text{ if } \exists \at{U}, \at{V} \in \Bunches(\At) \text{ s.t. } \at{U} \bunchWeakerThan \at{W}, \at{V} \bunchWeakerThan \at{W},  \\
                & \supp{\at{U}}{\baseX} \varSigma(\phi), \text{ and} \supp{\at{V}}{\baseX} \varTheta, \text{ then} \supp{\at{R}(\at{W})}{\baseX} \chi 
            \end{split}
            \end{equation}
            Given $\baseC \baseGeq \baseB$ and $\supp{\at{Q}_2}{\baseC} \varTheta$, \eqref{eq:general-and-elim-sound-step-1} implies: 
            \begin{equation}
                \forall \baseY \baseGeq \baseC, \forall \at{U} \in \Bunches(\At), \text{ if } \supp{\at{U}}{\baseY} \varSigma(\phi), \text{ then } \supp{\at{R}( \at{U} \aacomma \at{Q}_2 )}{\baseY} \chi
            \end{equation}
            This is precisely $\varSigma(\phi) \supp{\at{R}((\cdot) \aacomma \at{Q}_2) }{\baseC} \chi$. 
        
            Now with $\varDelta \supp{}{} \phi \lor \psi$, $\varSigma(\phi) \supp{\at{R}((\cdot) \aacomma \at{Q}_2)}{\baseC} \chi$, and $\varSigma(\psi) \supp{\at{R}((\cdot) \aacomma \at{Q}_2)}{\baseC} \chi$ in hand, we can apply IH to conclude $\varSigma(\varDelta) \supp{ \at{R}((\cdot) \aacomma \at{Q}_2) }{\baseC} \chi$. Together with $\supp{\at{Q}_1}{\baseC} \varSigma(\varDelta)$, it follows that $\supp{ \at{R}( \at{Q}_1 \aacomma \at{Q}_2 ) }{\baseC} \chi$. Since $\at{Q}_1 \acomma \at{Q}_2 \bunchWeakerThan \at{P} \acomma \at{P}$, by Proposition~\ref{prop:monotonicity} we have $\supp{\at{R}(\at{P} \aacomma \at{P})}{\baseC} \chi$, and it implies $\supp{\at{R}(\at{P})}{\baseC} \chi$, by Proposition~\ref{prop:supp-resource-weaken-contract-deep}. 
            \item $\commavar$ is $\mmcomma$. Given $\varSigma(\phi) \mcomma \varTheta \supp{\at{R}(\cdot)}{\baseB} \chi$, $\varSigma(\psi) \mcomma \varTheta \supp{\at{R}(\cdot)}{\baseB} \chi$, and $\varDelta \supp{}{} \phi \lor \psi$, we assume $\supp{\at{P}}{\baseC} \varGamma(\varDelta)$ for some $\baseC \baseGeq \baseB$ and $\at{P} \in \Bunches(\At)$. The goal is to show $\supp{\at{R}(\at{P})}{\baseC} \chi$. First, $\supp{\at{P}}{\baseC} \varSigma(\varDelta) \mcomma \varTheta$ means that there exist $\at{Q}_1, \at{Q}_2 \in \Bunches(\At)$ satisfying $\at{Q}_1 \mcomma \at{Q}_2 \bunchWeakerThan \at{P}$, such that $\supp{\at{Q}_1}{\baseC} \varSigma(\varDelta)$ and $\supp{\at{Q}_2}{\baseC} \varTheta$. 
            Next we show that $\varSigma(\phi) \supp{\at{R}( (\cdot) \mmcomma \at{Q}_2 )}{\baseB} \chi$ and $\varSigma(\psi) \supp{\at{R}( (\cdot) \mmcomma \at{Q}_2 )}{\baseB} \chi$ hold, which enables the application of IH. We only look at $\varSigma(\phi) \supp{\at{R}( (\cdot) \mmcomma \at{Q}_2 )}{\baseB} \chi$, and that for $\psi$ is exactly the same by replacing $\phi$ with $\psi$. Spelling out the definition of $\varSigma(\phi) \mcomma \varTheta \supp{\at{R}(\cdot)}{\baseB} \chi$, we have: 
            \begin{equation}\label{eq:general-and-elim-sound-step-2}
            \begin{split}
                & \forall \baseX \baseGeq \baseB,  \forall \at{W} \in \Bunches(\At), \text{ if } \exists \at{U}, \at{V} \in \Bunches(\At) \text{ s.t. } \at{U} \mcomma \at{V} \bunchWeakerThan \at{W},  \\
                & \supp{\at{U}}{\baseX} \varSigma(\phi), \text{ and} \supp{\at{V}}{\baseX} \varTheta, \text{ then} \supp{\at{R}(\at{W})}{\baseX} \chi 
            \end{split}
            \end{equation}
            Given $\baseC \baseGeq \baseB$ and $\supp{\at{Q}_2}{\baseC} \varTheta$, \eqref{eq:general-and-elim-sound-step-2} implies: 
            \begin{equation}
                \forall \baseY \baseGeq \baseC, \forall \at{U} \in \Bunches(\At), \text{ if } \supp{\at{U}}{\baseY} \varSigma(\phi), \text{ then } \supp{\at{R}( \at{U} \mmcomma \at{Q}_2 )}{\baseY} \chi
            \end{equation}
            This is precisely $\varSigma(\phi) \supp{\at{R}((\cdot) \mmcomma \at{Q}_2) }{\baseC} \chi$. 
        
            Now with $\varDelta \supp{}{} \phi \lor \psi$, $\varSigma(\phi) \supp{\at{R}((\cdot) \mmcomma \at{Q}_2)}{\baseC} \chi$, and $\varSigma(\psi) \supp{\at{R}((\cdot) \mmcomma \at{Q}_2)}{\baseC} \chi$ in hand, we can apply IH to conclude $\varSigma(\varDelta) \supp{ \at{R}((\cdot) \mmcomma \at{Q}_2) }{\baseC} \chi$. Together with $\supp{\at{Q}_1}{\baseC} \varSigma(\varDelta)$, it follows that $\supp{ \at{R}( \at{Q}_1 \mmcomma \at{Q}_2 ) }{\baseC} \chi$, so $\supp{\at{R}(\at{P})}{\baseC} \chi$ by $\at{P} \bunchStrongerThan \at{Q}_1 \mcomma \at{Q}_2$ and Proposition~\ref{prop:monotonicity}. 
            \end{itemize}
        \end{itemize}
        This completes the inductive proof. 
\end{proof}

\begin{proof}[Proof of Lemma~\ref{lem:soundness-key-general}.\ref{lem:general-bot-elim-sound}]
    We proceed by induction on the structure of $\phi$. 
    \begin{itemize}[label = --]
        \item When $\phi$ is atomic, the desired statement follows immediately from \eqref{cl:BI-BeS-supp:bot}. 
        \item $\phi = \sigma \mand \tau$. We assume some $\baseC \baseGeq \baseB$ and $\at{P}$ such that $\supp{\at{P}}{\baseC} \bot$. The goal is to show $\supp{\at{R}(\at{P})}{\baseC} \sigma \mand \tau$; in other words, for arbitrary $\baseX \baseGeq \baseC$, $\at{U}(\cdot) \in \BunchesWithHole(\At)$, and $\at{p}$, if $\sigma \mcomma \tau \supp{\at{U}(\cdot)}{\baseX} \at{p}$, then $\supp{\at{U}(\at{R}(\at{P}))}{\baseX} \at{p}$. So we fix some $\baseD \baseGeq \baseC$, $\at{T} \in \BunchesWithHole(\At)$, and $\at{q} \in \At$ such that $\sigma \mcomma \tau \supp{\at{T}(\cdot)}{\baseD} \at{q}$. The goal is then to show $\supp{\at{T}(\at{R}(\at{P}))}{\baseD} \at{q}$. By \eqref{cl:BI-BeS-supp:bot}, this follows immediately from $\supp{\at{P}}{\baseC} \bot$ and $\baseD \baseGeq \baseC$. 
        \item $\phi = \mtop$. Given some $\baseC \baseGeq \baseB$ and $\at{P} \in \Bunches(\At)$ such that $\supp{\at{P}}{\baseC} \bot$, we prove $\supp{\at{R}(\at{P})}{\baseC} \mtop$. So we further fix some $\baseD \baseGeq \baseC$, $\at{T}(\cdot) \in \BunchesWithHole(\At)$, and $\at{q} \in \At$ such that $\supp{\at{T}(\at{R}(\munit))}{\baseD} \at{q}$, and prove $\supp{\at{T}(\at{R}(\at{P}))}{\baseD} \at{q}$. Again, this follows immediately from $\supp{\at{P}}{\baseC} \bot$ by \eqref{cl:BI-BeS-supp:bot}. 
        \item $\phi = \sigma \mto \tau$. Given some $\baseC \baseGeq \baseB$ and $\at{P} \in \Bunches(\At)$ such that $\supp{\at{P}}{\baseC} \bot$, we prove $\supp{\at{R}(\at{P})}{\baseC} \sigma \mto \tau$; by \eqref{cl:BI-BeS-supp:implication}, we prove $\sigma \supp{\at{R}(\at{P}) \mmcomma (\cdot)}{\baseC}\tau$. So, assume further some $\baseD \baseGeq \baseC$ and $\at{Q} \in \Bunches(\At)$ satisfying $\supp{\at{Q}}{\baseD} \sigma$, the goal is to prove $\supp{\at{R}(\at{P}) \mmcomma \at{Q}}{\baseD} \tau$. We apply IH for $\tau$ which says $\bot \supp{\at{R}(\cdot) \mmcomma \at{Q}}{\baseC} \tau$, and with $\supp{\at{P}}{\baseC} \bot$ we have $\supp{\at{R}(\at{P}) \mmcomma \at{Q}}{\baseD} \tau$. 
        \item The case for $\land$, $\top$, and $\to$ are similar to that for $\mand$, $\mtop$, and $\mto$, respectively, so we do not elaborate here. 
        \item $\phi = \sigma \lor \psi$. Given some $\baseC \baseGeq \baseB$ and $\at{P} \in \Bunches(\At)$ such that $\supp{\at{P}}{\baseC} \bot$, we prove $\supp{\at{R}(\at{P})}{\baseC} \sigma \lor \tau$. To this end, we further fix some arbitrary $\baseD \baseGeq \baseC$, $\at{T}(\cdot) \in \BunchesWithHole(\At)$, and $\at{q} \in \At$ satisfying $\phi \supp{\at{T}(\cdot)}{\baseD} \at{q}$, $\psi \supp{\at{T}(\cdot)}{\baseD} \at{q}$, and the goal is to show $\supp{\at{T}(\at{R}(\at{P}))}{\baseD} \at{q}$. 
        This follows immediately from $\supp{\at{P}}{\baseC} \bot$ by \eqref{cl:BI-BeS-supp:bot}. 
        \item $\phi = \bot$. This immediately follows from \eqref{cl:BI-BeS-supp:bot}. 
    \end{itemize}
    This completes the inductive proof. 
\end{proof}

The following results state the base cases necessary for the inductions in Lemma~\ref{lem:soundness-key-general}. 
\begin{lemma}\label{lem:soundness-key-base}
The following hold:
\begin{enumerate}
    \item  \label{lem:soundness-key-mtop-base}    If $\munit \supp{\at{T}(\cdot)}{\baseB} \chi$ and $\varDelta \supp{\at{S}(\cdot)}{\baseB} \mtop$, then $\varDelta \supp{\at{T}(\at{S}(\cdot))}{\baseB} \chi$.
    \item \label{lem:soundness-key-mand-base}
    If $\varDelta \supp{\at{S}(\cdot)}{\baseB} \phi \mand \psi$ and $\phi \mcomma \psi \supp{\at{T}(\cdot)}{\baseB} \chi$, then $\varDelta \supp{\at{T}(\at{S}(\cdot))}{\baseB} \chi$.
    \item \label{lem:soundness-key-and-base}
    If $\varDelta \supp{\at{S}(\cdot)}{\baseB} \phi \land \psi$ and $\phi \acomma \psi \supp{\at{T}(\cdot)}{\baseB} \chi$, then $\varDelta \supp{\at{T}(\at{S} (\cdot))}{\baseB} \chi$. 
    \item \label{lem:soundness-key-or-base}
    If $\varDelta \supp{\at{S}(\cdot)}{\baseB} \phi \lor \psi$, $\phi \supp{\at{T}(\cdot)}{\baseB} \chi$, and $\psi \supp{\at{T}(\cdot)}{\baseB} \chi$, then $\varDelta \supp{\at{T}(\at{S}(\cdot))}{\baseB} \chi$.
    \end{enumerate}
\end{lemma}

\begin{proof}[Proof of Lemma~\ref{lem:soundness-key-base}.\ref{lem:soundness-key-mtop-base}]
    Assume \textsc{(A)} $\munit \supp{\at{T}(\cdot)}{\baseB} \chi$ and \textsc{(B)} $\varDelta \supp{\at{S}(\cdot)}{\baseB} \mtop$. We require to show \textsc{(C)} $\varDelta \supp{\at{T}(\at{S}(\cdot))}{\baseB} \chi$.

    \labelandtag{eq:soundness-key-mtop-assumption-1}{\textsc{A}}
    \labelandtag{eq:soundness-key-mtop-assumption-2}{\textsc{B}}
    \labelandtag{eq:soundness-key-mtop-conclusion-1}{\textsc{C}}
    
    We fix some $\baseC \baseGeq \baseB$ and $\at{P}$ such that $\supp{\at{P}}{\baseC} \varDelta$, and the goal is to show $\supp{\at{T}(\at{S}(\at{P}))}{\baseC} \chi$. From~\eqref{eq:soundness-key-mtop-assumption-2} and $\suppM{\baseC}{\at{P}} \Delta$, we have $\supp{\at{S}(\at{P})}{\baseC} \mtop$. 
    It suffices to prove the following for arbitrary $\at{Q} \in \Bunches(\At)$: 
    \begin{equation}
    \label{eq:soundness-key-mtop-1}
        \text{if } \munit \supp{\at{T}(\cdot)}{\baseB} \chi \text{ and} \supp{\at{Q}}{\baseC} \mtop, \text{ then} \supp{\at{T}(\at{Q})}{\baseC} \chi. 
    \end{equation}
    This is because by taking $\at{Q}$ to be $\at{S}(\at{P})$ in \eqref{eq:soundness-key-mtop-1}, one immediately has $\supp{\at{T}(\at{S}(\at{P}))}{\baseC} \chi$. 
    Furthermore, since $\supp{\at{U}}{\baseX} \munit$ iff $\at{U} \bunchStrongerThan \munit$, \eqref{eq:soundness-key-mtop-1} is equivalent to: 
    \begin{equation}
    \label{eq:soundness-key-mtop-2}
        \text{if } \supp{\at{T}(\at{U})}{\baseX} \chi \text{, for arbitrary } \baseX \baseGeq \baseB \text{ and } \at{U} \bunchStrongerThan \munit, \text{ and} \supp{\at{Q}}{\baseC} \mtop, \text{ then} \supp{\at{T}(\at{Q})}{\baseC} \chi. 
    \end{equation}
    Therefore, it suffices to show:
    \begin{equation}
    \label{eq:soundness-key-mtop-3}
     \text{ if} \supp{\at{T}(\munit)}{\baseX} \chi \text{, for arbitrary $\baseX \supseteq \baseB$, and} \supp{\at{Q}}{\baseC} \mtop, \text{ then} \supp{\at{T}(\at{Q})}{\baseC} \chi. 
    \end{equation}
    In other words, the goal is to show: if $\supp{\at{T}(\munit)}{\baseB} \chi$ and $\supp{\at{Q}}{\baseC} \mtop$, then $\supp{\at{T}(\at{Q})}{\baseC} \chi$.
    We prove \eqref{eq:soundness-key-mtop-3} by induction on the structure of $\chi$:
    %
    \begin{itemize}[label = --]
        \item $\chi$ is atomic, say it is some $\at{p} \in \At$. By definition, $\supp{\at{Q}}{\baseC} \mtop$ means for arbitrary $\baseX \baseGeq \baseC$, $\supp{\at{T}(\munit)}{\baseX} \at{p}$ implies $\supp{\at{T}(\at{Q})}{\baseX} \at{p}$. The desired result obtains by choosing $\baseX = \baseC$.

        \item $\chi = \sigma \land \tau$. In order to prove $\supp{\at{T}(\at{Q})}{\baseC} \sigma \land \tau$, we fix some arbitrary $\baseD \baseGeq \baseC$, $\at{R}(\cdot) \in \BunchesWithHole(\At)$ and $\at{q} \in \At$ such that $\sigma \acomma \tau \supp{\at{R}(\cdot)}{\baseD} \at{q}$, and the goal is to prove $\supp{\at{R}(\at{T}(\at{Q}))}{\baseD} \at{q}$. Now $\supp{\at{T}(\munit)}{\baseB} \sigma \land \tau$ together with $\sigma \acomma \tau \supp{\at{R}(\cdot)}{\baseD} \at{q}$ entail $\supp{\at{R}(\at{T}(\munit))}{\baseD} \at{q}$. This together with $\supp{\at{Q}}{\baseC} \mtop$ and $\baseD \baseGeq \baseC$ imply $\supp{\at{R}(\at{T}(\at{Q}))}{\baseD} \at{q}$. 
        \item $\chi = \top$. In order to prove $\supp{\at{T}(\at{Q})}{\baseC} \top$, we fix some arbitrary $\baseD \baseGeq \baseC$, $\at{R}(\cdot) \in \BunchesWithHole(\At)$, and $\at{q} \in \At$ satisfying $\supp{\at{R}(\aunit)}{\baseD} \at{q}$, and show that $\supp{\at{R}(\at{T}(\at{Q}))}{\baseD} \at{q}$.
        $\supp{\at{T}(\munit)}{\baseB} \top$ together with $\supp{\at{R}(\aunit)}{\baseD} \at{q}$ imply $\supp{\at{R}(\at{T}(\munit))}{\baseD} \at{q}$. This together with $\supp{\at{Q}}{\baseC} \mtop$ imply $\supp{\at{R}(\at{T}(\at{Q}))}{\baseD} \at{q}$ by \eqref{cl:BI-BeS-supp:mtop}. 
        \item $\chi = \sigma \to \tau$. We apply \eqref{cl:BI-BeS-supp:implication} to \eqref{eq:soundness-key-mtop-3} and get: if $\sigma \supp{\at{T}(\munit) \aacomma (\cdot)}{\baseX} \tau$, for arbitrary $\baseX \baseGeq \baseB$, and $\supp{\at{Q}}{\baseC} \mtop$, then $\sigma \supp{\at{T}(\at{Q}) \aacomma (\cdot)}{\baseC} \tau$.
        Given the assumptions, to prove $\sigma \supp{\at{T}(\at{Q}) \aacomma (\cdot)}{\baseC} \tau$, we fix some $\baseD \baseGeq \baseC$ and $\at{R} \in \Bunches(\At)$ such that $\supp{\at{R}}{\baseD} \sigma$, and show $\supp{\at{T}(\at{Q}) \aacomma \at{R}}{\baseD} \tau$.
        From $\sigma \supp{\at{T}(\munit) \aacomma (\cdot)}{\baseX} \tau$ and $\supp{\at{R}}{\baseD} \sigma$, we infer $\supp{\at{T}(\munit) \aacomma \at{R}}{\baseD} \tau$. From this and $\supp{\at{Q}}{\baseC} \mtop$, we use IH on $\tau$ to infer $\supp{\at{T}(\at{Q}) \aacomma \at{R}}{\baseD} \tau$.
        \item $\chi = \sigma \mand \tau$. Given the assumptions, in order to show $\supp{\at{T}(\at{Q})}{\baseC} \sigma \mand \tau$, we fix some arbitrary $\baseD \baseGeq \baseC$, $\at{R}(\cdot) \in \BunchesWithHole(\At)$, and $\at{q} \in \At$ such that $\sigma \mcomma \tau \supp{\at{R}(\cdot)}{\baseD} \at{q}$, and show $\supp{\at{R(T(Q))}}{\baseD} \at{q}$.
        From $\supp{\at{T}(\munit)}{\baseB} \sigma \mand \tau$ and $\sigma\mcomma \tau \supp{\at{R(\cdot)}}{\baseD} \at{q}$, infer $\supp{\at{R(T(\munit))}}{\baseD} \at{q}$. From this and $\supp{\at{Q}}{\baseC} \mtop$, we have $\supp{\at{R(T(Q))}}{\baseD} \at{q}$ by \eqref{cl:BI-BeS-supp:mtop}.
        \item $\chi = \mtop$. 
        Given the assumptions, to prove $\supp{\at{T(Q)}}{\baseC} \mtop$, we fix some $\baseD \baseGeq \baseC$, $\at{R(\cdot)} \in \BunchesWithHole(\At)$, and $\at{q} \in \At$ such that $\supp{\at{R}(\munit)}{\baseD} \at{q}$, and prove that $\supp{\at{R}(\at{T}(\at{Q}))}{\baseD} \at{q}$.
        From $\supp{\at{R}(\munit)}{\baseD} \at{q}$ and $\supp{\at{T}(\munit)}{\baseC} \mtop$, infer $\supp{\at{R}(\at{T}(\munit))}{\baseD} \at{q}$ by \eqref{cl:BI-BeS-supp:mtop}. From this and $\supp{\at{Q}}{\baseC} \mtop$, infer $\supp{\at{R}(\at{T}(\at{Q}))}{\baseD} \at{q}$, by \eqref{cl:BI-BeS-supp:mtop} again.
        \item $\chi = \sigma \mto \tau$.
        Apply \eqref{cl:BI-BeS-supp:wand} to \eqref{eq:soundness-key-mtop-3}, then the goal becomes: if $\sigma \supp{\at{T}(\munit) \mmcomma (\cdot)}{\baseB} \tau$ and $\supp{\at{Q}}{\baseC} \mtop$, then $\sigma \supp{\at{T}(\at{Q}) \mmcomma (\cdot)}{\baseC} \tau$. 
        Assume the assumptions, to show $\sigma \supp{\at{T}(\at{Q}) \mmcomma (\cdot)}{\baseC} \tau$, we fix some $\baseD \baseGeq \baseC$ and $\at{R} \in \Bunches(\At)$ such that $\supp{\at{R}}{\baseD} \sigma$, and show $\supp{\at{T}(\at{Q}) \mmcomma \at{R}}{\baseD} \tau$. 
        From $\supp{\at{R}}{\baseD} \sigma$ and $\sigma \supp{\at{T}(\munit) \mmcomma (\cdot)}{\baseB} \tau$, infer $\supp{\at{T}(\munit) \mmcomma \at{R}}{\baseD} \tau$. From this and $\supp{\at{Q}}{\baseC} \mtop$, we infer $\supp{\at{T}(\at{Q}) \mmcomma \at{R}}{\baseD} \tau$, by IH on $\tau$.
        \item $\chi = \sigma \lor \tau$.
        Given the assumptions, in order to show $\supp{\at{T}(\at{Q})}{\baseC} \sigma \lor \tau$, we fix some $\baseD \baseGeq \baseC$, $\at{R}(\cdot) \in \Bunches(\At)$, and $\at{q} \in \At$ satisfying $\sigma \supp{\at{R}(\cdot)}{\baseD} \at{q}$ and $\tau \supp{\at{R}(\cdot)}{\baseD} \at{q}$, and show that $\supp{\at{R}(\at{T}(\at{Q}))}{\baseD} \at{q}$. 
        From $\supp{\at{T}(\munit)}{\baseB} \sigma \lor \tau$, $\sigma \supp{\at{R}(\cdot)}{\baseD} \at{q}$ and $\tau \supp{\at{R}(\cdot)}{\baseD} \at{q}$, it follows that $\supp{\at{R}(\at{T}(\munit))}{\baseD} \at{q}$ by \eqref{cl:BI-BeS-supp:disjunction}. From this and $\supp{\at{Q}}{\baseC} \mtop$, we conclude that $\supp{\at{R}(\at{T}(\at{Q}))}{\baseE} \at{q}$, by \eqref{cl:BI-BeS-supp:mtop}. 
        \item $\chi = \bot$. 
        Given the assumptions, in order to show $\supp{\at{T(Q)}}{\baseC} \bot$, in order to show, we simply fix some arbitrary $\at{R}(\cdot) \in \BunchesWithHole(\At)$ and $\at{q} \in \At$, and then show $\supp{\at{R}(\at{T}(\at{Q}))}{\baseD} \at{q}$. 
        From $\supp{\at{T}(\munit)}{\baseB} \bot$, infer $\supp{\at{R}(\at{T}(\munit))}{\baseC} \at{q}$ by~\eqref{cl:BI-BeS-supp:bot}.
        From this and $\supp{\at{Q}}{\baseC} \mtop$, we conclude that $\supp{\at{R}(\at{T}(\at{Q}))}{\baseD} \at{q}$, by \eqref{cl:BI-BeS-supp:mtop}. 
    \end{itemize}
    This completes the inductive proof. 
\end{proof}

\begin{proof}[Proof of Lemma~\ref{lem:soundness-key-base}.\ref{lem:soundness-key-mand-base}]
    By induction on the structure of $\chi$. The proof follows the same line of reasoning of the additive counterpart, see Lemma~\ref{lem:soundness-key-base}.\ref{lem:soundness-key-and-base}. 
\end{proof}

\begin{proof}[Proof of Lemma~\ref{lem:soundness-key-base}.\ref{lem:soundness-key-and-base} ]

     Assume \textsc{(A)} $\varDelta \supp{\at{S}(\cdot)}{\baseB} \phi \land \psi$ and \textsc{(B)} $\phi \acomma \psi \supp{\at{T}(\cdot)}{\baseB} \chi$. We require to show \textsc{(C)} $\varDelta \supp{\at{T}(\at{S} (\cdot))}{\baseB} \chi$. 

    \labelandtag{item:soundness-key-and-1}{\textsc{(A)}}
    \labelandtag{item:soundness-key-and-2}{\textsc{(B)}}
    \labelandtag{item:soundness-key-and-conclusion}{\textsc{(C)}}

    We fix some $\baseC \baseGeq \baseB$ and $\at{P} \in \Bunches(\At)$ such that $\supp{\at{P}}{\baseC} \varDelta$, and prove $\supp{\at{T(S(P))}}{\baseC} \chi$ by induction on the structure of $\chi$. Note that $\supp{\at{P}}{\baseC} \varDelta$ and \eqref{item:soundness-key-and-1} entail $\supp{\at{S(P)}}{\baseC} \varphi \land \psi$. 
    \begin{itemize}[label = --]
        \item $\chi$ is atomic. This follows immediately from \eqref{cl:BI-BeS-supp:At}. 
        \item $\chi = \sigma \land \tau$. Spelling out the definition of the goal $\supp{\at{T(S(P))}}{\baseC} \sigma \land \tau$, it suffices to fix arbitrary $\baseD \baseGeq \baseC$, $\at{R}(\cdot) \in \BunchesWithHole(\At)$, and $\at{q} \in \At$ such that $\sigma \acomma \tau \supp{\at{R}(\cdot)}{\baseD} \at{q}$, and show $\supp{\at{R(T(S(P)))}}{\baseD} \at{q}$. Apply~\eqref{cl:BI-BeS-supp:inf} and~\eqref{cl:BI-BeS-supp:conjunction} to        ~\ref{item:soundness-key-and-2} together with $\sigma \acomma \tau \supp{\at{R}(\cdot)}{\baseD} \at{q}$, we have: 
        \begin{equation*}
            \forall \baseY \baseGeq \baseD, \forall \at{U} \in \Bunches(\At),  \text{ if} \supp{\at{U}}{\baseY} \varphi \acomma \psi, \text{ then} \supp{\at{R(T(U))}}{\baseY} \at{q} 
        \end{equation*}
        That is, $\varphi \acomma \psi \supp{\at{R(T(\cdot))}}{\baseD} \at{q}$. By the IH, this and $\supp{\at{S(P)}}{\baseC} \varphi \land \psi$ imply $\supp{\at{R(T(S(P)))}}{\baseD} \at{q}$. 
        \item $\chi = \top$. According to \eqref{cl:BI-BeS-supp:top}, to prove $\supp{\at{T(S(P))}}{\baseC} \top$, we fix some arbitrary $\baseD \baseGeq \baseC$, $\at{R}(\cdot) \in \BunchesWithHole(\At)$, and $\at{q} \in \At$ such that $\supp{\at{R}(\aunit)}{\baseD} \at{q}$, and the goal is to prove $\supp{\at{R(T(S(P)))}}{\baseD} \at{q}$.~\ref{item:soundness-key-and-2} together with $\supp{\at{R}(\aunit)}{\baseD} \at{q}$ implies: 
        \begin{equation*}
            \forall \baseY \baseGeq \baseD, \forall \at{U} \in \Bunches(\At), \text{ if} \supp{\at{U}}{\baseY} \varphi \acomma \psi, \text{ then} \supp{\at{R(T(U))}}{\baseY} 
            \at{q} 
        \end{equation*}
        That is, $\varphi \acomma \psi \supp{\at{R(T(\cdot))}}{\baseD} \at{q}$. This together with $\supp{\at{S(P)}}{\baseB} \varphi \land \psi$ imply $\supp{\at{R(T(S(P)))}}{\baseD} \at{q}$, by the IH. 
        \item $\chi = \sigma \to \tau$. The goal $\supp{\at{T(S(P))}}{\baseC} \sigma \to \tau$ is equivalent to $\sigma \supp{\at{T(S(P))} \aacomma (\cdot)}{\baseC} \tau$. So we fix some $\baseD \baseGeq \baseC$ and $\at{Q} \in \Bunches(\At)$ such that $\supp{\at{Q}}{\baseD} \sigma$, and show $\supp{\at{T(S(P))} \aacomma \at{Q}}{\baseD} \tau$.~\ref{item:soundness-key-and-2} $\varphi \acomma \psi \supp{\at{T}(\cdot)}{\baseB} \sigma \to \tau$ together with $\supp{\at{Q}}{\baseD} \sigma$ says $\varphi \acomma \psi \supp{\at{T}(\cdot) \aacomma \at{Q}}{\baseD} \tau$, and together with $\supp{\at{S(P)}}{\baseB} \varphi \land \psi$, it follows that $\supp{\at{T(S(P))} \aacomma \at{Q}}{\baseD} \tau$, by the IH. 
        \item $\chi = \sigma \mand \tau$. 
        To prove $\supp{\at{T}(\at{S}(\at{P}))}{\baseC} \sigma \mand \tau$, we further fix some $\baseD \baseGeq \baseC$, $\at{R}(\cdot) \in \BunchesWithHole(\At)$, and $\at{q} \in \At$ such that $\sigma \mcomma \tau \supp{\at{R}(\cdot)}{\baseD} \at{q}$. The goal is to prove $\supp{\at{R}( \at{T}(\at{S}(\at{P})) )}{\baseD} \at{q}$. Now $\sigma \mcomma \tau \supp{\at{R}(\cdot)}{\baseD} \at{q}$ and~\ref{item:soundness-key-and-2} imply $\varphi \acomma \psi \supp{\at{R}(\at{T}(\cdot))}{\baseD} \at{q}$. This together with $\supp{\at{S}(\at{P})}{\baseC} \varphi \land \psi$ entail $\supp{ \at{R} (\at{T} (\at{S} (\at{P}) ) ) }{\baseD} \at{q}$, according to \eqref{cl:BI-BeS-supp:conjunction}. 
        \item $\chi = \mtop$. Similar to the $\top$ case. 
        \item $\chi = \sigma \mto \tau$. Similiar to the $\to$ case. 
        \item $\chi = \sigma \lor \tau$. In order to show $\supp{\at{T(S(P))}}{\baseC} \sigma \lor \tau$, we fix some $\baseD \baseGeq \baseC$, $\at{R}(\cdot) \in \BunchesWithHole(\At)$, and $\at{q} \in \At$ such that $\sigma \supp{\at{R}(\cdot)}{\baseD} \at{q}$ and $\tau \supp{\at{R}(\cdot)}{\baseD} \at{q}$. The goal is $\supp{\at{R(T(S(P)))}}{\baseD} \at{q}$. Note that $\varphi \acomma \psi \supp{\at{T}(\cdot)}{\baseB} \sigma \lor \tau$ together with $\sigma \supp{\at{R}(\cdot)}{\baseD} \at{q}$ and $\tau \supp{\at{R}(\cdot)}{\baseD} \at{q}$ imply $\varphi \acomma \psi \supp{\at{R(T(\cdot))}}{\baseD} \at{q}$. Apply the IH to $\supp{\at{S(P)}}{\baseC} \varphi \land \psi$ and $\varphi \acomma \psi \supp{\at{R(T(\cdot))}}{\baseD} \at{q}$, it follows that $\supp{\at{R(T(S(P)))}}{\baseD} \at{q}$. 
        \item $\chi = \bot$.
        In order to show $\supp{\at{T(S(P))}}{\baseC} \bot$, we fix some arbitrary $\baseD \baseGeq \baseC$, $\at{U}(\cdot) \in \BunchesWithHole(\At)$ and $\at{q} \in \At$, and show $\supp{\at{U(T(S(P)))}}{\baseD} \at{q}$.
        Given $\varphi \acomma \psi \supp{\at{T(\cdot)}}{\baseB} \bot$ and $\baseD \baseGeq \baseB$, we have $\varphi \acomma \psi \supp{\at{U(T(\cdot))}}{\baseD} \at{q}$ by (\ref{cl:BI-BeS-supp:inf}) and (\ref{cl:BI-BeS-supp:bot}). Apply the IH to $\supp{\at{S(P)}}{\baseC} \varphi \land \psi$ and $\varphi \acomma \psi \supp{\at{U(T(\cdot))}}{\baseD} \at{q}$, we conclude $\supp{\at{U(T(S(P)))}}{\baseD} \at{q}$.
    \end{itemize}
    This completes the induction. 
\end{proof}

\begin{proof}[Proof of Lemma~\ref{lem:soundness-key-base}.\ref{lem:soundness-key-or-base}]
    Assume \textsc{(A)}~$\varDelta \supp{\at{S}(\cdot)}{\baseB} \phi \lor \psi$, \textsc{(B)}~$\phi \supp{\at{T}(\cdot)}{\baseB} \chi$, and \textsc{(C)}~$\psi \supp{\at{T}(\cdot)}{\baseB} \chi$. We require to show \textsc{(D)}~$\varDelta \supp{\at{T}(\at{S}(\cdot))}{\baseB} \chi$.    We proceed by induction on the structure of $\chi$. 
    
      \labelandtag{eq:soundness-key-or-assumption-1}{\textsc{(A)}}
    \labelandtag{eq:soundness-key-or-assumption-2}{\textsc{(B)}}
    \labelandtag{eq:soundness-key-or-assumption-3}{\textsc{(C)}}
    \labelandtag{eq:soundness-key-or-conclusion-1}{\textsc{(D)}}
    
    \begin{itemize}[label=--]
        \item $\chi$ is an atom. This follows immediately from the definition \eqref{cl:BI-BeS-supp:disjunction}. 
        \item $\chi = \sigma \land \tau$. We fix some $\baseC \baseGeq \baseB$ and $\at{P}$ such that $\supp{\at{P}}{\baseC} \varDelta$. It then suffices to show $\supp{\at{T}(\at{S}(\at{P}))}{\baseC} \sigma \land \tau$ (under~\ref{eq:soundness-key-or-assumption-1},~\ref{eq:soundness-key-or-assumption-2},~\ref{eq:soundness-key-or-assumption-3}). 
        Towards this, we further fix some $\baseD \baseGeq \baseC$, $\at{Q}(\cdot) \in \BunchesWithHole(\At)$, and $\at{q} \in \At$ such that $\sigma \acomma \tau \supp{\at{Q}(\cdot)}{\baseD} \at{q}$, and show that $\supp{\at{Q}(\at{T}(\at{S}(\at{P})))}{\baseD} \at{q}$. 
        Given~\ref{eq:soundness-key-or-assumption-1}, we have $\supp{\at{S}(\at{P})}{\baseC} \phi \lor \psi$. 
        Note that~\ref{eq:soundness-key-or-assumption-2} and $\sigma \acomma \tau \supp{\at{Q}(\cdot)}{\baseD} \at{q}$ imply $\phi \supp{\at{Q}(\at{T}(\cdot))}{\baseD} \at{q}$;~\ref{eq:soundness-key-or-assumption-3} and $\sigma \acomma \tau \supp{\at{Q}(\cdot)}{\baseD} \at{q}$ imply $\psi \supp{\at{Q}(\at{T}(\cdot))}{\baseD} \at{q}$. 
        So apply the IH to $\supp{\at{S}(\at{P})}{\baseD} \phi \lor \psi$, $\phi \supp{\at{Q}(\at{T}(\cdot))}{\baseD} \at{q}$, and $\psi \supp{\at{Q}(\at{T}(\cdot))}{\baseD} \at{q}$, we get $\supp{\at{Q}(\at{T}(\at{S}(\at{P})))}{\baseD} \at{q}$. 
        %
        \item $\chi = \top$. We fix some $\baseC \baseGeq \baseB$ and $\at{P}$ such that $\supp{\at{P}}{\baseC} \varDelta$, and it suffices to show $\supp{\at{T}(\at{S}(\at{P}))}{\baseC} \top$. Towards this, we fix some $\baseD \baseGeq \baseC$, $\at{Q}(\cdot) \in \BunchesWithHole(\At)$, and $\at{q} \in \At$ such that $\supp{\at{Q}(\aunit)}{\baseD} \at{q}$, and the goal is to show $\supp{\at{Q}(\at{T}(\at{S}(\at{P})))}{\baseD} \at{q}$. 
        $\supp{\at{P}}{\baseC} \varDelta$ and~\ref{eq:soundness-key-or-assumption-1} entail $\supp{\at{S}(\at{P})}{\baseC} \phi \lor \psi$. 
        Note that $\supp{\at{Q}(\aunit)}{\baseD} \at{q}$ and~\ref{eq:soundness-key-or-assumption-2} imply $\phi \supp{\at{Q}(\at{T}(\cdot))}{\baseD} \at{q}$: for arbitrary $\baseX \baseGeq \baseD$, $\at{U} \in \Bunches(\At)$ such that $\supp{\at{U}}{\baseX} \phi$, by~\ref{eq:soundness-key-or-assumption-2} we have $\supp{\at{T}(\at{U})}{\baseX} \top$, so $\supp{\at{Q}(\aunit)}{\baseD} \at{q}$ implies $\supp{\at{Q}(\at{T}(\at{U}))}{\baseD} \at{q}$. 
        Similarly, $\supp{\at{Q}(\aunit)}{\baseD} \at{q}$ and~\ref{eq:soundness-key-or-assumption-3} imply $\psi \supp{\at{Q}(\at{T}(\cdot))}{\baseD} \at{q}$. Apply the IH to $\supp{\at{S}(\at{P})}{\baseC} \phi \lor \psi$, $\supp{\at{Q}(\at{T}(\cdot))}{\baseD} \at{q}$, and $\psi \supp{\at{Q}(\at{T}(\cdot))}{\baseD} \at{q}$, we have $\phi \supp{\at{Q}(\at{T}( \at{S}(\at{P}) ))}{\baseD} \at{q}$. 
        \item  $\chi = \sigma \to \tau$. We fix some $\baseC \baseGeq \baseB$ and $\at{P}$ such that $\supp{\at{P}}{\baseC} \varDelta$. The goal is to show $\supp{\at{T}(\at{S}(\at{P}))}{\baseC} \sigma \to \tau$; that is, $\sigma \supp{\at{T}(\at{S}(\at{P})) \aacomma (\cdot)}{\baseC} \tau$; that is, for arbitrary $\baseX \baseGeq \baseC$ and $\at{U} \in \Bunches(\At)$, if $\supp{\at{U}}{\baseX} \sigma$, then $\supp{\at{T}(\at{S}(\at{P})) \aacomma \at{U}}{\baseX} \tau$. So fix some $\baseD \baseGeq \baseC$ and $\at{Q} \in \Bunches(\At)$ such that $\supp{\at{Q}}{\baseD} \sigma$, we show that $\supp{\at{T}(\at{S}(\at{P})) \aacomma \at{Q}}{\baseD} \tau$. By the IH, it suffices to show: 
        \begin{align}
            & \supp{\at{S}(\at{P})}{\baseD} \phi \lor \psi \label{lem:soundness-key-or-4} \\
            & \phi \supp{\at{T}(\cdot) \aacomma \at{Q}}{\baseD} \tau \label{lem:soundness-key-or-5} \\
            & \psi \supp{\at{T}(\cdot) \aacomma \at{Q}}{\baseD} \tau \label{lem:soundness-key-or-6} 
        \end{align}
        Here \eqref{lem:soundness-key-or-4} follows immediately from~\ref{eq:soundness-key-or-assumption-1}, $\supp{\at{P}}{\baseC} \varDelta$, and $\baseD \baseGeq \baseC \baseGeq \baseB$. For \eqref{lem:soundness-key-or-5} and \eqref{lem:soundness-key-or-6}, we simply show the former as the reasoning for the latter is similar. Using \eqref{cl:BI-BeS-supp:inf}, \eqref{lem:soundness-key-or-5} means for arbitrary $\baseY \baseGeq \baseD$ and $\at{U} \in \Bunches(\At)$, if $\supp{\at{U}}{\baseY} \phi$, then $\supp{\at{T}(\at{U}) \aacomma \at{Q}}{\baseY} \tau$. So we fix some $\baseE \baseGeq \baseD$ and $\at{R} \in \Bunches(\At)$ such that $\supp{\at{R}}{\baseE} \phi$, and show that $\supp{\at{T}(\at{R}) \aacomma \at{Q}}{\baseE} \tau$. Given $\supp{\at{R}}{\baseE} \phi$, $\phi \supp{\at{T}(\cdot)}{\baseB} \sigma \to \tau$, and $\baseE \baseGeq \baseB$, it follows that $\supp{\at{T}(\at{R})}{\baseE} \sigma \to \tau$; that is, $\sigma \supp{\at{T}(\at{R}) \aacomma (\cdot)}{\baseE} \tau$. This together with $\supp{\at{Q}}{\baseD} \sigma$ imply that $\supp{\at{T}(\at{R}) \aacomma \at{Q}}{\baseE} \tau$. 
        \item $\chi = \sigma \mand \tau$. Fix some $\baseC \baseGeq \baseB$ and $\at{P} \in \Bunches(\At)$ such that $\supp{\at{P}}{\baseC} \varDelta$, the goal is to show that $\supp{\at{T}(\at{S}(\at{P}))}{\baseC} \sigma \mand \tau$. By \eqref{cl:BI-BeS-supp:mand}, this amounts to fix some $\baseD \baseGeq \baseC$, $\at{Q}(\cdot) \in \BunchesWithHole(\At)$, and $\at{q} \in \At$ such that $\sigma \mcomma \tau \supp{\at{Q}(\cdot)}{\baseD} \at{q}$, and show that $\supp{\at{Q}(\at{T}(\at{S}(\at{P})))}{\baseD} \at{q}$. By the base case, it suffices to show that: 
        \begin{align}
            & \supp{\at{S}(\at{P})}{\baseD} \phi \lor \psi \label{lem:soundness-key-or-7} \\
            & \phi \supp{\at{Q}(\at{T}(\cdot))}{\baseD} \at{q} \label{lem:soundness-key-or-8} \\
            & \psi \supp{\at{Q}(\at{T}(\cdot))}{\baseD} \at{q} \label{lem:soundness-key-or-9} 
        \end{align}
        Here \eqref{lem:soundness-key-or-7} follows immediately from assumption~\ref{eq:soundness-key-or-assumption-1} and $\supp{\at{P}}{\baseC} \varDelta$; we only show \eqref{lem:soundness-key-or-9}, and the argument for \eqref{lem:soundness-key-or-8} is similar. 
        So we fix some $\baseE \baseGeq \baseD$ and $\at{R} \in \Bunches(\At)$ such that $\supp{\at{R}}{\baseE} \psi$, and show that $\supp{\at{Q}(\at{T}(\at{R}))}{\baseE} \at{q}$. Given $\supp{\at{R}}{\baseE} \psi$, $\psi \supp{\at{T}(\cdot)}{\baseB} \sigma \mand \tau$, and $\baseE \baseGeq \baseB$, it follows that $\supp{\at{T}(\at{R})}{\baseE} \sigma \mand \tau$. Together with $\sigma \mcomma \tau \supp{\at{Q}(\cdot)}{\baseD} \at{q}$ and $\baseE \baseGeq \baseD$, it follows by \eqref{cl:BI-BeS-supp:mand} that $\supp{\at{Q}( \at{T}(\at{R}) )}{\baseE} \at{q}$. 
        \item $\chi = \mtop$. It follows from the same line of reasoning for $\chi = \top$. 
        \item $\chi = \sigma \wand \tau$. Fix some $\baseC \baseGeq \baseB$ and $\at{P} \in \Bunches(\At)$ satisfying $\supp{\at{P}}{\baseC} \varDelta$, the goal is then to show that $\supp{\at{T}(\at{S}(\at{P}))}{\baseC} \sigma \wand \tau$; in other words, by \eqref{cl:BI-BeS-supp:wand} and \eqref{cl:BI-BeS-supp:inf}, to show that for arbitrary $\baseX \baseGeq \baseC$ and $\at{U}$, if $\supp{\at{U}}{\baseX} \sigma$, then $\supp{\at{T}(\at{S}(\at{P})) \mmcomma \at{U}}{\baseX} \tau$. 
        So let us fix some $\baseD \baseGeq \baseC$ and $\at{Q} \in \Bunches(\At)$ such that $\supp{\at{Q}}{\baseD} \sigma$, and show that $\supp{\at{T}(\at{S}(\at{P})) \mmcomma \at{Q}}{\baseD} \tau$. 
        By the IH, it suffices to show that: 
        \begin{align}
            & \supp{\at{S}(\at{P})}{\baseD} \phi \lor \psi \label{lem:soundness-key-or-10} \\
            & \phi \supp{\at{T}(\cdot) \mmcomma \at{Q}}{\baseD} \tau \label{lem:soundness-key-or-11} \\ 
            & \psi \supp{\at{T}(\cdot) \mmcomma \at{Q}}{\baseD} \tau \label{lem:soundness-key-or-12}
        \end{align}
        As \eqref{lem:soundness-key-or-10} follows immediately from~\ref{eq:soundness-key-or-assumption-1} $\varDelta \supp{\at{S}(\cdot)}{\baseB} \phi \lor \psi$, $\supp{\at{P}}{\baseC} \varDelta$, and $\baseD \baseGeq \baseC \baseGeq \baseB$, we show \eqref{lem:soundness-key-or-11} and \eqref{lem:soundness-key-or-12}. Since the latter two are similar, we only showcase \eqref{lem:soundness-key-or-11}. Towards this, it suffices to fix some $\baseE \baseGeq \baseD$ and $\at{R} \in \Bunches(\At)$ such that $\supp{\at{R}}{\baseE} \phi$, and show that $\supp{\at{T}(\at{R}) \mmcomma \at{Q}}{\baseE} \tau$. 
        It follows from~\ref{eq:soundness-key-or-assumption-2} $\phi \supp{\at{T}(\cdot)}{\baseB} \sigma \wand \tau$, $\supp{\at{R}}{\baseE} \phi$, and $\baseE \baseGeq \baseB$ that $\supp{\at{T}(\at{R})}{\baseE} \sigma \wand \tau$; that is, $\sigma \supp{\at{T}(\at{R}) \mmcomma (\cdot)}{\baseE} \tau$. This together with $\supp{\at{Q}}{\baseD} \sigma$ and $\baseE \baseGeq \baseD$ entail that 
        $\supp{\at{T}(\at{R}) \mmcomma \at{Q}}{\baseE} \tau$.  
        \item $\chi = \sigma \lor \tau$. 
        We fix some $\baseC \baseGeq \baseB$ and $\at{P}$ satisfying $\supp{\at{P}}{\baseC} \varDelta$, and the goal is $\supp{\at{T}(\at{S}(\at{P}))}{\baseC} \sigma \lor \tau$; which according to \eqref{cl:BI-BeS-supp:disjunction} boils down to showing that for arbitrary $\baseX \baseGeq \baseB$, $\at{U}(\cdot)$, $\at{p}$, if $\sigma \supp{\at{U}(\cdot)}{\baseX} \at{p}$ and $\tau \supp{\at{U}(\cdot)}{\baseX} \at{p}$, then $\supp{\at{U}(\at{T}(\at{S}))}{\baseX} \at{p}$. So we fix some $\baseC \baseGeq \baseB$, $\at{W}(\cdot) \in \BunchesWithHole(\At)$, $\at{q} \in \At$ satisfying that $\sigma \supp{\at{W}(\cdot)}{\baseC} \at{q}$ and $\tau \supp{\at{W}(\cdot)}{\baseC} \at{q}$, and show that $\supp{\at{W}(\at{T}(\at{S}))}{\baseC} \at{q}$. By the IH, it suffices to show that: 
        \begin{align}
            & \supp{\at{S}}{\baseC} \phi \lor \psi \label{lem:soundness-key-or-13} \\ 
            & \phi \supp{\at{W}(\at{T}(\cdot))}{\baseC} \at{q} \label{lem:soundness-key-or-14} \\ 
            & \psi \supp{\at{W}(\at{T}(\cdot))}{\baseC} \at{q} \label{lem:soundness-key-or-15} 
        \end{align}
        Again, \eqref{lem:soundness-key-or-13} follows from the combination of \ref{eq:soundness-key-or-assumption-1}, $\baseC \baseGeq \baseB$, and $\supp{P}{\baseC}\Delta$. We now prove \eqref{lem:soundness-key-or-15}, and omit the similar proof for \eqref{lem:soundness-key-or-14}. So we fix some $\baseD \baseGeq \baseC$ and $\at{Q} \in \Bunches(\At)$ such that $\supp{\at{Q}}{\baseD} \psi$, and show that $\supp{\at{W}(\at{T}(\at{Q}))}{\baseD} \at{q}$. It follows from~\ref{eq:soundness-key-or-assumption-3} $\psi \supp{\at{T}(\cdot)}{\baseB} \sigma \lor \tau$, $\supp{\at{Q}}{\baseD} \psi$, and $\baseD \baseGeq \baseB$ that $\supp{\at{T}(\at{Q})}{\baseD} \sigma \lor \tau$. By monotonicity, $\sigma \supp{\at{W}(\cdot)}{\baseC} \at{q}$ and $\tau \supp{\at{W}(\cdot)}{\baseC} \at{q}$ imply $\sigma \supp{\at{W}(\cdot)}{\baseD} \at{q}$ and $\tau \supp{\at{W}(\cdot)}{\baseD} \at{q}$, respectively. Together with $\supp{\at{T}(\at{Q})}{\baseD} \sigma \lor \tau$, it follows by \eqref{cl:BI-BeS-supp:disjunction} that $\supp{\at{W}(\at{T}(\at{Q}))}{\baseD} \at{q}$. 
        \item $\chi = \bot$. The reasoning is similar to that for $\top$. 
    \end{itemize}
    This completes the induction. 
\end{proof}

\section{Completeness} \label{sec:completeness}

\begin{theorem}[Completeness] \label{thm:BI-BeS-completeness}
    If $\varGamma \supp{}{} \phi$, then $\varGamma \proves \phi$.
\end{theorem}

We follow the strategy used by Sandqvist~\cite{Sandqvist2015base} for \IPL{} and by the authors for \IMLL{}~\cite{gheorghiu2023imll}, suitably modified for \BI{}. We require to show that  $\varGamma \supp{}{} \phi$ implies that there is an $\system{NBI}$-proof witnessing $\varGamma \proves{} \phi$. To this end, we associate to each sub-formula $\rho$ of $\varGamma \cup \{\phi\}$ a unique atom $\at{r}$ --- denoted $\rho^\flat$ below,  $(-)^\flat$ is a function from formulas to atoms --- and construct a base $\base{N}$ such that $\at{r}$ behaves in $\base{N}$ as $\rho$ behaves in $\system{NBI}$. Moreover, formulae and their atomizations are semantically equivalent in any extension of $\base{N}$ so that support in $\base{N}$ characterizes both validity 
 and provability --- that is $\phi^\flat \supp{(\cdot)}{\baseBI} \phi$ and $\phi \supp{(\cdot)}{\baseBI} \phi^\flat$. When $\rho \in \At$, we take $\at{r} := \rho$, but for complex $\rho$ we choose $r$ to be alien to $\varGamma \cup\{\phi\}$.


We proceed with a formal treatment.  Let $\varXi$ be the set of all sub-formulae of $\varGamma$ and $\phi$. Let $\flatBI{(\cdot)} \colon \varXi \to \At$ be an injection that acts as an identity on $\varXi\cap\At$ (i.e., $\flatBI{\at{p}} = \at{p}$ for $\at{p}\in \varXi \cap \At$). Let  $\deflatBI{(\cdot)}$ be the left-inverse of $\flatBI{(\cdot)}$ --- that is,
\[
\deflatBI{\at{p}}
:=
\begin{cases}
\chi & \mbox{if $\at{p}=\flatBI{\chi}$} \\
\at{p}  & \mbox{otherwise}
\end{cases}
\]
Both extend to bunches pointwise with identity on units; that is, 
\[
\begin{array}{c@{\qquad}c}
\varDelta^\flat :=
\begin{cases}
   \flatBI{\phi} & \mbox{if } \varDelta = \phi \in \Formulas \\
    \aunit & \mbox{if } \varDelta  = \aunit \\
   \munit & \mbox{if } \varDelta  = \munit \\
     \flatBI{\varDelta}_1 \fatsemi  
     \flatBI{\varDelta}_2 & \mbox{if } 
 \varDelta = \varDelta_1 \fatsemi \varDelta_2 \\
    \flatBI{\varDelta}_1 \fatcomma \flatBI{\varDelta}_2 & \mbox{if }  \varDelta = \varDelta_1 \fatcomma \varDelta_2
\end{cases}
&
\deflatBI{ \at{P} } :=
\begin{cases}
   \deflatBI{\at{p}} &  \mbox{if } 
 \at{P} = \at{p} \in \Atoms \\
 \aunit & \mbox{if } \at{P}  = \aunit \\
   \munit & \mbox{if } \at{P}  = \munit\\
    \deflatBI{ \at{P} }_1 \fatsemi \deflatBI{ \at{P} }_2 & \mbox{if }   \at{P}  =  \at{P} _1 \fatsemi  \at{P} _2 \\
    \deflatBI{ \at{P} }_1 \fatcomma \deflatBI{ \at{P} }_2 &  \mbox{if }  \at{P}  =  \at{P}_1 \fatcomma  \at{P}_2
\end{cases}
\end{array}
\]
We may write $\Gamma(\Delta)^\flat$ and  $P(Q)^\sharp$ to denote $(\Gamma(\Delta))^\flat$ and  $(P(Q))^\sharp$, respectively. 

\begin{figure}[t]
\hrule \vspace{2mm}
        \[
        \begin{array}{c}
        \infer[\ref{baseN:id}]{\flatBI{\sigma} \seq \flatBI{\sigma}}{ } 
        \quad 
        \infer[\ref{baseN:exch}]{\flatBI{\varDelta} \seq \flatBI{\sigma}}{\flatBI{\varDelta'} \seq \flatBI{\sigma}} 
        \quad 
        \infer[\ref{baseN:weak}]{\flatBI{\varDelta( \varLambda \addcomma \varLambda')} \seq \flatBI{\sigma}}{\flatBI{\varDelta(\varLambda)} \seq \flatBI{\sigma}} 
        \quad 
        \infer[\ref{baseN:cont}]{\flatBI{\varDelta(\varLambda)} \seq \flatBI{\sigma}}{\flatBI{\varDelta(\varLambda \addcomma \varLambda)} \seq \flatBI{\sigma}} 
        \\[1.5ex]
        \infer[\ref{baseN:mtop-intro}]{\munit \seq \flatBI{\mtop}}{} 
        \quad
        \infer[\ref{baseN:mtop-elim}]{\flatBI{\varDelta(\varLambda)} \seq \at{p}}{\flatBI{\varDelta(\munit)} \seq \at{p} & \flatBI{\varLambda} \seq \flatBI{\mtop}} 
        \quad 
        \infer[\ref{baseN:wand-intro}]{\flatBI{\varDelta} \seq \flatBI{(\sigma \wand \tau)}}{\flatBI{\varDelta} \mcomma \flatBI{\sigma} \seq \flatBI{\tau}}
        \\[1.5ex]
        \infer[\ref{baseN:wand-elim}]{\flatBI{\varDelta} \multcomma \flatBI{\varLambda} \seq \flatBI{\tau}}{\flatBI{\varDelta} \seq \flatBI{(\sigma \wand \tau)} & \flatBI{\varLambda} \seq \flatBI{\sigma}} 
        \quad 
        \infer[\ref{baseN:mand-intro}]{\flatBI{\varDelta} \multcomma \flatBI{\varLambda} \seq \flatBI{(\sigma \mand \tau)}}{\flatBI{\varDelta} \seq \flatBI{\sigma} & \flatBI{\varLambda} \seq \flatBI{\tau}} 
        \quad 
        \infer[\ref{baseN:mand-elim}]{\flatBI{\varDelta(\varLambda)} \seq \at{p}}{\flatBI{\varDelta(\sigma \multcomma \tau)} \seq \at{p} & \flatBI{\varLambda} \seq \flatBI{(\sigma \mand \tau)}}
        \\[1.5ex]
        \infer[\ref{baseN:atop-intro}]{\aunit \seq \flatBI{\top}}{} 
        \quad
        \infer[\ref{baseN:atop-elim}]{\flatBI{\varDelta(\varLambda)} \seq \at{p}}{\flatBI{\varDelta(\aunit)} \seq \at{p} & \flatBI{\varLambda} \seq \flatBI{\top}} 
        \quad 
        \infer[\ref{baseN:to-intro}]{\flatBI{\varDelta} \seq \flatBI{(\sigma \to \tau)}}{\flatBI{(\varDelta \acomma \sigma)} \seq \flatBI{\tau}} 
        \\[1.5ex]
        \infer[\ref{baseN:to-elim}]{\flatBI{\varDelta} \acomma \flatBI{\varLambda} \seq \flatBI{\tau}}{\flatBI{\varDelta} \seq \flatBI{(\sigma \to \tau)} & \flatBI{\varLambda} \seq \flatBI{\sigma}} 
        \quad 
        \infer[\ref{baseN:land-intro}]{\flatBI{\varDelta} \acomma \flatBI{\varLambda} \seq \flatBI{(\sigma \land \tau)}}{\flatBI{\varDelta} \seq \flatBI{\sigma} & \flatBI{\varLambda} \seq \flatBI{\tau}} 
        \quad 
        \infer[\ref{baseN:land-elim}]{\flatBI{\varDelta(\varLambda)} \seq \at{p}}{\flatBI{\varDelta(\sigma \acomma \tau)} \seq \at{p} & \flatBI{\varLambda} \seq \flatBI{(\sigma \land \tau)}}
        \\[1.5ex]
        \infer[\ref{baseN:lor-intro}]{\flatBI{\varDelta} \seq \flatBI{(\sigma_1 \lor \sigma_2)}}{\flatBI{\varDelta} \seq \flatBI{\sigma}_i}
        \quad 
        \infer[\ref{baseN:lor-elim}]{\flatBI{\varDelta(\varLambda)} \seq \at{p}}{\flatBI{\varDelta(\sigma)} \seq \at{p} & \flatBI{\varDelta(\tau)} \seq \at{p} & \flatBI{\varLambda} \seq \flatBI{(\sigma \lor \tau)}} 
        \quad 
        \infer[\ref{baseN:bot-elim}]{\flatBI{\varDelta} \seq \at{p}}{\flatBI{\varDelta} \seq \flatBI{\bot}} 
        \end{array}
        \vspace{2mm}
        \]
\hrule
	\caption{The Natural Base $\baseBI$}
	\label{fig:baseN} 
 \end{figure}

\labelandtag{baseN:id}{\flatBI{\rn{id}}}
\labelandtag{baseN:exch}{\flatBI{\exch}}
\labelandtag{baseN:weak}{\flatBI{\weak}}
\labelandtag{baseN:cont}{\flatBI{\cont}} \labelandtag{baseN:mtop-intro}{\flatBI{\irn{\mtop}}}
\labelandtag{baseN:mtop-elim}{\flatBI{\ern{\mtop}}}
\labelandtag{baseN:wand-intro}{\flatBI{\irn{\wand}}}
\labelandtag{baseN:wand-elim}{\flatBI{\ern{\wand}}}
\labelandtag{baseN:mand-intro}{\flatBI{\irn{\mand}}}
\labelandtag{baseN:mand-elim}{\flatBI{\ern{\mand}}}
\labelandtag{baseN:atop-intro}{\flatBI{\irn{\atop}}}
\labelandtag{baseN:atop-elim}{\flatBI{\ern{\atop}}}
\labelandtag{baseN:to-intro}{\flatBI{\irn{\to}}}
\labelandtag{baseN:to-elim}{\flatBI{\ern{\to}}}
\labelandtag{baseN:land-intro}{\flatBI{\irn{\land}}}
\labelandtag{baseN:land-elim}{\flatBI{\ern{\land}}}
\labelandtag{baseN:lor-intro}{\flatBI{\irn{\lor}}}
\labelandtag{baseN:lor-elim}{\flatBI{\ern{\lor}}}
\labelandtag{baseN:bot-elim}{\flatBI{\ern{\bot}}}

We construct a base $\baseBI$ such that $\flatBI{\sigma}$ behaves in $\baseBI$ as $\sigma$ behaves in $\calculusBI$. The base $\baseBI$ contains all instances of the rules of Figure~\ref{fig:baseN}, which merits comparison with Figure~\ref{fig:NBI}, for all $\at{p} \in \At$, $\sigma, \sigma_1,\sigma_2, \tau \in \varXi$, and $\varDelta, \varDelta',\varLambda,\varLambda' \in \Bunches{(\varXi)}$ such that $\varDelta \equiv \varDelta'$. Note, we need not include $\ref{baseN:id}$, $\ref{baseN:exch}$, $\ref{baseN:weak}$, and $\ref{baseN:cont}$ because of \myhyperlink{def:derivability-in-a-base:taut},  \myhyperlink{def:derivability-in-a-base:exch},  \myhyperlink{def:derivability-in-a-base:weak},  and \myhyperlink{def:derivability-in-a-base:cont}, respectively. Theorem~\ref{thm:BI-BeS-completeness} (Completeness) follows from the following three properties: 
\begin{itemize}[label=--]
    \item (AtComp). Let $\at{S}(\cdot) \in \BunchesWithHole(\At)$, $\at{p} \in \At$, and $\base{B}$ be a base: 
    $\at{P} \supp{\at{S(\cdot)}}{\base{B}} \at{p}  \qquad \text{iff} \qquad \at{S(P)} \proves[\base{B}] \at{p}. 
    $
    \item (Flat).  For any sub-formula $\xi \in \varXi$ and $\base{N}' \supseteq \base{N}$: $\supp{\at{S}}{\base{N}'} \xi^\flat \qquad \text{iff} \qquad  \supp{\at{S}}{\base{N}'} \xi
    $. 
    \item (Nat). Let $S \in \Bunches(\At)$ and $\at{p} \in \At$: 
    if $\at{S} \proves[\base{N}] \at{p} \text{, then } \at{S}^\natural \proves[\system{NBI}] \at{p}^\natural$.
\end{itemize}

\labelandtag{lem:atcomp}{\textbf{AtComp}}
\labelandtag{lem:flatequivalence}{\textbf{Flat}}
\labelandtag{lem:sharpening}{\textbf{Nat}}

The first claim (i.e.,~\ref{lem:atcomp}) follows from~\myhyperlink{def:derivability-in-a-base:cut} in the definition of derivability in a base (Definition~\ref{def:derivability-in-a-base}) and the ground case of completeness. It is proved by induction on support according to the structure of $\at{P}$. 

The second claim (i.e.,~\ref{lem:flatequivalence}) formalizes the idea that every formula $\xi$ appearing in $\varGamma \cup \{\phi\}$ behaves the same as $\flatBI{\xi}$ in any base extending $\baseBI$. Consequently, 
$\flatBI{\varGamma} \supp{}{\baseBI'} \flatBI{\phi}$ iff $\varGamma \supp{}{\baseBI'} \phi$ for any $\baseBI' \supseteq \baseBI$.

The third claim (i.e.,~\ref{lem:sharpening}) intuitively says that $\baseBI$ is a faithful atomic encoding of $\system{NBI}$, witnessed by $\deflatBI{(\cdot)}$. This together with \eqref{lem:atcomp} guarantees that every $\xi \in \varXi$ behaves in $\baseBI$ as $\flatBI{\xi}$ in $\baseBI$, thus as $\deflatBI{\left( \flatBI{\xi} \right)} = \xi$ in $\system{NBI}$. It is proved by induction on derivability in a base.

\begin{proof}[Proof of Theorem~\ref{thm:BI-BeS-completeness} --- Completeness] 
Assume $\varGamma \supp{}{} \phi$ and let $\baseN$ be the bespoke base for $\varGamma \seq \phi$ defined as in Figure~\ref{fig:baseN}. By \eqref{lem:flatequivalence}, $\varGamma^\flat \supp{(\cdot)}{\baseN} \phi^\flat$. Therefore, by \eqref{lem:atcomp}, $\flatBI{\varGamma} \proves[\baseN] \flatBI{\phi}$. Finally, by \eqref{lem:sharpening}, $\deflatBI{\left( \flatBI{\varGamma} \right)} \proves[\system{NBI}] \deflatBI{\left( \flatBI{\phi} \right)}$ --- that is, $\varGamma \proves[\system{NBI}] \phi$. 
\end{proof} 

It remains only to prove the ancillary lemmas.

\begin{lemma}[\ref{lem:atcomp}]
    Let $\at{S}(\cdot) \in \BunchesWithHole(\At)$ and $\at{p} \in \At$ and let $\base{B}$ be a base: 
    \[ \at{P} \supp{\at{S}(\cdot)}{\base{B}} \at{p}  \qquad \qquad \text{iff} \qquad \qquad \at{S}(\at{P}) \proves[\base{B}] \at{p} \] 
\end{lemma}
\begin{proof}
For the `only if' direction, assume $\at{P} \supp{\at{S}(\cdot)}{\baseB} \at{p}$. By \eqref{cl:BI-BeS-supp:inf}, $\forall \baseX \baseGeq \baseB$ and $\forall \at{U} \in \Bunches(\At)$, if $\supp{\at{U}}{\baseX} \at{P}$, then $\supp{\at{S}(\at{U})}{\baseX} \at{p}$. Set $\baseX = \baseB$ and $\at{U} = \at{P}$.  Since, $\supp{\at{S}(\at{P})}{\baseB} \at{p}$ implies $\at{S}(\at{P}) \proves[\baseB] \at{p}$, it remains only to show $\supp{\at{P}}{\baseB} \at{P}$. This follows from Proposition~\ref{prop:support-identity}

For the `if' direction, we proceed by induction on the structure of $\at{P}$. Throughout we assume $\at{S}(\at{P}) \proves[\baseB] \at{p}$, which we call the \emph{first assumption}. By \eqref{cl:BI-BeS-supp:inf}, the desired conclusion is equivalent to the following: 
\[
\mbox{$\forall \baseC \baseGeq \baseB$ and $\forall \at{Q} \in \Bunches(\At)$, if $\supp{\at{Q}}{\baseC} \at{P}$, then $\supp{\at{S}(\at{Q})}{\baseC} \at{p}$.}
\]
 Therefore, let $\baseC \baseGeq \baseB$ and $ \at{Q} \in \Bunches(\At)$ be arbitrary such that $\supp{\at{Q}}{\baseC} \at{P}$, which we call the \emph{second} assumption. We require to establish $\supp{\at{S}(\at{Q})}{\baseC} \at{p}$ --- that is, $\at{S}(\at{Q}) \proves[\baseC] \at{p}$. 
 \begin{itemize}[label=--]
     \item $\at{P}$ is in $\Atoms$. Trivial by the fact that derivability in a base is monotonic with respect to base-extensions (Proposition~\ref{prop:monotonicity}) and~\myhyperlink{def:derivability-in-a-base:cut} in Definition~\ref{def:derivability-in-a-base}. 
     \item $\at{P}$ is $\aunit$. By the monotonicity of derivability in a base on the first assumption,  $\at{S}(\aunit) \proves[\baseC] \at{p}$. By~\myhyperlink{def:derivability-in-a-base:weak},  $\at{S}(\aunit \acomma \at{Q}) \proves[\baseC] \at{p}$. Hence (by~\myhyperlink{def:derivability-in-a-base:exch}), $\at{S}(\at{Q}) \proves[\baseC] \at{p}$.
     \item $\at{P}$ is $\munit$.  By the monotonicity of derivability in a base on the first assumption,  $\at{S}(\munit) \proves[\baseC] \at{p}$. By \eqref{cl:BI-BeS-supp:munit} on the second assumption, $\at{Q} \bunchStrongerThan \munit$. In other words, $\at{Q} \equiv \munit \acomma \at{R}$ for some $\at{R} \in \Bunches(\Atoms)$. Therefore, $\at{S}(\at{Q}) \proves[\baseC] \at{p}$ obtains by~\myhyperlink{def:derivability-in-a-base:weak}. 
     \item $\at{P}$ is $\at{P}_1 \acomma \at{P}_2$.  By the monotonicity of derivability in a base on the first  assumption,  $\at{S}(\at{P}_1 \acomma \at{P}_2) \proves[\baseC] \at{p}$. By \eqref{cl:BI-BeS-supp:acomma} on the second assumption, $\supp{\at{Q}}{\baseC} \at{P}_1$ and $\supp{\at{Q}}{\baseC} \at{P}_2$, since $\at{Q} \bunchStrongerThan \at{Q}$. Hence, by the induction hypothesis (IH), $\at{S}(\at{Q} \acomma \at{Q}) \proves[\baseC] \at{p}$. Whence, by~\myhyperlink{def:derivability-in-a-base:cont}, $\at{S}(\at{Q}) \proves[\baseC] \at{p}$.
     \item $\at{P}$ is $\at{P}_1 \mcomma \at{P}_2$. By the monotonicity of derivability in a base on the first assumption,  $\at{S}(\at{P}_1 \acomma \at{P}_2) \proves[\baseC] \at{p}$. By \eqref{cl:BI-BeS-supp:mcomma} on the second assumption, there are $\at{Q}_1$ and $\at{Q}_2$ such that $\at{Q} \bunchStrongerThan (\at{Q}_1 \mcomma \at{Q}_2)$,  $\supp{\at{Q}_1}{\baseC} \at{P}_1$, $\supp{\at{Q}_2}{\baseC} \at{P}_2$. Hence, by the IH twice applied, $\at{S}(\at{Q}_1 \mcomma \at{Q}_2) \proves[\baseC] \at{p}$. Whence, by~\myhyperlink{def:derivability-in-a-base:weak}, $\at{S}(\at{Q}) \proves[\baseC] \at{p}$.
 \end{itemize}
This completes the induction.
\end{proof}
Next, we turn to \eqref{lem:flatequivalence}. For this, the crucial observation is that for arbitrary sub-formula $\xi$ of $\varGamma$ or $\varphi$, $\flatBI{\xi}$ behaves the same in $\proves[\baseN]$ as $\xi$ in $\supp{}{\baseN}$. 
\begin{lemma}
\label{prop:flat-connective}
    The following holds for arbitrary $\baseM \baseGeq \baseN$, and sub-formula $\xi$ of $\varGamma$ or $\varphi$:
    \begin{enumerate}
        \item\label{item:flat-connective-mand} $\at{S} \proves[\baseM] \flatBI{(\sigma \mand \tau)}$ iff for arbitrary $\baseX \baseGeq \baseM$, $\at{U}(\cdot) \in \BunchesWithHole(\At)$, $\at{p} \in \At$, if $\at{U}(\flatBI{\sigma} \mcomma \flatBI{\tau}) \proves[\baseX] \at{p}$, then $\at{U(S)} \proves[\baseX] \at{p}$. 
        \item\label{item:flat-connective-mto} $\at{S} \proves[\baseM] \flatBI{(\sigma \mto \tau)}$ iff $\at{S} \mcomma \flatBI{\sigma} \proves[\baseM] \flatBI{\tau}$. 
        \item\label{item:flat-connective-mtop} $\at{S} \proves[\baseM] \flatBI{\mtop}$ iff for arbitrary $\baseX \baseGeq \baseM$, $\at{U}(\cdot) \in \BunchesWithHole(\At)$, $\at{p} \in \At$, if $\at{U}(\munit) \proves[\baseX] \at{p}$, then $\at{U(S)} \proves[\baseX] \at{p}$. 
        \item\label{item:flat-connective-land} $\at{S} \proves[\baseM] \flatBI{(\sigma \land \tau)}$ iff for arbitrary $\baseX \baseGeq \baseM$, $\at{U}(\cdot) \in \BunchesWithHole(\At)$, $\at{p} \in \At$, if $\at{U}(\flatBI{\sigma} \acomma \flatBI{\tau}) \proves[\baseX] \at{p}$, then $\at{U(S)} \proves[\baseX] \at{p}$. 
        \item\label{item:flat-connective-to} $\at{S} \proves[\baseM] \flatBI{(\sigma \to \tau)}$ iff $\at{S} \acomma \flatBI{\sigma} \proves[\baseM] \flatBI{\tau}$. 
        \item\label{item:flat-connective-top} $\at{S} \proves[\baseM] \flatBI{\top}$ iff for arbitrary $\baseX \baseGeq \baseM$, $\at{U}(\cdot) \in \BunchesWithHole(\At)$, $\at{p} \in \At$, if $\at{U}(\aunit) \proves[\baseX] \at{p}$, then $\at{U(S)} \proves[\baseX] \at{p}$. 
        \item\label{item:flat-connective-lor} $\at{S} \proves[\baseM] \flatBI{(\sigma \lor \tau)}$ iff for arbitrary $\baseX \baseGeq \baseM$, $\at{U}(\cdot) \in \BunchesWithHole(\At)$, $\at{p} \in \At$, if $\at{U}(\flatBI{\sigma}) \proves[\baseX] \at{p}$ and $\at{U}(\flatBI{\tau}) \proves[\baseX] \at{p}$, then $\at{U(S)} \proves[\baseX] \at{p}$. 
        \item\label{item:flat-connective-bot} $\at{S} \proves[\baseM] \flatBI{\bot}$ iff for arbitrary $U(\cdot) \in \BunchesWithHole(\At)$, $\at{p} \in \At$, $\at{U(S)} \proves[\baseM] \at{p}$. 
    \end{enumerate}
\end{lemma}
\begin{proof}
    We fix some $\baseM \baseGeq \baseN$ and $\at{S}$. In each case, the two directions of the ``if and only if'' statement follows from the corresponding introduction rule and elimination rule in $\baseN$, respectively. We only demonstrate some representative cases, the others being similar:
    \begin{enumerate}
        \item[4.] We consider the two directions separately. 
        \begin{itemize}[label=--]
            \item Left to Right. Given $\at{S} \proves[\baseM] \flatBI{(\sigma \land \tau)}$, we fix some $\baseO \baseGeq \baseM$, $\at{T(\cdot)} \in \BunchesWithHole(\At)$, and $\at{q} \in \At$ such that $\at{T}(\flatBI{\sigma} \acomma \flatBI{\tau}) \proves[\baseO] \at{q}$. The goal is to show $\at{T(S)} \proves[\baseO] \at{q}$. Note that the rule $\Bigl( \at{S} \seq \flatBI{(\sigma \land \tau)}, \at{T}(\flatBI{\sigma} \acomma \flatBI{\tau}) \seq \at{q} \Bigr) \Rightarrow \at{T(S)} \seq \at{q}$ is in $\baseN$, thus in $\baseO$, so applying~\myhyperlink{def:derivability-in-a-base:rule}, we have $\at{T(S)} \proves[\baseO] \at{q}$. 
            \item Right to Left. Assume the RHS, it follows in particular that if $\flatBI{\sigma} \acomma \flatBI{\tau} \proves[\baseM] \flatBI{(\sigma \land \tau)}$, then $\at{S} \proves[\baseM] \flatBI{(\sigma \land \tau)}$. Here $\flatBI{\sigma} \acomma \flatBI{\tau} \proves[\baseM] \flatBI{(\sigma \land \tau)}$ follows from $\flatBI{\sigma} \proves[\baseM] \flatBI{\sigma}$, $\flatBI{\tau} \proves[\baseM] \flatBI{\tau}$, and that $\flatBI{\sigma} \seq \flatBI{\sigma}, \flatBI{\tau} \seq \flatBI{\tau} \Rightarrow \flatBI{\sigma} \acomma \flatBI{\tau} \seq \flatBI{(\sigma \land \tau)}$ is a rule in $\baseN$ (thus in $\baseM$). 
        \end{itemize}
        \item[6.] We consider the two directions separately. 
        \begin{itemize}[label=--]
            \item Left to Right. Given $\at{S} \proves[\baseM] \flatBI{\top}$, we fix some $\baseO \baseGeq \baseM$, $\at{T}(\cdot) \in \BunchesWithHole(\At)$, $\at{q} \in \At$ such that $\at{T}(\aunit) \proves[\baseO] \at{q}$. The goal is to show $\at{T}(\at{S}) \proves[\baseO] \at{q}$. Note that the rule $\bigl( \at{T}(\aunit) \seq \at{q}, \at{S} \seq \flatBI{\top} \bigr) \Rightarrow \at{T}(\at{S}) \seq \at{q}$ is in $\baseO$, so apply~\myhyperlink{def:derivability-in-a-base:rule} to $\at{T}(\aunit) \proves[\baseO] \at{q}$ and $\at{S} \proves[\baseO] \flatBI{\top}$ we have $\at{T}(\at{S}) \proves[\baseM] \at{q}$. 
            \item Right to Left. Given the RHS, it follows in particular that, if $\aunit \proves[\baseM] \flatBI{\top}$, then $\at{S} \proves[\baseM] \flatBI{\top}$. Since the rule $\empset \Rightarrow \aunit \seq \flatBI{\top}$ is in $\baseO$, it follows that $\aunit \proves[\baseM] \flatBI{\top}$, so $\at{S} \proves[\baseM] \flatBI{\top}$. 
        \end{itemize}
        \item[8.] We consider the two directions separately. 
        \begin{itemize}[label=--]
            \item Right to Left.  Given the RHS, one may choose $\at{U}:= (\cdot)$ and $\at{p} =\flatBI{\bot}$ to obtain the desired result. 
        \end{itemize}
    \end{enumerate}
    This completes the case analysis.
\end{proof}
\begin{lemma}[\ref{lem:flatequivalence}] 
    For any sub-formula $\xi$ of $\varGamma$ or $\phi$, and arbitrary $\base{N}' \supseteq \base{N}$, we have:  
    \begin{equation}\label{eq:flat-equivalence}
        \supp{\at{S}}{\base{N}'} \xi^\flat \qquad \qquad \text{iff} \qquad  \qquad \supp{\at{S}}{\base{N}'} \xi \tag{$\dagger$}
    \end{equation}
\end{lemma}
\begin{proof}
   Let $\baseM \baseGeq \baseN$ be arbitrary. We proceed by induction on the structure of $\xi$.
    \begin{itemize}[label = --]
        \item $\xi$ is atomic. By definition, $\flatBI{\xi}$ is exactly $\xi$, thus \eqref{eq:flat-equivalence} is a tautology. 
        \item $\xi$ is $\sigma \mand \tau$. This obtains as follows:
        \begin{align}
           \supp{\at{S}}{\baseM} \flatBI{(\sigma \mand \tau)} & \qquad \text{iff} \qquad  \at{S}  \proves[\baseM] \flatBI{(\sigma \mand \tau)}  \tag*{\eqref{cl:BI-BeS-supp:At}} \\ 
            & \qquad \text{iff} \qquad \forall \baseX \baseGeq \baseM, \forall \at{U}(\cdot) \in \BunchesWithHole(\At), \forall \at{p} \in \At, \tag{Proposition~\ref{prop:flat-connective}.\ref{item:flat-connective-mand}} \\
            & \qquad \hspace{6ex} \mbox{if $\at{U}(\flatBI{\sigma} \mcomma \flatBI{\tau}) \proves[\baseX] \at{p}$, then  $\at{U(S)}\proves[\baseX] \at{p}$} \tag*{} \\
            & \qquad \text{iff} \qquad   \forall \baseX \baseGeq \baseM, \forall \at{U}(\cdot) \in \BunchesWithHole(\At),
             \tag{IH + \ref{lem:atcomp}} \\
            & \qquad \hspace{6ex} \mbox{if $\sigma \mcomma \tau \supp{\at{U}(\cdot)}{\baseX} \at{p}$, then  $\supp{\at{U}(\at{S})}{\baseX} \at{p} $} \tag*{}\\
            & \qquad \text{iff} \qquad   \supp{\at{S}}{\baseM} \sigma \mand \tau & \tag*{\eqref{cl:BI-BeS-supp:mand}} 
        \end{align}
        \item $\xi$ is $\mtop$. This obtains as follows:
                \begin{align}
             \supp{\at{S}}{\baseM} \flatBI{\mtop} & \qquad \text{iff} \qquad \at{S} \proves[\baseM] \flatBI{\mtop}  \tag*{\eqref{cl:BI-BeS-supp:At}} \\
            & \qquad \text{iff} \qquad \forall \baseX \baseGeq \baseM, \forall \at{U}(\cdot) \in \BunchesWithHole(\At), \forall \at{p} \in \At, \tag{Proposition~\ref{prop:flat-connective}.\ref{item:flat-connective-top}}  \\
            & \qquad \hspace{6ex}
            \text{ if } \at{U}(\munit) \proves[\baseX] \at{p}, \text{ then } \at{U(S)} \proves[\baseX] \at{p}  \notag \\
            & \qquad \text{iff} \qquad \forall \baseX \baseGeq \baseM, \forall \at{U}(\cdot) \in \BunchesWithHole(\At), \forall \at{p} \in \At, \tag{IH} \\
            & \qquad \hspace{6ex} \text{ if } \munit \supp{\at{U}(\cdot)}{\baseX} \at{p}, \text{ then } \at{S} \supp{\at{U}}{\baseX} \at{p} \notag  \\
            & \qquad \text{iff} \qquad \supp{\at{S}}{\baseM}\mtop  \tag*{\eqref{cl:BI-BeS-supp:mtop}}
        \end{align}
        \item $\xi$ is $\sigma \mto \tau$. This obtains as follows:
        \begin{align}
             \supp{\at{S}}{\baseM} \flatBI{(\sigma \mto \tau)} & \qquad \text{iff} \qquad \at{S} \proves[\baseM] \flatBI{(\sigma \mto \tau)}  \tag*{\eqref{cl:BI-BeS-supp:At}} \\
            & \qquad \text{iff} \qquad \at{S} \mcomma \flatBI{\sigma} \proves[\baseM] \flatBI{\tau}  \tag{Proposition~\ref{prop:flat-connective}.\ref{item:flat-connective-mto}}  \\
            & \qquad \text{iff} \qquad \flatBI{\sigma} \supp{\at{S} \mmcomma (\cdot)}{\baseM} \flatBI{\tau} \tag{\ref{lem:atcomp}} \\
            & \qquad \text{iff} \qquad \sigma \supp{\at{S} \mmcomma (\cdot)}{\baseM} \tau \tag{IH+\eqref{cl:BI-BeS-supp:inf}} \\
            & \qquad \text{iff} \qquad \supp{\at{S}}{\baseM} \sigma \mto \tau \tag*{\eqref{cl:BI-BeS-supp:wand}}
        \end{align}
        \item $\xi$ is $\sigma \land \tau$.  \emph{Mutatis mutandis} on the $\ast$-case
        \item $\xi$ is $\top$. \emph{Mutatis mutandis} on the $\mtop$-case
        \item $\xi$ is $\sigma \to \tau$.  \emph{Mutatis mutandis} on the $\mto$-case
        \item $\xi$ is $\sigma \lor \tau$.  \emph{Mutatis mutandis} on the $\ast$-case
        \item $\xi$ is $\bot$. This obtains as follows:
    \begin{align}
        \supp{\at{S}}{\baseM} \flatBI{\bot} & \qquad \text{iff} \qquad \at{S} \proves[\baseM] \flatBI{\bot}  \tag*{\eqref{cl:BI-BeS-supp:At}} \\
        & \qquad \text{iff} \qquad \mbox{$\forall \at{U}(\cdot) \in \BunchesWithHole(\At), \forall \at{p} \in \Atoms, \, \at{U(S)} \proves[\baseM] \at{p}$} \tag{Proposition~\ref{prop:flat-connective}.\ref{item:flat-connective-bot}} \\
        & \qquad \text{iff} \qquad \mbox{$\forall \at{U}(\cdot) \in \BunchesWithHole(\At), \forall \at{p} \in \Atoms, \, \at{S}\supp{\at{U(\cdot)}}{\baseM} \at{p}$} \tag{\ref{lem:atcomp}} \\
        & \qquad \text{iff} \qquad \mbox{$\forall \at{U}(\cdot) \in \BunchesWithHole(\At), \forall \at{p} \in \Atoms, \, \supp{\at{U(S)}}{\baseM} \at{p}$} \tag{Proposition~\ref{prop:support-identity}+\eqref{cl:BI-BeS-supp:inf}}\\
        & \qquad \text{iff} \qquad \mbox{$\supp{\at{S}}{\baseM} \bot$} \tag{\ref{cl:BI-BeS-supp:bot}}
    \end{align}
    \end{itemize}
    This completes the induction.
\end{proof}

Finally, Lemma~\ref{lem:sharpening} says that $(\cdot)^\sharp$ lifts derivability in $\baseN$ faithfully to derivability in BI. 
\begin{lemma}[\ref{lem:sharpening}]
    Let $S \in \Bunches(\At)$ and $\at{p} \in \At$: 
    if $\at{S} \proves[\base{N}] \at{p} \text{, then } \at{S}^\natural \proves{} \at{p}^\natural$.
\end{lemma}
\begin{proof}
    We proceed by induction on $\at{S} \proves[\baseN] \at{p}$ according to Definition~\ref{def:derivability-in-a-base}. 
    \begin{itemize}[label=--]
        \item~\myhyperlink{def:derivability-in-a-base:taut}. In this case, it must be that $\at{S} = \at{p}$. The result follows by the $\rn{taut}$-rule in $\system{NBI}$.
    \item~\myhyperlink{def:derivability-in-a-base:rule}. Let $\Xi$ be all the sub-formulae of $\varGamma$ and $\phi$ relative to which $\biflat{(-)}$ and $\baseN$ are constructed. We proceed by case analysis on the rules in $\base{N}$. We only illustrate $\biflat{\ern \mand}$, the others being similar. 
    
    It must be that $\at{S} = \biflat{\at{\varGamma'}(\varDelta)}$, for some $\varGamma' \in \BunchesWithHole(\varXi)$ and $\varDelta \in \Bunches(\varXi)$, such that $\biflat{\varGamma'(\psi_1 \mcomma \psi_2)} \proves[\baseN] \biflat{\chi}$ and $\biflat{\varDelta} \proves[\baseN] \biflat{(\psi_1 \mand \psi_2)}$, for some $\psi_1, \psi_2, \chi \in \varXi$. By the IH, $\varGamma'(\psi_1 \mcomma \psi_2) \proves \chi$ and $\varDelta \proves (\psi_1 \mand \psi_2)$, using that $\binat{(-)}$ is a left-inverse of $\biflat{(-)}$. The result follows by $\ern \mand$ in $\system{NBI}$.

    \item~\myhyperlink{def:derivability-in-a-base:weak}. It must be that $\at{S} = \at{P}(\at{Q} \fatsemi \at{Q}')$, for some $\at{P} \in \BunchesWithHole(\At)$ and $\at{Q}, \at{Q}' \in \Bunches(\Atoms)$, and $\at{P}(\at{Q}) \proves[\baseN] \at{p}$. By the IH, $\binat{\at{P}(\at{Q})} \proves \binat{\at{p}}$. The result follows from application of the $\weak$-rule in $\system{NBI}$.
        \item~\myhyperlink{def:derivability-in-a-base:cont}. \emph{Mutatis mutandis} on the case for~\myhyperlink{def:derivability-in-a-base:weak}.
        \item~\myhyperlink{def:derivability-in-a-base:exch}. \emph{Mutatis mutandis} on the case for~\myhyperlink{def:derivability-in-a-base:cont}.
        \item~\myhyperlink{def:derivability-in-a-base:cut}.  It must be that $\at{S} = \at{P}(\at{Q})$, for some $\at{P} \in \BunchesWithHole(\At)$ and $\at{Q} \in \Bunches(\Atoms)$, such that $\at{Q} \proves[\baseN] \at{q}$ and $\at{S}(\at{q}) \proves[\baseN] \at{p}$, for some $\at{q} \in \Atoms$. By the IH, both $\binat{\at{Q}} \proves \binat{\at{q}}$ and $\binat{\at{S}(\at{q})} \proves \binat{\at{p}}$. The result follows from $\rn{cut}$-admissibility for BI (Lemma~\ref{lem:cut-admissibility}).
    \end{itemize}
    This completes the induction.
\end{proof}

\section{Conclusion} \label{sec:conclusion}

This paper delivers a B-eS for BI. While BI can be understood as the free combination of IPL and IMLL, both of which have been given B-eS before (see Sandqvist~\cite{Sandqvist2015hypothesis} and Gheorghiu et al.~\cite{gheorghiu2023imll}), several technical challenges remain in giving BI's B-eS. Foremost is that BI uses a more complex data-structure for contexts than either of the other logics: bunches. Relative to the earlier work, giving B-eS for IPL and for IMLL, handling bunches with the appropriate structural behaviour is the major challenge in giving a B-eS for BI. These details are contained above. Nevertheless, there is faithful embedding from IMLL and IPL into BI that is reflected in the relationship between their B-eS and the one for BI in this paper. Importantly, BI shares many connectives and the concept of bunches with relevance logics. The treatment herein suggests a method for a systematic/uniform presentation of the B-eS for relevance logics (e.g., following the approach by Read~\cite{Read1988}). 

A reason for studying the B-eS of substructural logics in the systems-oriented sciences is that it clearly and naturally supports, within the context of BI's usual resource reading \cite{pym2019resource} (cf. Figure~\ref{fig:suppBI}, clauses for $\mmcomma$ and $\aacomma$), the celebrated number-of-uses reading of the implications. This reading is rarely, if at all, reflected in the truth-functional semantics of these logics. In this paper, we illustrate the use of the B-eS for BI by a toy example in the setting of \emph{access control}. This is a promising area for the use of B-eS as it is dynamic in the sense that `access' is about actions as opposed to states and this dynamics corresponds to the passing of resources through implications on the number-of-uses reading. See \cite{MFPS-IRS-2024} for a more substantive discussion of these ideas by the present authors. 

In \cite{GGP2024stateeffect}, the authors have proposed a \emph{state-effect} reading of model-theoretic and proof-theoretic semantics that explains their relationship. In this reading, when logic is used in modelling systems, M-tS captures the properties of a system, while P-tS/B-eS captures its dynamics. Since BI has a well-established and intuitive reading as a logic of `resources', this work may be used in modelling dynamical systems in which there is some appropriate notion of resource, include automata, Petri nets, simulation modelling, etc. On the last, Kuorikoski and Reijula \cite{jaakko} have recently shown P-tS (in general) to be a suitable paradigm of meaning for executable models. 

Developing these applications in areas such as resource semantics and modelling remains future work.

\section*{Acknowledgements}
 
We are grateful to Yll Buzoku, Diana Costa, Sonia Marin, and Elaine Pimentel for many useful discussions on the P-tS for substructural logics, to Eike Ritter and Edmund Robinson for discussions on the category-theoretic formulation of B-eS, and to Tim Button for discussion of P-tS's philosophical context.

\section*{Declarations}

This work has been partially supported by the UK EPSRC grants EP/S013008/1 AND EP/R006865/1, and by Gheorghiu's EPSRC Doctoral Studentship.

\bibliography{refs}

\begin{thebibliography}{53}
\providecommand{\natexlab}[1]{#1}
\providecommand{\url}[1]{{#1}}
\providecommand{\urlprefix}{URL }
\providecommand{\doi}[1]{\url{https://doi.org/#1}}
\providecommand{\eprint}[2][]{\url{#2}}
 \bibcommenthead

\bibitem[{Abadi(2003)}]{abadi2003logic}
Abadi M (2003) {Logic in Access Control}. In: Kolaitis PG (ed) Logic in
  Computer Science --- LICS. IEEE, p 228--233

\bibitem[{Allwein and Dunn(1993)}]{Allwein}
Allwein G, Dunn JM (1993) {Kripke Models for Linear Logic}. The Journal of
  Symbolic Logic 58(2):514--545

\bibitem[{Anderson and Pym(2016)}]{ANDERSON201663}
Anderson G, Pym DJ (2016) {A Calculus and Logic of Bunched Resources and
  Processes}. Theoretical Computer Science 614:63--96.
  \doi{https://doi.org/10.1016/j.tcs.2015.11.035}

\bibitem[{Barwise and Seligman(1997)}]{BarwiseSeligman1997}
Barwise J, Seligman J (1997) Information Flow: The Logic of Distributed
  Systems. Cambridge University Press

\bibitem[{Brandom(2000)}]{Brandom2000}
Brandom R (2000) {Articulating Reasons: An Introduction to Inferentialism}.
  Harvard University Press

\bibitem[{Brotherston(2012)}]{Brotherston2012}
Brotherston J (2012) {Bunched Logics Displayed}. Studio Logica
  100(6):1223--1254. \doi{10.1007/s11225-012-9449-0}

\bibitem[{Br\"unnler(2009)}]{Brunnler2009}
Br\"unnler K (2009) {Deep Sequent Systems for Modal Logic}. Archive for
  Mathematical Logic 48(6):551--577. \doi{10.1007/s00153-009-0137-3}

\bibitem[{Coumans et~al(2014)Coumans, Gehrke, and {van
  Rooijen}}]{COUMANS201450}
Coumans D, Gehrke M, {van Rooijen} L (2014) Relational semantics for full
  linear logic. Journal of Applied Logic 12(1):50--66.
  \doi{10.1016/j.jal.2013.07.005}

\bibitem[{van Dalen(2013)}]{vanDalen}
van Dalen D (2013) Logic and Structure, 5th edn. Universitext, Springer,
  \doi{10.1007/978-1-4471-4558-5}

\bibitem[{Docherty(2019)}]{Docherty2019}
Docherty S (2019) {Bunched Logics: A Uniform Approach}. PhD thesis, University
  College London

\bibitem[{Dummett(1991)}]{Dummett1991logical}
Dummett M (1991) {The Logical Basis of Metaphysics}. Harvard University Press

\bibitem[{Galmiche et~al(2005)Galmiche, M\'{e}ry, and Pym}]{Galmiche2005}
Galmiche D, M\'{e}ry D, Pym D (2005) {The Semantics of BI and Resource
  Tableaux}. Mathematical Structures in Computer Science 15(6):1033–1088.
  \doi{10.1017/S0960129505004858},
  \urlprefix\url{https://doi.org/10.1017/S0960129505004858}

\bibitem[{Gheorghiu and Pym(2023{\natexlab{a}})}]{GP2023-SABI}
Gheorghiu A, Pym D (2023{\natexlab{a}}) {Semantical Analysis of the Logic of
  Bunched Implications}. Studia Logica 111:525--571.
  \doi{https://doi.org/10.1007/s11225-022-10028-z}

\bibitem[{Gheorghiu et~al(2024{\natexlab{a}})Gheorghiu, Gu, and
  Pym}]{GGP2024stateeffect}
Gheorghiu A, Gu T, Pym D (2024{\natexlab{a}}) {A Note on the Practice of
  Logical Inferentialism: The State-Effect Interpretation, Definitional
  Reflection, and Completeness}. 2nd Logic and Philosophy: Historical and
  Contemporary Issues Conference, Vilnius, Lithuania, May 2024. Manuscript at
  \url{https://arxiv.org/pdf/2403.10546}, accessed August 2024

\bibitem[{Gheorghiu et~al(2024{\natexlab{b}})Gheorghiu, Gu, and
  Pym}]{MFPS-IRS-2024}
Gheorghiu A, Gu T, Pym D (2024{\natexlab{b}}) {Inferentialist Resource
  Semantics}. In: Proc. 40th Mathematical Foundations of Programming Semantics
  (MFPS). Electronic Notes in Theoretical Informatics and Computer Science
  (ENTICS), {manuscript at \url{https://arxiv.org/abs/2402.09217}}, accessed
  August 2024

\bibitem[{Gheorghiu and Marin(2021)}]{Alex:Focusing}
Gheorghiu AV, Marin S (2021) {Focused Proof-search in the Logic of Bunched
  Implications}. In: Kiefer S, Tasson C (eds) {Foundations of Software Science
  and Computation Structures - FOSSACS 24}, LNCS, vol 12650. Springer, p
  247--267, \doi{10.1007/978-3-030-71995-1_13}

\bibitem[{Gheorghiu and Pym(2023{\natexlab{b}})}]{Alex:BI_Semantics}
Gheorghiu AV, Pym DJ (2023{\natexlab{b}}) {Semantical Analysis of the Logic of
  Bunched Implications}. Studia Logica \doi{10.1007/s11225-022-10028-z}

\bibitem[{Gheorghiu et~al(2021)Gheorghiu, Docherty, and
  Pym}]{Alex:Samsonschrift}
Gheorghiu AV, Docherty S, Pym DJ (2021) {Reductive Logic, Coalgebra, and
  Proof-search: A Perspective from Resource Semantics}. In: Palmigiano A,
  Sadrzadeh M (eds) {Samson Abramsky on Logic and Structure in Computer Science
  and Beyond}. Springer Outstanding Contributions to Logic Series, Springer,
  \doi{10.1007/978-3-031-24117-8_23}

\bibitem[{Gheorghiu et~al(2023)Gheorghiu, Gu, and Pym}]{gheorghiu2023imll}
Gheorghiu AV, Gu T, Pym DJ (2023) {Proof-theoretic Semantics for Intuitionistic
  Multiplicative Linear Logic}. \emph{Proc 32nd Tableaux}, LNAI 14278:367--385

\bibitem[{Girard(1995)}]{girard1995}
Girard JY (1995) {Linear Logic: its Syntax and Semantics}. In: Girard JY,
  Lafont Y, Regnier L (eds) Advances in Linear Logic. London Mathematical
  Society Lecture Note Series, Cambridge University Press, p 1–42,
  \doi{10.1017/CBO9780511629150.002}

\bibitem[{Goldfarb(2016)}]{goldfarb2016dummett}
Goldfarb W (2016) {On Dummett’s ``Proof-theoretic Justifications of Logical
  Laws''}. In: Piecha T, Schroeder-Heister P (eds) {Advances in Proof-theoretic
  Semantics}. Springer, p 195--210

\bibitem[{Ishtiaq and O'Hearn(2001)}]{Ishtiaq2001}
Ishtiaq SS, O'Hearn PW (2001) {BI as an Assertion Language for Mutable Data
  Structures}. In: Hankin C, Schmidt D (eds) Symposium on Principles of
  Programming Languages --- POPL. Association for Computing Machinery (ACM), p
  14--26, \doi{10.1145/360204.375719}

\bibitem[{Jaakko~Kuorikoski(2022)}]{jaakko}
Jaakko~Kuorikoski SR (2022) {Making It Count: An Inferentialist Account of
  Computer Simulation}. \url{https://osf.io/preprints/socarxiv/v9bmr},
  \doi{10.31235/osf.io/v9bmr}, accessed January 2023

\bibitem[{Kripke(1965)}]{kripke1965semantical}
Kripke SA (1965) {Semantical Analysis of Intuitionistic Logic I}. {Studies in
  Logic and the Foundations of Mathematics} 40:92--130

\bibitem[{Lambek and Scott(1988)}]{lambek1988introduction}
Lambek J, Scott PJ (1988) {Introduction to Higher-order Categorical Logic},
  vol~7. Cambridge University Press,
  \urlprefix\url{https://dl.acm.org/doi/10.5555/53626}

\bibitem[{Makinson(2014)}]{makinson2014inferential}
Makinson D (2014) {On an Inferential Semantics for Classical Logic}. Logic
  Journal of IGPL 22(1):147--154. \doi{10.1093/jigpal/jzt038}

\bibitem[{O'Hearn and Pym(1999)}]{o1999logic}
O'Hearn PW, Pym DJ (1999) {The Logic of Bunched Implications}. Bulletin of
  Symbolic Logic 5(2):215--244. \doi{10.2307/421090}

\bibitem[{Piecha(2016)}]{Piecha2016completeness}
Piecha T (2016) {Completeness in Proof-theoretic Semantics}. In: Piecha T,
  Schroeder-Heister P (eds) {Advances in Proof-theoretic Semantics}. Springer,
  p 231--251

\bibitem[{Piecha and Schroeder{-}Heister(2016)}]{Schroeder2016atomic}
Piecha T, Schroeder{-}Heister P (2016) {Atomic Systems in Proof-Theoretic
  Semantics: Two Approaches}. In: \'{A}ngel Nepomuceno~Fern\'{a}ndez, Martins
  OP, Redmond J (eds) {Epistemology, Knowledge and the Impact of Interaction}.
  Springer Verlag, \doi{10.1007/978-3-319-26506-3_2}

\bibitem[{Piecha and Schroeder-Heister(2017)}]{Piecha2017definitional}
Piecha T, Schroeder-Heister P (2017) {The Definitional View of Atomic Systems
  in Proof-theoretic Semantics}. In: Arazim P, L\'a{}vi\v{c}ka T (eds) {The
  Logica Yearbook 2016}. College Publications London, p 185--200

\bibitem[{Piecha and Schroeder-Heister(2019)}]{Piecha2019incompleteness}
Piecha T, Schroeder-Heister P (2019) {Incompleteness of Intuitionistic
  Propositional Logic with Respect to Proof-theoretic Semantics}. Studia Logica
  107(1):233--246. \doi{10.1007/s11225-018-9823-7}

\bibitem[{Piecha et~al(2015)Piecha, de~Campos~Sanz, and
  Schroeder-Heister}]{Piecha2015failure}
Piecha T, de~Campos~Sanz W, Schroeder-Heister P (2015) {Failure of Completeness
  in Proof-theoretic Semantics}. Journal of Philosophical Logic 44(3):321--335.
  \doi{10.1007/s10992-014-9322-x}

\bibitem[{Prawitz(1971)}]{Prawitz1971ideas}
Prawitz D (1971) {Ideas and Results in Proof Theory}. In: Fenstad J (ed)
  {Studies in Logic and the Foundations of Mathematics}, vol~63. Elsevier, p
  235--307, \doi{10.1016/S0049-237X(08)70849-8}

\bibitem[{Prawitz({2006 [1965]})}]{Prawitz2006natural}
Prawitz D ({2006 [1965]}) {Natural Deduction: A Proof-theoretical Study}.
  Courier Dover Publications, \doi{10.2307/2271676}

\bibitem[{Pym(2019)}]{pym2019resource}
Pym DJ (2019) {Resource Semantics: Logic as a Modelling Technology}. {ACM}
  {SIGLOG} News 6(2):5--41. \doi{10.1145/3326938.3326940}

\bibitem[{Pym et~al(2019)Pym, Spring, and O’Hearn}]{pym2019separation}
Pym DJ, Spring JM, O’Hearn P (2019) {Why Separation Logic Works}. Philosophy
  \& Technology 32:483--516. \doi{10.1007/s13347-018-0312-8}

\bibitem[{Pym et~al(2022)Pym, Ritter, and Robinson}]{Pym2022catpts}
Pym DJ, Ritter E, Robinson E (2022) {Proof-theoretic Semantics in Sheaves
  (Extended Abstract)}. In: Hirvonen {\AA}, Vel\'azquez-Quesada F (eds)
  {Proceedings of the Eleventh Scandinavian Logic Symposium --- SLSS 11}.
  Scandinvavian Logic Society

\bibitem[{Read(1988)}]{Read1988}
Read S (1988) {Relevant Logic}. Basil Blackwell, \doi{10.2307/2219818}

\bibitem[{Reynolds(2002)}]{reynolds2002separation}
Reynolds JC (2002) {Separation Logic: A Logic for Shared Mutable Data
  Structures}. In: Plotkin G (ed) Logic in Computer Science --- LICS. Springer,
  p 55--74, \doi{10.1109/LICS.2002.1029817}

\bibitem[{Sandqvist(2005)}]{Sandqvist2005inferentialist}
Sandqvist T (2005) {An Inferentialist Interpretation of Classical Logic}. PhD
  thesis, Uppsala University

\bibitem[{Sandqvist(2009)}]{Sandqvist2009CL}
Sandqvist T (2009) {Classical Logic without Bivalence}. Analysis 69(2):211--218

\bibitem[{Sandqvist(2015{\natexlab{a}})}]{Sandqvist2015base}
Sandqvist T (2015{\natexlab{a}}) {Base-extension Semantics for Intuitionistic
  Sentential Logic}. Logic Journal of the IGPL 23(5):719--731.
  \doi{10.1093/jigpal/jzv021}

\bibitem[{Sandqvist(2015{\natexlab{b}})}]{Sandqvist2015hypothesis}
Sandqvist T (2015{\natexlab{b}}) {Hypothesis-discharging Rules in Atomic
  Bases}. In: Wansing H (ed) {Dag Prawitz on Proofs and Meaning}. Springer, p
  313--328

\bibitem[{Sandqvist(2022)}]{sandqvistwld}
Sandqvist T (2022) {Atomic Bases and the Validity of Peirce's Law}.
  \url{https://sites.google.com/view/wdl-ucl2022/schedule\#h.ttn75i73elfw},
  {World Logic Day --- University College London (Accessed June 2023)}

\bibitem[{Schroeder-Heister(2006)}]{Schroeder2006validity}
Schroeder-Heister P (2006) {Validity Concepts in Proof-theoretic Semantics}.
  Synthese 148(3):525--571. \doi{10.1007/s11229-004- 6296-1}

\bibitem[{Schroeder-Heister(2008)}]{Schroeder2007modelvsproof}
Schroeder-Heister P (2008) {Proof-Theoretic versus Model-Theoretic
  Consequence}. In: Pelis M (ed) {The Logica Yearbook 2007}. Filosofia,
  \doi{95b360ffe9ad174fd305539813800ea23fec33de}

\bibitem[{Schroeder-Heister(2018)}]{SEP-PtS}
Schroeder-Heister P (2018) {Proof-Theoretic Semantics}. In: Zalta EN (ed) {The
  Stanford Encyclopedia of Philosophy}, {S}pring 2018 edn. Metaphysics Research
  Lab, Stanford University

\bibitem[{Stafford(2021)}]{Stafford2021}
Stafford W (2021) {Proof-theoretic Semantics and Inquisitive Logic}. Journal of
  Philosophical Logic

\bibitem[{Stafford and Nascimento(forthcoming)}]{Stafford2023}
Stafford W, Nascimento V (forthcoming) {Following All the Rules: Intuitionistic
  Completeness for Generalized Proof-Theoretic Validity}. Analysis
  \doi{10.1093/analys/anac100}

\bibitem[{Szabo(1969)}]{Gentzen}
Szabo ME (ed)  (1969) {The Collected Papers of Gerhard Gentzen}. North-Holland
  Publishing Company, \doi{10.2307/2272429}

\bibitem[{Tarski(1936)}]{Tarski1936}
Tarski A (1936) O poj\c{e}ciu wynikania logicznego. Przegl\c{a}d Filozoficzny
  39

\bibitem[{Tarski(2002)}]{Tarski2002}
Tarski A (2002) On the concept of following logically. History and Philosophy
  of Logic 23(3):155--196. \doi{10.1080/0144534021000036683}

\bibitem[{Troelstra and Schwichtenberg(2000)}]{troelstra2000basic}
Troelstra AS, Schwichtenberg H (2000) {Basic Proof Theory}. Cambridge
  University Press, \doi{10.1017/CBO9781139168717}

\end{thebibliography}

\end{document}